\documentclass[letterpaper,11pt]{article}
\usepackage[hyphens]{url}
\usepackage{xurl}
\usepackage[letterpaper,margin=1in]{geometry}
\usepackage{graphicx}

\usepackage{times}
\let\oldparagraph\paragraph
\renewcommand{\paragraph}[1]{\oldparagraph{#1.}}

\usepackage[font=small,labelfont=bf]{caption}

\usepackage{mathtools}
\usepackage{amsthm}
\usepackage{amsmath}
\usepackage{amsfonts}
\usepackage{amssymb}

\usepackage{algorithm}
\usepackage{algorithmic}

\usepackage{authblk}
\usepackage{preamble}
\usepackage[
    backend=biber,
    style=alphabetic,
    natbib=true,
    maxbibnames=99,
    maxalphanames=4,
    minalphanames=3
]{biblatex}
\AtBeginBibliography{\sloppy
    \setlength{\emergencystretch}{3em}\Urlmuskip=0mu plus 2mu\relax
}

\title{How Many Votes Is a Lie Worth? \\ Measuring Strategyproofness through Resource Augmentation}
\author[1,2]{Ratip Emin Berker}
\author[1,2]{Vincent Conitzer}
\author[3]{Eden Hartman}
\author[1,2]{\\Jiayuan Liu}
\author[1,2]{Caspar Oesterheld}

\affil[1]{Carnegie Mellon University}
\affil[2]{Foundations of Cooperative AI Lab (FOCAL)}
\affil[3]{Bar-Ilan University}
\affil[ ]{\texttt{\{rberker, conitzer, jiayuan4, coesterh\}@cs.cmu.edu, eden.r.hartman@gmail.com}}

\begin{document}

\pagenumbering{roman}
\maketitle

\begin{abstract}
    It is well known, by the Gibbard–Satterthwaite Theorem, that when there are more than two candidates, \emph{any} non-dictatorial voting rule can be manipulated by untruthful voters. 
    But \emph{how much stronger} is an untruthful voter than a truthful one under different voting rules? We suggest measuring the potential comparative advantage of a strategic voter by asking how many copies of their (truthful) vote must be added to the election in order to achieve an outcome as good as their best manipulation. Intuitively, this definition quantifies what a voter can gain by manipulating in comparison to what they would have gained by finding like-minded voters to join the election (equivalently, by increasing their voter weight). The higher the former is, the more incentive a voter will have to manipulate, even when it is computationally costly. Our measure is prior-free and requires no assumptions on the voters' cardinal utilities.

    Using this framework, we obtain a principled method to measure and compare the \emph{manipulation leverage} for different voting rules. We analyze and report this leverage for well-known and broad classes of social choice functions. In particular, we show that while the manipulation leverage of many rules grows linearly with the number of voters, Borda Count's is independent of this number.  Instead, its manipulation leverage is $m-2$, where $m$ is the number of candidates, which can be much lower in many practical settings.
    Building on this result, we prove that the positional scoring rule with the smallest manipulation leverage will always be either Borda Count (if the number of voters outweighs twice the number of candidates) or Plurality (otherwise). Further, we prove that any rule satisfying a weak form of majority consistency (and thus any Condorcet extension) cannot outperform Plurality,
    while any majoritarian rule will perform significantly worse. Consequently, Borda Count has much lower manipulation leverage than these wide classes of rules under large electorates.
    
    Motivated by our results on well-known (classes of) rules, we investigate whether \emph{any} reasonable rule can have lower manipulation leverage than Borda Count. We answer this question in the negative by establishing a quantitative extension of the Gibbard--Satterthwaite theorem, proving \emph{no anonymous and onto voting rule} can have manipulation leverage that is lower than $\ncand-2$ across all electorate sizes. In other words, any onto rule treating all voters equally will induce incentives for misreporting preferences, \emph{even if} the voter is offered $k$ additional truthful votes as an alternative to manipulating for some $k \geq \ncand-3$. As a consequence, we prove that Borda Count is indeed optimal among all such rules that satisfy these two minimal assumptions.
    
    Overall, our work establishes a clear and prior-free separation between different rules in terms of manipulability, paving the way for many interesting directions for understanding the comparative advantage of a strategic voter with limited resources.
\end{abstract}
\clearpage
\tableofcontents
\clearpage
\pagenumbering{arabic}
\begin{center}
\emph{All voting rules are manipulable. But some rules are more manipulable than others.}
\end{center}

\section{Introduction}
Alice has found herself in a conundrum. It is election day in Wonderland, and she firmly wants to help her favorite candidate, the Mad Hatter, become the next mayor. A well-versed student in social choice theory, Alice knows that any reasonable voting rule is manipulable, meaning voting untruthfully may help reach her goal~\citep{gibbard1973manipulation, satterthwaite1975strategy}. So how should she vote?

For Alice, answering this question will not be easy. She is aware that voter manipulation is $\NP$-hard for many voting rules~\citep{bartholdi1989computational,bartholdi1991single,faliszewski2010ai}, and she worries that she might not find her optimal vote before the election day is over. Further, Alice is unsure about how the other residents of Wonderland will vote, and fears that a misestimate may cause her manipulation to do her more harm than good. Indeed, for many voting rules, Alice cannot theoretically guarantee lying will be \emph{safe} unless she knows a significant portion of other votes~\citep{hartman2025s}, and she is further troubled by empirical results demonstrating the perils of manipulating under uncertainty~\citep{eckel1989strategic,kube2010when,tyszler2016information}. Even if Alice is generally aware of the preferences of her compatriots, she realizes there may be other forms of uncertainty, such as rapidly shifting preferences or not knowing who will actually show up to the election~\citep{meir2018simultaneous}. Worst of all, she is scared other voters might be strategizing themselves! Consequently, the intel Alice has on these voters may also have been cleverly handpicked by them to avoid being exploited~\citep{hartline2021mechanism}.

In the midst of her despair, Alice remembers she has a secret weapon: a car. Many of her neighbors fully agree with her (true) preferences over the mayoral candidates, but they will not interrupt their tea time to go and vote---unless someone drives them. Filled with excitement, Alice decides to abandon any efforts to compute her optimal untruthful vote and instead drive her honest neighbors to the polls.\footnote{For a real-life comparison, consider transportation services by various get-out-the-vote (GOTV) efforts.} After all, this alternative plan requires no computation, no information on the rest of the electorate, and not even information on the voting rule being used.\footnote{Although Alice likely knows that in Wonderland, election winners are determined using Dodgson's rule~\citep{dodgson1876method}.} However, just as she is about to reach for her car keys, an anxious question strikes Alice: \emph{just how many neighbors will she have to drive, before she can be certain that abandoning her efforts to manipulate was worth it?}

\subsection{How Many Votes Is a Lie Worth?}\label{sec:how-many}

In this paper, we resolve Alice's conundrum. More precisely, we investigate the number of additional (truthful) copies a voter would need in order to produce a result as desirable to her as her best manipulation. For most natural voting rules, getting sufficiently many copies should eventually reach this goal, since these rules favor overwhelming majorities. However, the exact number of copies needed can be quite different for different rules, as the following example shows.

\begin{example}\label{ex:intro}
    Consider two voting rules that map voters' rankings over $\ncand$ candidates to a winner. In \emph{Borda Count} $(\borda)$, each voter awards $\ncand-j$ points to their $j^\text{th}$-ranked candidate, and the candidate with the most points wins. \emph{Instant Runoff} voting ($\irun$), on the other hand, iteratively eliminates the candidate with the fewest voters ranking them top (out of the remaining candidates), and the last remaining candidate wins. Given 4 candidates and $\nvote \geq 9$ voters, say the voters' preferences over candidates are as follows:
    \medskip
    \begin{center}
        \begin{tabular}{c|l}
             \cellcola Voter 1 & \cellcola $\cand{1} \vote{} \cand{2} \vote{} \cand{3} \vote{} \cand{4}$\\
             \cellcolb Voter 2 &  \cellcolb $\cand{4} \vote{} \cand{1} \vote{} \cand{2} \vote{} \cand{3}$\\
             \cellcola  $\frac{\nvote-1}{4}$ voters &  \cellcola  $\cand{2} \vote{} \cand{3} \vote{} \cand{4} \vote{} \cand{1}$\\
             \cellcolb $\frac{\nvote-1}{4}$ voters & \cellcolb $\cand{3} \vote{} \cand{4} \vote{} \cand{1} \vote{} \cand{2}$\\
              \cellcola $\frac{\nvote-3}{2}$ voters & \cellcola $\cand{4} \vote{} \cand{3} \vote{} \cand{2} \vote{} \cand{1}$\\
        \end{tabular}
    \end{center}
    \medskip
    Using Instant Runoff, $\cand{1}$ is eliminated first, followed by $\cand{3}$, and finally $\cand{2}$, making $\cand{4}$ the winner. Similarly, under Borda Count, the scores of candidates $\cand{1}, \cand{2}, \cand{3}$, and $\cand{4}$ are $\frac{\nvote+19}{4}, \frac{5 \nvote+3}{4}, \frac{9\nvote-13}{4}$, and $\frac{9\nvote-9}{4}$, respectively, so the winner is once again $\cand{4}$. Suppose voter 1 misreports her preferences as ${\color{purple} \cand{3}} \manip{} \cand{1} \manip{} \cand{2} \manip{} \cand{4}$ instead. Now, the winner under either rule becomes $\cand{3}$, whom voter 1 indeed prefers to $\cand{4}$. However, while adding just two more copies of voter 1's truthful vote to the original election would have also made $\cand{3}$ the Borda Count winner, this is not the case for Instant Runoff: with less than $\frac{\nvote-5}{4}$ additional copies of voter 1's truthful vote, the $\irun$ winner is still $\cand{4}$. Even with more truthful copies, $\cand{4}$ may still win, depending on how we break ties.\footnote{See \Cref{sec:manipot} on how we deal with ties.} Voter 1 will need $\nvote-1$ copies of her truthful vote (in which case they become a majority, and $\cand{1}$ wins) to get an outcome that she undoubtedly (\ie, regardless of the tiebreaker) prefers to misreporting.
\end{example}

As \Cref{ex:intro} demonstrates, the same untruthful vote can outweigh different numbers of truthful ones under different voting rules, by an arbitrarily large margin. To analyze this ``exchange rate'' between lies and truthful votes for different rules, we define the \emph{\textbf{manipulation leverage}} of a voting rule as the number of truthful copies needed (in the worst case) to ensure a voter could not have reached a better outcome by misreporting (see \Cref{sec:manipot} for a formal definition). We attach three interpretations to this metric:

\begin{enumerate}[leftmargin=0.5cm,topsep=2pt,itemsep=1pt]
    \item[1)] \emph{Get out the vote (Alice's conundrum):} If the electorate consists of communities with identical preferences, any member can dedicate their resources to increase voter participation from their community. 
    Equivalently, they may increase their \emph{voter weight} in exchange for resources (\eg, money or effort) in other settings\footnote{\Cref{sec:game-theory} formalizes this game-theoretic model, where voters choose between manipulating and increasing their voter weight.}, such as corporate elections (by buying more shares) or proof-of-stake-based blockchain governance~\citep{boehmer2024approval,burdges2020ovefview,cevallos2021verifably}.
    The manipulation leverage of a rule is the \emph{minimum} opportunity cost of misreporting (in terms of the cost of increasing voter weight by one) sufficient to disincentivize manipulation. If the manipulation cost (\eg, computational, or for gathering information) is larger, voters are better off investing their resources in raising community turnout or increasing their voter weight.
    \item[2)] \emph{Endogenous participation:} Taking a dual approach to (1), we can think of voters representing communities, with their weight corresponding to the number of members. The community leader(s) can ask the members to vote strategically, but this may come at the cost of losing a fraction of the community who is unwilling to misreport their preferences (\eg, Australian parties being ``punished'' by their voter base after issuing tactical how-to-vote cards in Instant Runoff elections~\citep{raue2009instant,abc2017WA}). The manipulation leverage is the minimum reduction in voter weight to guarantee manipulating does more harm than good.
    \item[3)] \emph{Approximate strategyproofness:} As no non-dictatorial rule can achieve strategyproofness~\citep{gibbard1973manipulation, satterthwaite1975strategy}, a long line of work has investigated how close any rule can get to it (see \Cref{sec:manip-lit} for an overview). Establishing approximations based on cardinal utilities (\eg, ``being truthful guarantees $\alpha$ fraction of the utility from best manipulation'') is difficult since we only have access to ordinal preferences. For instance, in \Cref{ex:intro}, the improvement in voter 1's utility (after changing the winner from $\cand{4}$ to $\cand{3}$) is tremendous if she is almost indifferent between $\{\cand{1}, \cand{2}, \cand{3}\}$ but really dislikes $\cand{4}$. However, it may also be insignificant if she is almost indifferent between $\{\cand{3}, \cand{4}\}$ instead. Rather than a utility-based analysis, we thus turn to a \emph{resource augmentation} approach, where the outcome of the ``best-possible'' solution (in this case, the optimal manipulation) is compared with another simple method equipped with additional resources (in this case, truth-telling with ``bonus'' votes)~\citep{roughgarden2020beyond}. Thus, a rule's manipulation leverage captures how approximately strategyproof it is. Resource augmentation and similar bicriteria approaches are frequently employed in mechanism design, from where we draw our inspiration (see \Cref{sec:bulow-klemperer}).
\end{enumerate}

Overall, the manipulation leverage framework offers a principled method of measuring the (worst-case) advantage of a strategic voter over her (augmented) truthful self. By establishing a clear and prior-free separation between different rules in terms of manipulability, our work paves the way for the search for rules that provide voters with minimal advantage from manipulating.

\begin{center}
    \begin{table}[t]
    \centering

\begin{NiceTabular}{|c|c||c|c|c}[cell-space-limits = 2pt]
    \cline{1-4}
      \cellcolor{\tabcold}\textbf{Class of rules}  & \cellcolor{\tabcold} \textbf{Rule}  & \cellcolor{\tabcold} \textbf{Lower} & \cellcolor{\tabcold} \textbf{Upper}& \\
    \hhline{|=|=||=|=|~ }
      \Block[fill=\tabcolc]{4-1}{Positional\\Scoring Rules,\\$(\score_1,\score_2,\ldots ,\score_\ncand$)}  & \cellcolor{\tabcolc} Borda Count ($\borda$) 
     & \Block[fill=\tabcolc]{1-2}{$m-2$ (\Cref{thm:borda})} 
     & & \\
    \cline{2-4}
      & \cellcolor{\tabcolc} Any rule with $s_1 = s_2$ & \Block[fill=\tabcolc]{1-2}{$\infty$ (\Cref{thm:pos-infinity})} & &\\
    \cline{2-4}
    
      & \cellcolor{\tabcolc} Any rule with $s_1\neq s_2$ &\cellcolor{\tabcolc} $\min\{\borda, \plurality\}$ (\cref{thm:pos-best}) & \cellcolor{\tabcolc} See \cref{thm:pos-upper} &\\
    \cline{2-5}
     
      & \cellcolor{\tabcola} Plurality (\plurality)  
     & \Block[fill=\tabcola]{1-2}{$\lceil \frac{n-1}{2} \rceil$ (\Cref{thm:plur})} 
     &  &\Block[fill=\tabcola,draw]{6-1}{\rotatebox{-90}{Majority-consistent}} \\
    \Hline\Hline
    
      \Block[fill=\tabcola]{2-1}{Condorcet\\Extensions} 
     & \cellcolor{\tabcola} Black's Rule $(\black)$ & \Block[fill=\tabcola]{1-2}{$\nvote-1$ (\Cref{thm:black}) }& \\
    \cline{2-4}
      & \cellcolor{\tabcola} Maximin  $(\maximin)$ & \Block[fill=\tabcola]{1-2}{$\approx \frac{(\ncand-2)\nvote}{\ncand-1}$ (\Cref{thm:maximin})} &  &\\
    \Hline\Hline
    
     \Block[fill=\tabcolb]{7-1}{Other rules\\and rule classes}
     & \cellcolor{\tabcola} Instant Runoff ($\irun$)  & \Block[fill=\tabcola]{2-2}{$\nvote-1$ (\Cref{thm:runoff})}&& \\
    \cline{2-2}
      &  \cellcolor{\tabcola}  Plurality with Runoff $(\plurWrunoff)$ & & \\
    \cline{2-4}
      & \cellcolor{\tabcola} \textbf{Any} majority-consistent rule
     & \cellcolor{\tabcola} $\frac{\nvote-1}{2}$ (\Cref{thm:biranking}) & \cellcolor{\tabcola} \makecell{$\nvote-1$ (\Cref{remark:majority-consistent})} &\\
    \cline{2-5}
      & \cellcolor{\tabcolb} \textbf{Any} majoritarian rule
     & \cellcolor{\tabcolb} $\nvote-2$ (\Cref{thm:majoritarian}) & \cellcolor{\tabcolb} --- &\\
    \cline{2-4}
      & \cellcolor{\tabcolb} \textbf{Any} runoff rule
     & \cellcolor{\tabcolb} \makecell{$\frac{\nvote-1}{2}$ (\Cref{prop:runoff-biranking})} & \cellcolor{\tabcolb} --- &\\
    \cline{2-4}
     & \cellcolor{\tabcolb} Pareto $(\pareto)$ & \Block[fill=\tabcolb]{2-2}{$\infty$ (\Cref{prop:paromni}) }& &\\
    \cline{2-2}
      &  \cellcolor{\tabcolb}  Omninomination $(\omni)$ & && \\
    \cline{1-4}

      \Block[fill=Navy!80!black,name=universal-bound]{1-2}{\textbf{\color{yellow}Any anonymous, onto voting rule}}
     & & \cellcolor{Navy!80!black} {\color{white}\hypersetup{linkcolor=white}$\ncand-2$ (\Cref{thm:universal})} & \cellcolor{Navy!80!black} {\color{white}---} &\\
    \cline{1-4}
    
\CodeAfter
\begin{tikzpicture}
    \draw[white,line width=1.1\arrayrulewidth,line cap=butt,
          shorten >=.6\arrayrulewidth,shorten <=.6\arrayrulewidth]
        (universal-bound.north east) -- (universal-bound.south east)
        ($(universal-bound.north east)+(-\doublerulesep-\arrayrulewidth,0pt)$) --
        ($(universal-bound.south east)+(-\doublerulesep-\arrayrulewidth,0pt)$)
        (col-4 |- universal-bound.north east) --
        (col-4 |- universal-bound.south east);
\end{tikzpicture}
\end{NiceTabular}

    \caption{Summary of results showing bounds on the manipulation leverage of the rules and rule classes studied in this paper. For lower bounds, see the theorem/proposition statement for the assumptions on $\nvote$ \& $\ncand$. The final row assumes an upper bound on the manipulation leverage that holds uniformly over all electorate sizes $\nvote$.}
    \label{tab:rules-bounds}
    \end{table}
\end{center}

\subsection{Our Contributions } \label{sec:contributions}

\paragraph{Manipulation Leverage} We introduce a \emph{prior-free} metric for measuring the comparative advantage of a manipulator over a truthful voter. Our approach is based on a property of voting rules we call \emph{$k$-augmentation strategyproofness} (\Cref{def:kmanip}), which requires that, in any election with $\nvote$ voters, a voter will always weakly prefer adding $k$ copies of her truthful self to any possible manipulation. We then define the \emph{manipulation leverage} $\spdeg{\scf}$ of a voting rule $\scf$ as the minimum $k$ for which it is $k'$-augmentation strategyproof on $\nvote$-voter elections for every $k'\geq k$ (\Cref{def:manipot}). The manipulation leverage is a major departure from existing manipulability measures, which often require assumptions on how voters' preferences are distributed or on the utility function behind their ordinal preferences (see \Cref{sec:manip-lit}).

With our toolkit complete, \Cref{sec:rules} initiates our main agenda of finding the manipulation leverage of well-studied (classes of) rules. While a single election instance suffices to lower-bound the manipulation leverage of a specific rule, upper bounds must account for all instances, which, as we will see, will require unique techniques for each rule (class). As a warm-up, we show that the manipulation leverage of each of three rules employed in practice (Plurality, Plurality with Runoff, and Instant Runoff) is $\Theta(\nvote)$, where $\nvote$ is the number of voters (\Cref{thm:plur,thm:runoff}). Our result shows that under these rules, a single manipulation can (in the worst case) be as powerful as introducing enough copies to become a large fraction (for the latter two, a majority) of the voters, which is troublesome for sizeable electorates. To counter this, we show that the manipulation leverage of Borda Count is \emph{independent} of the number of voters (\Cref{thm:borda}). Instead, it is precisely $\ncand-2$, where $\ncand$ is the number of candidates, offering a significant advantage when $\nvote \gg \ncand$.\footnote{Indeed, we recall that in \Cref{ex:intro}, exactly $\ncand-2=2$ copies were enough to outweigh manipulation under Borda Count.}

\paragraph{Positional Scoring Rules} Having seen that the manipulation leverage of different rules can be arbitrarily far from one another, we turn to general rule classes. In particular, we look at \emph{positional scoring rules}, an infinite-sized class that includes both Plurality and Borda Count. We show that the manipulation leverage of \emph{any} positional scoring rule can be decomposed into the maximum over $\ncand$ terms (one for each position in the ballot), each of which is a minimum over one $\nvote$-dependent and one $\nvote$-independent term (\Cref{thm:pos-lower,thm:pos-upper}). Iteratively applying this bound gives us our main characterization result for this class.

\begin{theorem}[Informal statement of \Cref{thm:pos-best}]
    For any value of $\nvote$ and $\ncand$, the positional scoring rule with the lowest manipulation leverage is either Borda Count (if $ \ncand-2< \frac{\nvote-1}{2}$) or Plurality (otherwise).
\end{theorem}

Our result indicates that these two rules alone make up the entire Pareto frontier of positional scoring rules in terms of (lowest) manipulation leverage. Indeed, the manipulation leverage of any other positional scoring rule is strictly worse and can even be unbounded (\Cref{thm:pos-infinity}), indicating that under some rules, manipulating can be more powerful than becoming an arbitrarily large fraction of the electorate.

\paragraph{Majoritarian and (Biranking) Majority-Consistent Rules} Having established that no positional scoring rule can outperform Borda Count (in the $\nvote \gg \ncand$ regime), we turn to other rule classes to investigate how much this result generalizes. A ubiquitous notion in social choice theory---dating back to the 18$^{\text{th}}$ century~\citep{CaritatMarquisdeCondorcet85:Essai}---is that of a \emph{Condorcet winner}: a candidate who defeats everyone else in a pairwise matchup. A voting rule is said to be \emph{Condorcet-consistent} (or a \emph{Condorcet extension}) if it picks the Condorcet winner whenever one exists. Such rules have a trivial upper bound of $\nvote-1$ on their manipulation leverage, as with that many copies, the voter's top choice becomes a Condorcet winner.
We first consider \emph{majoritarian} rules, whose outcome is solely determined by pairwise defeats. For any majoritarian Condorcet extension satisfying neutrality (all candidates treated equally), we obtain a lower bound that nearly matches the trivial upper bound. In fact, our result applies to all majoritarian rules, without requiring Condorcet consistency.

\begin{theorem}[Informal statement of \Cref{prop:majoritarian-nonfaithful,thm:majoritarian}]
    The manipulation leverage of any (neutral) majoritarian rule is at least $\nvote-2$.
\end{theorem}
We then turn to Condorcet extensions without the majoritarian restriction and show that no such rule can achieve manipulation leverage better than $\Omega(\nvote)$. {In fact, we prove a stronger result by significantly weakening} the Condorcet consistency axiom to what we call \emph{biranking majority consistency}:
whenever an election only contains two voter types (\ie, each voter submits one of two specific rankings), a biranking-majority-consistent rule picks the top choice of the more-frequent ranking. This is a much wider class than Condorcet extensions; in particular, we show that they include all \emph{runoff rules}---another infinite-sized class---which generalize Instant Runoff.

\begin{theorem}[Informal statement of \Cref{thm:biranking}]
    The manipulation leverage of any (anonymous) rule satisfying biranking majority consistency is at least $\frac{\nvote-1}{2}$.
\end{theorem}

We remark that Borda Count fails biranking majority consistency,\footnote{Consider a slight majority that ranks the minority's top choice second, but the minority ranks the majority's top choice last.} allowing it to achieve much lower manipulation leverage in the $\nvote \gg \ncand$ regime, in particular, leverage that does not grow with $\nvote$.

\paragraph{Game-Theoretic Benefits of $\nvote$-Independent \texttt{MLev}}
In \Cref{sec:game-theory}, we take an intermission from our analysis of the manipulation leverage of specific rules and rule classes to show that having manipulation leverage independent of $\nvote$ has advantages that extend beyond large electorates. Whereas the other sections focus on a single voter's choice between manipulating and adding $k$ truthful copies, here
we consider a game in which all voters face this choice simultaneously, formalizing interpretation (1) from \Cref{sec:how-many}. If the rule has an $\nvote$-independent upper bound on its manipulation leverage, setting $k$ to this bound renders truthful voting a weakly dominant strategy (\Cref{prop:game-n-indep}).
Without such a bound, on the other hand, we show that all voters voting truthfully 
may not be a Nash equilibrium \emph{for any} $k$, even with just three voters (\Cref{prop:game-n-dep}). As such, $\nvote$-independent manipulation leverage can be critical for the stability of truthful augmentation.

\paragraph{A Universal Lower Bound}
Our main impossibility result goes beyond the classes of voting rules studied above. Motivated by the preceding section,  \Cref{sec:universal} asks how small $\nvote$-independent manipulation leverage can get. Our answer comes in the form of  a quantitative extension of the Gibbard--Satterthwaite theorem \textbf{for all anonymous (voters are treated equally) and onto (any candidate wins in some election) voting rules}. While the classical theorem rules out strategyproofness, it leaves open how close a rule can come to it under our resource-augmentation measure. We show that no anonymous and onto rule can guarantee manipulation leverage below $\ncand-2$ across all electorate sizes. Thus, the number of truthful copies needed to outweigh a manipulation must, in the worst case, grow linearly with the number of candidates. This lower bound exactly matches Borda Count's manipulation leverage, establishing its optimality among all anonymous and onto voting rules. As such, neither choosing an existing rule outside the classes we have analyzed nor designing an entirely new one can improve on Borda Count under these two mild assumptions.

\begin{theorem}[Informal statement of \Cref{thm:universal}]
    Let $\scf$ be an anonymous and onto voting rule. If a finite bound $K$ satisfies $\spdeg{\scf}\leq K$ for every electorate size $\nvote$, then $K\geq\ncand-2$.
\end{theorem}

On a high level, the proof of \Cref{thm:universal} first ``stretches'' $\scf$ into a rule that is defined on rational profiles and inherits a weighted form of its manipulation-leverage bound, while also exhibiting additional properties like group faithfulness and conditional participation. By studying families of cyclic profiles and applying Sperner's lemma, we then construct a Lipschitz family of \emph{common junctions}: points approached arbitrarily closely by rational profiles electing each candidate. The local geometry of this family translates the manipulation-leverage bound on $\scf$ into linear constraints on pairwise comparison vectors lying in a common $(\ncand-1)$-dimensional space. Finally, the linear dependence among adjacent comparison vectors around a cycle of $\ncand$ candidates forces the manipulation leverage of $\scf$ to be  at least $\ncand-2$.

\paragraph{Key Takeaways} \Cref{tab:rules-bounds} gathers our main results. Overall, key takeaways from our work are as follows:  
\begin{itemize}[leftmargin=0.5cm,topsep=2pt,itemsep=1pt]
    \item Manipulation leverage provides a prior-free ``exchange rate'' between strategic voting and additional truthful support: a rule with low manipulation leverage leaves little worst-case advantage to lying over finding like-minded voters or increasing voter weight through other means. The measure provides meaningful separations between voting rules and rule classes in terms of manipulability.
    \item Borda Count stands out among well-studied rules by having manipulation leverage independent of the number of voters, which ensures the game-theoretic stability of choosing truthful augmentation over manipulation. Under large electorates, this rule outperforms all positional scoring rules, majoritarian rules, runoff rules, and (biranking) majority-consistent rules. Most importantly, it is optimal among all rules with bounded manipulation leverage over all electorate sizes. This demonstrates that under Borda Count, voters may have weak incentives for manipulation in settings where information on other voters is scarce and resources can instead be used to find more like-minded voters (or increase voter weight).
    \item There is a fundamental tradeoff between rules with low manipulation leverage and rules that unequivocally follow majorities. Rules for which being the top choice of an (even slight) majority guarantees victory will leave room for manipulations that are more powerful than becoming a significant fraction of the electorate. Designing rules that curb the comparative advantage of a manipulator requires making use of information from the ballots of the entire electorate, not just of the majority.
\end{itemize}

Our results and framework pave the way for many interesting directions for understanding the comparative advantage of manipulation under limited resources, some of which we discuss in \Cref{sec:future}.

\subsection{Related Work}\label{sec:related}

In this section, we first discuss the resource augmentation approach, in which our notion of manipulation leverage is rooted. We then discuss how the manipulation leverage differs from prior manipulability metrics.

\subsubsection{Resource Augmentation}\label{sec:bulow-klemperer} 
The definition of manipulation leverage is inspired by a classical result in auction theory by \citet{bulow1994auctions}. Take a single-item auction where the valuations are drawn i.i.d.\ from some regular distribution. Here, a mechanism designer trying to maximize (expected) revenue faces challenges analogous to a strategic voter: the optimal auction may not be simple and requires knowing the distribution of bidders' valuations. Further, the correct choice under one distribution can significantly hurt revenue under another~\citep{hartline2021mechanism}. Luckily for the designer, as \citeauthor{bulow1994auctions} show, running a second-price auction with $\nvote+1$ bidders will generate at least as much revenue as the revenue-maximizing auction with $\nvote$ bidders. Much like a voter reporting her truthful preferences, running a second-price auction is simple and requires no knowledge of preferences; in other words, it is \emph{prior-free}. Intuitively, the theorem shows that the auctioneer is better off finding an additional bidder than trying to learn the bidders' preferences and compute the optimal auction.

There is a large body of work that extends this result to various auction settings by asking \emph{how many} additional bidders are needed so that running the ``simplest'' auction (\eg, VCG) yields expected revenue at least as high as the optimal auction---\eg, \citep{Alon2017,Brustle2022Competition, brustle2024competition, ezra2025competition, 
ezra2024competition, liu2018competition, derakhshan2024settling,beyhaghi2019optimal}.
This measure is referred to as the \emph{competition complexity}. This same approach of comparing a simple algorithm to the optimal solution with fewer resources has been applied in other domains under the name \emph{resource augmentation} (\cf.\  \citep{roughgarden2020beyond} for an overview), including allocation problems~\citep{akrami2025fair, budish2011combinatorial}, selfish routing~\citep{roughgarden2002bad,friedman2004genericity}, scheduling~\citep{kalyanasundaram2000speed,anand2012resource}, and online paging~\citep{sleator1985amortized}.

Despite the parallels, our model also exhibits several key differences from the above work. For instance, in our setting, the manipulator is no longer the mechanism designer, but one of the \emph{participants} in the mechanism. Consequently, while the resource augmentation approach in auction theory typically corresponds to adding bidders (who do not share the designer's incentives, but still benefit the designer by participating), in our case the resource being augmented is the weight of the voter herself, via truthful copies.

\subsubsection{Degree(s) of Manipulability}\label{sec:manip-lit}

Strategic voting has long been one of the cornerstones of the (computational) social choice literature. Indeed, following the seminal result by \citet{gibbard1973manipulation} and \citet{satterthwaite1975strategy} showing that any non-dictatorial and onto voting rule can be manipulated by voters (when $\ncand>2$), a rich line of work has focused on measuring the manipulability of individual rules. Broadly speaking, we can group these measures under those that study (1) the \textbf{cost} of, (2) the \textbf{gains} from, and (3) the \textbf{frequency} of manipulations.

\paragraph{Manipulation Cost} 
An extensive line of research has attempted to counter the pessimism of the Gibbard-Satterthwaite impossibility theorem by demonstrating that manipulation can have high \emph{costs} under common voting rules, including computational costs \citep{bartholdi1989computational, bartholdi1991single,conitzer2007elections,conitzer2002complexity,faliszewski2008copeland,obraztsova2011ties,faliszewski2010ai,walsh2011computational,davies2014complexity} and the cost of obtaining sufficient information on other votes \citep{conitzer2011dominating,Reijngoud2012,endriss2016strategic,hartman2025s}. Despite these obstacles, voters may still misreport their vote if the potential \emph{gain} is sufficiently large (which we discuss next), especially as complexity results do not rule out heuristic manipulation algorithms that work well in practice~\citep{Conitzer06:Nonexistence,davies2014complexity}. Further, much like strategyproofness itself, most complexity results induce a binary classification on rules (manipulation is tractable or not), making it difficult to compare rules that fall on the same side of the dichotomy. The manipulation leverage of a rule, on the other hand, is much more fine-grained, and can depend on $\nvote$ or $\ncand$ (or both, \eg, \Cref{thm:maximin}), on the positional scoring vector, or can be unbounded (\Cref{tab:rules-bounds}).

\paragraph{Gains From Manipulation} Another direction to study the manipulability of rules is to investigate how much a voter can \emph{gain} from misreporting, and hence their incentives to do so. While metrics such as the \emph{incentive ratio} can be used to study gains from manipulation in settings with cardinal utilities (\eg, \citep{chen2011profitable,chen2022incentive, li2024bounding, cheng2022tight, cheng2019improved, bei2025incentive, tao2024fair}), this is not applicable to our model, where only ordinal preferences are available. For such settings, \citet{campbell2009gains} propose measuring the manipulative gain of a rule by the largest improvement, measured in positions in the manipulator's
ranking, that a voter can obtain by misreporting; related works consider
similar ordinal improvements in expectation
\citep{smith1999manipulability,aleskerov1999degree} (see the discussion on manipulation frequency below). Much like our manipulation leverage, the (manipulation) gain of \citep{campbell2009gains} is a worst-case measure over all profiles. However, their aim is less about distinguishing between rules (as we hope to do), and more about generalizing the Gibbard-Satterthwaite impossibility theorem: as the authors point out, ``for all non-dictatorial rules with at least three alternatives in the range, small maximal gains can only be achieved if some one individual has a great deal of power in determining [the winner]'', whereas we only consider rules that treat all voters equally (see \Cref{sec:manipot}). Further, the measures in \citep{smith1999manipulability,aleskerov1999degree,campbell2009gains} implicitly treat each ordinal improvement as equal and therefore do not account for the fact that the same preference order may be consistent with arbitrarily many utility functions; consequently, the same change in the ranking can correspond to very different utility gains (see interpretation (3) in \Cref{sec:how-many}). Our manipulation leverage, on the
other hand, allows us to bound the comparative advantage of a manipulator over \emph{all} possible utility functions consistent with her ordinal ranking. It achieves this by measuring manipulative advantage through the effort needed to
reproduce it truthfully---specifically, by augmenting the voter weight. A similar ordinal-versus-cardinal tension appears in the rich literature on \emph{distortion}~\citep{procaccia2006distortion}, which measures the worst-case ratio (over all ordinal profiles and consistent cardinal utility functions) between the optimal social welfare and that induced by a given ordinal voting rule; see~\citep{anshelevich2021distortion} for a survey and, \eg,~\citep{charikar2022metric,goel2025metric,charikar2024breaking,kizilkaya2022plurality,kizilkaya2023generalized} for some of the work since then. Distortion (and its ``fair'' variants~\citep{goel2017metric,gkatzelis2020resolving,ebadian2024optimized}) captures the worst-case welfare loss, whereas the manipulation leverage captures the worst-case strategic advantage.

\paragraph{Manipulation Frequency} Separately, a substantial body of work studies the \emph{frequency} with which manipulations can occur. To test different rules under this degree of manipulability (also called the Nitzan-Kelly index due to \citep{nitzan1985vulnerability,kelly1993almost}), one can study the fraction of elections in which the rule can be manipulated by using empirical data~\citep{durand2015towards,durand2023coalitional}, fixed distributions of preferences~\citep{aleskerov2019bounds,durand2025instant}, analytical approaches~\citep{friedgut2008elections,favardin2002borda}, or exhaustive enumeration over small instances~\citep{aleskerov1999degree,aleskerov2012manipulability}. \citet{mossel2015quantitative} prove a quantitative version of the Gibbard-Satterthwaite theorem lower-bounding the probability of manipulations under the impartial culture distribution (all preferences are equally likely). Along with the frequency of manipulations, some of this work measures the expected values of metrics such as the minimum coalition size necessary for manipulation~\citep{pritchard2007exact,chamberlin1985investigation} or the (ordinal) improvement of the outcome in the preferences of the manipulator~\citep{smith1999manipulability,aleskerov1999degree} (see the discussion on gains above). 
Measuring frequencies, however, inevitably requires an assumption over how the voters' preferences are distributed. In many cases, the assumed distribution is far from realistic, such as impartial culture or its variations. Manipulation leverage, on the other hand, is \emph{prior-free}, a property that lies at the heart of the resource augmentation approach (and of the distortion literature discussed above). As such, our results bound the advantage of a manipulator in \emph{any} profile, which can be critical for information-scarce settings. We discuss how manipulation leverage can be combined with existing metrics in \Cref{sec:future}(g).

\section{Preliminaries}\label{sec:prelim}

For any integer $t \geq 0$, we write $[t] \coloneq \{1,2,\ldots,t\}$. 

\paragraph{Election Instances} Let $\voters \coloneq [\nvote]$ be a finite set of \emph{voters} and $\cands \coloneq \{\cand{1}, \cand{2}, \ldots, \cand{\ncand}\}$ be a finite set of \emph{candidates} for some $\nvote \geq 2$ and $\ncand \geq 3$. We assume that each voter $i \in \voters$ has a strict \emph{preference order} (or \emph{ranking}) $\vote{i}$, which is a strict total order over the candidates $\cands$. We denote the set of all possible preference orders by $\linprefs$. For any voter $i \in \voters$ and any two candidates $\cand{},\cand{}' \in \cands$, we write $\cand{} \vote{i} \cand{}'$ to indicate $ (\cand{},\cand{}') \in \vote{i}$ (\ie, voter $i$ ranks $\cand{}$ strictly above $\cand{}'$) and $\cand{}' \voteq{i} \cand{}$ otherwise (\ie, either $\cand{}' \vote{i} \cand{}$ or $\cand{}'=\cand{}$). For any nonempty subset of candidates $B \subseteq C$ and any preference order $\vote{} \in \linprefs$, we let $\topv_{\vote{}} B$ denote the candidate in $B$ that is ranked highest by $\vote{}$.

A \emph{preference profile} $\profile \coloneq (\vote{i})_{i \in \voters} \in \linprefs^\nvote$ is a vector containing the preference orders of all voters. Given a profile  $\profile$, a voter $i \in \voters$, and an alternative preference order $\vote{i}' \in \linprefs$, we use $(\succ'_i, \succ_{-i})$ to denote the preference profile where voter $i$'s preference order is replaced by $\vote{i}'$ and all other voters' preferences remain unchanged, \ie,
\[ (\vote{i}', \vote{-i}) \coloneq (\vote{1}, \ldots, \vote{i-1}, \vote{i}', \vote{i+1}, \ldots, \vote{\nvote}). \]

We sometimes refer to $\manip{i}$ as a \emph{manipulation} by voter $i$, as opposed to her \emph{truthful vote} $\vote{i}$. For any profile $\profile$, voter $i \in \voters$, and integer $k \geq 0$, we use $\profile \vplus k(\vote{i})$ to denote $\profile$ with $k$ additional copies of voter $i$'s preference order $\vote{i}$, \ie,
\[ \kprofile{k}{i} \coloneq (\succ_1, \ldots, \succ_\nvote, \underbrace{\succ_i,\ldots,\succ_i}_{k\text{ times}}) 
\]
so that $\kprofile{k}{i}$ is a profile with $\nvote + k$ voters rather than $\nvote$.

\paragraph{Rules and Manipulability} A \emph{social choice function (SCF)} (or a \emph{voting rule}) is a mapping $\scf: \linprefs^* \rightarrow \cands$ from each preference profile to a unique\footnote{In \Cref{sec:manipot}, we explain how we incorporate set-valued functions (that might return ties) into our framework.} \emph{winner} among the candidates $\cands$.\footnote{In this paper, we only consider SCFs that are defined on all profiles with \emph{any} number of voters (including the trivial $\nvote=1$).} Given a profile $\profile$ and an SCF $\scf$, a manipulation $\manip{i}$ is \emph{profitable} for voter $i$ if reporting it results in a winner she strictly prefers, \ie, 
\begin{align*}
    \scf(\vote{i}', \vote{-i}) \succ_i \scf(\profile).
\end{align*}
An SCF is \emph{strategy-proof} if it admits no profitable manipulations, and it is \emph{manipulable} otherwise.

\section{Augmentation Strategyproofness and the Manipulation Leverage} \label{sec:manipot}

In this section, we will introduce the central metric of our paper, the \emph{manipulation leverage} of social choice functions. To do so, we first extend the strategyproofness definition given in \Cref{sec:prelim}.

\begin{definition}[$k$-Augmentation Strategyproofness]\label{def:kmanip}
   Fix an electorate size $\nvote$ and a non-negative integer $k$. A social choice function $\scf$ is \emph{$k$-augmentation strategyproof ($\ksp{k})$ on $\nvote$-voter profiles} if for any profile $\profile \in \linprefs^\nvote$, voter $i \in \voters$, and alternative preference order $\manip{i} \in \linprefs$, we have $\scf( \kprofile{k}{i}) \voteq{i} \scf\maniprofile{i}$.
\end{definition}
Here, $\nvote$ is the number of voters in the original profile $\profile$; the two profiles being compared have $\nvote+k$ and $\nvote$ voters, respectively. When $\nvote$ is clear from context, we simply say that $\scf$ is $\ksp{k}$. In words, an SCF is $\ksp{k}$ if a voter is never strictly better off misreporting her vote rather than having $k$ additional copies of her truthful self join the original profile. This indeed generalizes standard strategyproofness on $\nvote$-voter profiles, which is equivalent to 0-ASP. We now define the manipulation leverage of a rule.

\begin{definition}[Manipulation Leverage]\label{def:manipot}
    For any SCF $\scf$ and electorate size $\nvote$, we define the \emph{manipulation leverage}\footnote{An earlier version of this paper uses the term ``manipulation potential'' instead of ``manipulation leverage''.} of $\scf$ on $\nvote$-voter profiles as
    \[\spdeg{\scf}\coloneq \min \{k \in \mathbb{Z}_{\geq 0}: \scf\text{ is }\ksp{k'}\text{ on $\nvote$-voter profiles for all }k'\geq k\}\]
and $\spdeg{\scf}\coloneq \infty$ if no such $k$ exists.
\end{definition}

In words, $\spdeg{\scf}$ is the minimum number of additional truthful copies needed to ensure that a voter weakly prefers the outcome with her copies (even if more copies join) to that of any manipulation she could do in the original profile. From a resource augmentation perspective, an SCF with low manipulation leverage limits the relative impact of a manipulation, because a voter can achieve the same impact by telling the truth with a few additional resources (\ie, votes).

\paragraph{Non-Monotonicities} 
\Cref{def:manipot} does not rule out a $\ksp{k}$ social choice function $\scf$ with $\spdeg{\scf} > k$, since $\scf$ may not be $\ksp{k'}$ for some $k' > k$. While this may seem unnatural, adding copies of a voter's truthful vote may initially \emph{hurt} her in certain profiles under many natural SCFs.\footnote{This is closely related to the \emph{no-show paradox}, where a voter would be better off if they had not voted; \cf.\ \citep{moulin1988condorcet}.} Some profiles may exhibit even more layers of non-monotonicities, as the next example shows.

\begin{example}\label{ex:nonmono}
    Consider a profile $\profile$ with $\nvote=27$ voters and $\ncand=5$ candidates, consisting of:
    \medskip
    \begin{center}
                \begin{tabular}{r|l}
            \cellcola Voter 1 & \cellcola $\cand{1} \vote{1} \cand{2} \vote{1} \cand{3} \vote{1} \cand{4} \vote{1} \cand{5}$\\
            \cellcolb 3 voters & \cellcolb $\cand{1} \vote{i}  \cand{3}  \vote{i}  \cand{4} \vote{i} \cand{2} \vote{i} \cand{5}$\\
            \cellcola 3 voters & \cellcola $ \cand{3}  \vote{i} \cand{5} \vote{i} \cand{4} \vote{i} \cand{2} \vote{i} \cand{1}$\\
            \cellcolb 5 voters & \cellcolb $\cand{4} \vote{i}  \cand{3}  \vote{i}  \cand{2} \vote{i} \cand{5} \vote{i} \cand{1}$\\
            \cellcola 6 voters & \cellcola $\cand{2} \vote{i}  \cand{3}  \vote{i}  \cand{4} \vote{i} \cand{5} \vote{i} \cand{1}$\\
            \cellcolb 9 voters & \cellcolb $\cand{5} \vote{i}   \cand{3}  \vote{i} \cand{4} \vote{i} \cand{2} \vote{i} \cand{1}$
        \end{tabular}
    \end{center}
    \medskip
    If we run Instant Runoff ($\irun$, as defined in \Cref{ex:intro}) on $\profile$, it will eliminate $\cand{3}$ first, followed by $\cand{1}$, $\cand{2}$, and $\cand{5}$, making $\cand{4}$ the winner. If voter 1 instead reported the preference order ${\color{purple} \cand{3}} \manip{1} \cand{1} \manip{1} \cand{2} \manip{1} \cand{4} \manip{1} \cand{5}$, she would ensure $\irun\maniprofile{1}=\cand{3}$, which is a better outcome for her. As seen from \Cref{fig:nonmono}, on the other hand, adding gradually more copies of voter 1 to the original profile will initially result in a better outcome than the manipulation, followed by the worst possible outcome for her, before eventually securing a better outcome than manipulating even if more copies join.
\end{example}

\begin{figure}
    \begin{center}
        \begin{tikzpicture}
\begin{axis}[
    width=15cm,
    height=4.5cm,
    xmin=0, xmax=20.5,
    ymin=0.5, ymax=5.5,
    xlabel={$k=$ Number of additional copies of voter 1},
    ylabel={Voter 1's satisfaction},
    xtick={0,1,2,3,4,5,6,7,8,9,10,11,12,13,14,15,16,17,18,19,20.5},
    xticklabels={
        0,1,2,3,4,5,6,7,8,9,
        10,11,12,13,14,15,16,17,18,19,$\geq 20$
    },
    ytick={1,2,3,4,5},
    yticklabels={$\,\cand{5}$,$\,\cand{4}$,$\,\cand{3}$,$\,\cand{2}$,$\,\cand{1}$},
    axis lines=left,
    axis line style={-},
    clip=false,
    legend style={
        at={(0.53,0.6)},
        anchor=south west,
        draw=black,
        fill=white,
        rounded corners=2pt,
        legend cell align=left
    }
]

\addplot[
    draw=none,
    fill=red!15,
    forget plot
] coordinates {(0,0.5) (1,0.5) (1,5.5) (0,5.5)};

\addplot[
    draw=none,
    fill=green!15,
    forget plot
] coordinates {(1,0.5) (7,0.5) (7,5.5) (1,5.5)};

\addplot[
    draw=none,
    fill=red!15,
    forget plot
] coordinates {(7,0.5) (19,0.5) (19,5.5) (7,5.5)};

\addplot[
    draw=none,
    fill=green!15,
    forget plot
] coordinates {(19,0.5) (20.5,0.5) (20.5,5.5) (19,5.5)};

\addplot[dashed, thick, red]
coordinates {(0,3) (20.5,3)};
\addlegendentry{$\irun\maniprofile{1}$}

\addplot[thick, blue]
coordinates {(0,2) (1,2)};
\addplot[thick, blue]
coordinates {(1,4) (7,4)};
\addplot[thick, blue]
coordinates {(7,1) (19,1)};
\addplot[thick, blue]
coordinates {(19,5) (20.5,5)};

\addplot[
    only marks,
    mark=*,
    mark size=2pt,
    blue
] coordinates {
    (1,2) (1,4)
    (7,4) (7,1)
    (19,1) (19,5)
};
\addlegendentry{$\irun(\kprofile{k}{1})$}

\end{axis}
\end{tikzpicture}
        \caption{For the profile from \Cref{ex:nonmono}, the $\irun$ winners of $\kprofile{k}{1}$ and $\maniprofile{1}$ as a function of $k$. Green regions indicate the values of $k$ for which voter 1 prefers the outcome of the former to that of the latter. Even though voter 1 prefers the outcome of $\irun(\kprofile{k}{1})$ to that of $\irun\maniprofile{1}$ for $1 < k <7$, this is not the case for $7<k<19$. As a result, the manipulation leverage of $\irun$ is at least 19 for $\nvote=27$ and $\ncand=5$ (see \Cref{def:manipot}).}\label{fig:nonmono}
    \end{center}
\end{figure}
An alternative definition for manipulation leverage would have been the minimum $k$ for which a rule is $\ksp{k}$, regardless of its behavior for any $k'>k$. However, this type of resource augmentation would suffer from miscoordination issues similar to strategic voting, namely that overshooting can result in undesirable outcomes~\citep{slinko2014safe}. As our goal is to benchmark manipulation against the (resource augmented version of) a simple strategy, \Cref{def:manipot} leaves no room for such miscoordination. In any case, as we will see in \Cref{sec:rules}, whenever we prove an SCF $\scf$ has $\spdeg{\scf}\geq k$ for some $k$, we will also show that $\scf$ is not $\ksp{k'}$ for \emph{any} $k'<k$, so the result under either definition would have been the same.\footnote{One exception to this is our lower bound for biranking-majority-consistent rules presented in \Cref{thm:biranking} (see \Cref{remark:odd-n}).} This shows that while non-monotonicities can arise in specific profiles like \Cref{ex:nonmono}, they generally do not affect the manipulation leverage in the worst case.

\paragraph{Tie(breaker)s} 
Many of the rules we consider in later sections are in fact \emph{social choice correspondences (SCCs)} of the form $\scc: \linprefs^* \rightarrow \powerset(\cands) \setminus \{\emptyset\}$, where $\powerset(\cands)$ is the power set of $\cands$. In words, an SCC may return multiple winning candidates at once; in this case, we say the winners are \emph{tied}.
To extend \Cref{def:kmanip,def:manipot} to a social choice correspondence $\scc$, we equip it with a \emph{linear tie-breaker} $\tieorder \in \linprefs$ and define $\scc_\tiebreaker$ as the social choice function that maps a profile $\profile$ to the top choice of $\tieorder$ in $\scc(\profile)$, \ie,
\[
\scc_\tiebreaker(\profile) \coloneq \topv_{\tieorder} \scc(\profile).\footnote{When breaking ties in this way, $\ksp{k}$ for SCCs is a generalization of \emph{strong Fishburn-strategyproofness} (with equivalence at $k=0$). For an excellent overview of different tiebreakers for SCCs and the corresponding manipulability notions, see \citep{brandt2022strategyproof}.}
\]

We then say an SCC $\scc$ is $\ksp{k}$ on $\nvote$-voter profiles if $\scc_\tiebreaker$ is $\ksp{k}$ on $\nvote$-voter profiles \emph{for all} $\tieorder \in \linprefs$, so that we have
\[\spdeg{\scc} \coloneq \max_{\tieorder \in \linprefs} \spdeg{\scc_\tiebreaker}.\]

In fact, the manipulation leverage of the SCCs we consider will be independent of the specific tiebreaking order $\tieorder$, since they all satisfy \emph{neutrality}, \ie, their outputs are robust to relabeling the candidates. Formally, given a candidate permutation $\cperm: \cands \rightarrow \cands$ and profile $\profile$, define $\cperm(\profile)$ as the profile $\profile'$ where for any voter $i \in \voters$ and candidate pair $\cand{}, \cand{}' \in \cands$, we have $\cand{} \vote{i} \cand{}'$ if and only if $\cperm(\cand{}) \vote{i}' \cperm(\cand{}')$. Then, an SCC $\scc$ is \emph{neutral} if for any candidate permutation $\cperm$, profile $\profile$, and candidate $\cand{} \in \cands$, we have $\cand{} \in \scc(\profile)$ if and only if $\cperm(\cand{}) \in \scc(\cperm(\profile))$. It is worth noting that even if the set-valued SCC $\scc$ is neutral, the single-valued SCF $\scc_\tiebreaker$ (for a specific tiebreaker $\tieorder$) need not be, since $\tieorder$ will prioritize some candidates over others.
For simplicity, we sometimes slightly abuse terminology and refer to both SCFs and SCCs as ``rules'', in which case whether they are single- or multi-valued will be clear from context.

Additionally, all of the SCCs we consider satisfy \emph{anonymity}, \ie, their outputs are robust to relabeling the voters. Formally, given a voter permutation $\vperm: \voters \rightarrow \voters$ and a profile $\profile$, define $\vperm(\profile)$ as the profile $\profile'$ where for any voter $i \in \voters$ and candidate pair $\cand{}, \cand{}' \in \cands$, we have $\cand{} \vote{i} \cand{}'$ if and only if $\cand{} \vote{\vperm(i)}' \cand{}'$. Then, an SCC $\scc$ is \emph{anonymous} if for any voter permutation $\vperm$ and profile $\profile$, we have $\scc(\profile)= \scc(\vperm(\profile))$. Anonymous SCFs are defined analogously. With anonymous rules, only the number of times each vote is cast matters, and not the identities of the voters that cast them. This property will be helpful when dealing with a varying number of total voters, \eg, when going from $\profile$ to $\kprofile{k}{i}$ for some voter $i \in \voters$.

\section{The  Manipulation Leverage of Well-Known Voting Rules and  Classes}\label{sec:rules}

In this section, we report our results on the manipulation leverage of known social choice correspondences, which are summarized in \Cref{tab:rules-bounds}.
For each SCC $\scc$, the proofs consist of establishing upper and lower bounds for $\spdeg{\scc}$. For clarity, we outline below the general structure of our proofs. 

\paragraph{Proof Structure} Leveraging anonymity and neutrality of the social choice correspondences we consider (\Cref{sec:manipot}), we will focus, without loss of generality, on manipulations by voter $1 \in \voters$ and assume that her true preferences are $\cand{1} \vote{1} \cand{2} \vote{1} \cdots \vote{1} \cand{\ncand}$.\footnote{It is also possible to instead fix the tiebreaker to be $\cand{1} \tieorder \cand{2} \tieorder \cdots \tieorder \cand{\ncand}$, which amounts to lexicographic tiebreaking. Since the manipulator's ranking will come up much more often in our proofs, we fix that rather than the tiebreaker for convenience. } To show that a given SCC $\scc$ is not $\ksp{k}$ (\Cref{def:kmanip}) it will suffice to provide a specific profile $\profile$, manipulation $\manip{1}$, and tiebreaker $\tieorder$ such that $\scc_\tiebreaker\maniprofile{1} \vote{1} \scc_\tiebreaker(\kprofile{k}{1})$, which will also imply that the manipulation leverage of $\scc$ is strictly more than $k$, \ie, $\spdeg{\scc} > k$. For constructing these profiles, it will sometimes be helpful to assume a (constant) lower bound or divisibility requirement on $\nvote$ or $\ncand$; these will be clearly stated in the theorem statements. On the other hand, for proving upper bounds on the manipulation leverage---\ie, $\spdeg{\scc} \leq k$---we will have to prove that \emph{for any} $\profile, \manip{1}, \tieorder$, and $k'\geq k$, we have  $\scc_\tiebreaker(\kprofile{k'}{1}) \voteq{1} \scc_\tiebreaker\maniprofile{1}$.

For example, an SCC $\scc$ is \emph{majority-consistent} if whenever a candidate $\cand{} \in \cands$ is the top choice of a strict majority (\ie, $>\frac{\nvote}{2}$) of voters in $\profile$, we have $\scc(\profile)=\{\cand{}\}$. For any majority-consistent rule, adding at least $\nvote-1$ copies of a voter will guarantee that her most-preferred candidate wins.
Since no manipulation can lead to a better result than this, the upper bound criterion given above is satisfied.

\begin{claim}\label{remark:majority-consistent}
    The manipulation leverage of any majority-consistent rule is at most $\nvote-1$.
\end{claim}

Many, but not all, rules we consider are majority-consistent (for a non-example, consider Borda Count). As we show next, for some majority-consistent rules, we can achieve a tighter upper bound than that of \Cref{remark:majority-consistent}.

\subsection{Plurality (with and without Runoff), Instant Runoff, and Borda Count}\label{sec:warmup}

We start our analysis of manipulation leverage with four example rules.

\subsubsection{Plurality} 

Under Plurality  ($\plurality$), the winner is the candidate who is ranked top by the greatest number of voters (ties broken according to the tiebreaker $\tieorder$). While plurality is majority-consistent, as we show next, its manipulation leverage is roughly half of the upper bound in \Cref{remark:majority-consistent}.
\begin{theorem}\label{thm:plur} The manipulation leverage of plurality is $\spdeg{\plurality}=\lceil \frac{\nvote-1}{2} \rceil$. Further, $\plurality$ is not $\ksp{k}$ for any $0 \leq k<\lceil \frac{\nvote-1}{2} \rceil$.
\end{theorem}

Recall from \Cref{def:manipot} that to prove $\spdeg{\plurality} \geq \lceil \frac{\nvote-1}{2} \rceil$, it would have been sufficient to show $\plurality$ is not $\ksp{k}$ for $k=\lceil \frac{\nvote-1}{2} \rceil-1$. Hence, the second part of the theorem statement is indeed a stronger claim  than $\spdeg{\plurality} \geq \lceil \frac{\nvote-1}{2} \rceil$ (\cf.\ note on non-monotonicities in \Cref{sec:manipot}).

\begin{proof}
    We first prove that plurality is \emph{not} $\ksp{k}$ for any $0 \leq k <\lceil \frac{\nvote-1}{2} \rceil$. To do so, we fix any such $k$ and provide an instance (\ie, a specific profile, tiebreaker, and manipulation) where voter~1 strictly prefers the outcome of the manipulation to that of adding $k$ copies of her truthful vote. 
    Say the profile $\profile$ is:
    \begin{center}
    \begin{tabular}{c|l}
            \cellcola Voter $1$ & \cellcola  $\cand{1} \vote{1} \cand{2} \vote{1} \cdots \vote{1}\cand{\ncand}$\\
            \cellcolb $\lfloor \frac{\nvote-1}{2} \rfloor$ voters & \cellcolb $\cand{2} \vote{i} \cdots$ \\
            \cellcola  $\lceil \frac{\nvote-1}{2} \rceil$ voters & \cellcola  $\cand{3} \vote{i} \cdots $
    \end{tabular}
    \end{center}
     If $\nvote$ is odd, say the tiebreaker ranks $\cand{3} \tieorder \cand{2} \tieorder \cand{1}$; otherwise, say $\cand{2} \tieorder \cand{3} \tieorder \cand{1}$. If we add $k$  truthful copies of voter 1 (\ie, $\kprofile{k}{1}$), the number of voters that have $\cand{1}, \cand{2}, \cand{3}$ as their top candidate will be $k+1, \lfloor \frac{\nvote-1}{2} \rfloor$, and $\lceil \frac{\nvote-1}{2} \rceil$, respectively. We cannot have $\plurality_\tiebreaker(\kprofile{k}{1}) = \cand{1}$ since $k+1 \leq \lceil \frac{\nvote-1}{2} \rceil$ and $\cand{3} \tieorder \cand{1}$, so $\plurality_\tiebreaker$ will pick $\cand{3}$ over $\cand{1}$.  Similarly, $\cand{2}$ cannot win: if $\nvote$ is even, it has strictly fewer top-choice votes than $\cand{3}$; if $\nvote$ is odd, they are tied and the tiebreaker $\tieorder$ breaks the tie in favor of $\cand{3}$. We therefore have $\plurality_\tiebreaker(\kprofile{k}{1})= \cand{3}$. Now, go back to the original profile $\profile$ (with no additional copies of voter 1), and say voter 1 reports an alternative preference order $\manip{1}$ that ranks $\cand{2}$ top. If $\nvote$ is odd, $\cand{2}$ now has $\frac{\nvote+1}{2}$ voters ranking it top while $\cand{3}$ has $\frac{\nvote-1}{2}$. If $\nvote$ is even, they both have $\frac{\nvote}{2}$ voters ranking them top, and $\tieorder$ breaks the tie in favor of $\cand{2}$. Hence, $\plurality_{\tiebreaker}\maniprofile{1} = \cand{2} \vote{1} \cand{3} = \plurality_{\tiebreaker}(\kprofile{k}{1})$, implying plurality is not $\ksp{k}$ for any $0 \leq k < \lceil \frac{\nvote-1}{2} \rceil$. We thus have $\spdeg{\plurality} \geq \lceil \frac{\nvote-1}{2} \rceil$.

    For the upper bound, we will prove that plurality is $\ksp{k}$ for any $k \geq \lceil \frac{\nvote-1}{2} \rceil$. To do so, fix some profile $\profile$, tiebreaker $\tieorder$, manipulation $\manip{1}$, and integer $k \geq \lceil \frac{\nvote-1}{2} \rceil$. We want to show that $\plurality_\tiebreaker(\kprofile{k}{1}) \voteq{1} \plurality_\tiebreaker\maniprofile{1}$. First, we rule out two simple cases. If $\cand{1} \in \plurality(\profile)$ (\ie, $\cand{1}$ is \emph{one of the} plurality winners of $\profile$, pre-tiebreaking), then adding $k$ copies of $\vote{1}$ will make it the sole plurality winner (since $k \geq 1$), implying $\plurality_\tiebreaker(\kprofile{k}{1})  = \cand{1} \voteq{1} \plurality_\tiebreaker\maniprofile{1}$ since $\cand{1}$ is the top choice in $\vote{1}$. Similarly, if $\plurality_\tiebreaker(\profile)=\plurality_\tiebreaker\maniprofile{1}$ (\ie, voter 1's manipulation does not change the winner), we will indeed have $\plurality_\tiebreaker(\kprofile{k}{1}) \voteq{1} \plurality_\tiebreaker\maniprofile{1}$, since adding more copies of $\vote{1}$ to $\profile$ can either not change anything or make $\cand{1}$ the winner after getting sufficiently many votes. The only remaining case is when $\cand{1} \notin \plurality(\profile)$ and $\plurality_\tiebreaker(\profile) = \cand{} \neq \cand{}' = \plurality_\tiebreaker\maniprofile{1}$ for some $\cand{},\cand{}' \in \cands$. The only way voter 1 could have caused this change is by submitting a vote $\manip{1}$ that ranks $\cand{}'$ top instead of $\cand{1}$ (We cannot have $\cand{}'= \cand{1}$, since the number of voters ranking  $\cand{1}$ top can only decrease when going from $\profile$ to $\maniprofile{1}$, whereas this number can only increase for any other candidate, implying $\cand{1} \notin \plurality\maniprofile{1}$ since $\cand{1} \notin\ \plurality(\profile)$ by assumption). If $\nvote_{\cand{}}$ and $\nvote_{\cand{}'}$ are the number of voters that rank $\cand{}$ and $\cand{}'$ top in $\profile$, respectively, then we must have $\nvote_{\cand{}'} \geq \nvote_{\cand{}} -1$, since otherwise $\cand{}' \notin \plurality\maniprofile{1}$ (as the top-choice votes of $\cand{}'$ could have only gone up by one), which contradicts $\cand{}' = \plurality_\tiebreaker\maniprofile{1}$. Since voter 1 ranks $\cand{1} \notin \{\cand{},\cand{}'\}$ top in $\profile$, this implies we have $\nvote \geq \nvote_{\cand{}} + \nvote_{\cand{}'} +1 \geq 2  \nvote_{\cand{}} \Rightarrow  \nvote_{\cand{}} \leq \frac{\nvote}{2}$. Since $\cand{1}$ has $k+1 \geq \lceil \frac{\nvote-1}{2} \rceil+1 > \frac{\nvote}{2}$ voters that rank it top in $\kprofile{k}{1}$, this implies $\plurality_\tiebreaker(\kprofile{k}{1})  = \cand{1} \voteq{1} \plurality_\tiebreaker\maniprofile{1}$, completing our proof.
\end{proof}

That plurality beats the upper bound in \Cref{remark:majority-consistent} brings about a natural question: can we find a tighter bound for the manipulation leverage of \emph{all} majority-consistent rules? To answer this question, we next consider adding a single or multiple \emph{runoff rounds} to Plurality.

\subsubsection{Plurality with Runoff \& Instant Runoff}
Just like with Plurality, under Plurality with Runoff  ($\plurWrunoff$), each voter gives one point to her top-ranked candidate in the first round, but now the \emph{two} candidates with the highest scores advance to the runoff (\ie, the second round).
In the runoff, each voter gives one point to her preferred candidate among the two, and the candidate with the higher score wins. With Instant Runoff ($\irun$),\footnote{Also known as Single Transferable Vote (STV).} on the other hand, the winner is determined by iteratively eliminating the candidate with the fewest top-choice votes from the profile, as explained in \Cref{ex:intro}.
This continues until only one candidate remains, who is declared the winner.

For both rules, we define their set-valued output (prior to tiebreaking) using \emph{parallel-universes tiebreaking} \citep{conitzer2009preference}, which dictates that a candidate $c$ is one of the winners (\ie, $\cand{} \in \plurWrunoff(\profile)$ or $\cand{} \in \irun(\profile)$) if there exists \emph{some} way to break ties (when deciding who advances to the next round of the runoff) such that $\cand{}$ wins. As with other SCCs, the tiebreaker $\tieorder$ then resolves the tie in favor of its highest-ranked candidate in $\plurWrunoff(\profile)$ and $\irun(\profile)$ to determine $\plurWrunoff_\tiebreaker(\profile)$ and $\irun_\tiebreaker(\profile)$, respectively.

Both $\plurWrunoff$ and $\irun$ are majority-consistent. Further, as we show next, their manipulation leverage is exactly $\nvote-1$, showing that the upper bound in \Cref{remark:majority-consistent} is the tightest we can get for this class.

\begin{theorem}\label{thm:runoff}
    The manipulation leverage of Plurality with Runoff and the manipulation leverage of Instant Runoff are $\spdeg{\plurWrunoff}=\spdeg{\irun}=\nvote-1$ (lower bound assumes $\ncand \geq 4$, $\nvote \geq 5$ and that $\nvote-1$ is divisible by 4). Further, neither rule is $\ksp{k}$ for any $0 \leq k< \nvote-1$.
\end{theorem}

\begin{proof}
    The upper bound for both rules follows from \Cref{remark:majority-consistent}. It remains to prove that neither $\plurWrunoff$ nor $\irun$ is \ksp{$k$} for any $0 \leq k < \nvote-1$. Fix such a $k$, and consider the following profile $\profile$:
    \begin{center}
    \begin{tabular}{c|l}
            \cellcola Voter $1$ & \cellcola  $\cand{1} \vote{1} \cand{2} \vote{1} \cdots \vote{1}\cand{\ncand}$\\
            
            \cellcolb $\frac{\nvote-1}{4}$ voters & \cellcolb $\cand{2} \vote{i} \cand{3} \vote{i} \cand{4} \vote{i} \cdots$ \\
            
            \cellcola  $\frac{\nvote-1}{4}$ voters & \cellcola  $\cand{3} \vote{i} \cand{4} \vote{i} \cand{2} \vote{i} \cdots $\\
            
            \cellcolb $\frac{\nvote-1}{2}$ voters & \cellcolb  $\cand{4} \vote{i} \cdots $
    \end{tabular}
    \end{center}
    Say the tiebreaker $\tieorder$ ranks $\cand{4}$ top. We first show that $\plurWrunoff_\tiebreaker(\kprofile{k}{1})=\irun_\tiebreaker(\kprofile{k}{1})=\cand{4}$. To do so, it is sufficient to prove that $\cand{4}$ wins under \emph{some} way of breaking ties in each rule---\ie, we have $\cand{4} \in \plurWrunoff(\kprofile{k}{1}) \cap \irun(\kprofile{k}{1})$---as $\tieorder$ will break the tie in favor of $\cand{4}$. Consider two cases.
    
    \textbf{Case 1:} $0 \leq k \leq \frac{\nvote-5}{4}$: Here, $\cand{4}$ starts with strictly more voters ranking it top than any other candidate, $\cand{2}$ and $\cand{3}$ tie for second place, and $\cand{1}$ has weakly fewer points than any other candidate. Under Plurality with Runoff ($\plurWrunoff$), if $\cand{2}$ is advanced to the runoff with $\cand{4}$, then $\cand{4}$ wins. With Instant Runoff ($\irun$), if $\cand{1}$ is eliminated first (after the elimination of all $\cand{\ell}$ for $\ell>4)$, it will be followed by $\cand{3}$ and then $\cand{2}$, once again making $\cand{4}$ win. Hence, $\cand{4}$ is among the winners for both rules, and thus is the final winner picked by $\tieorder$.
    
    \textbf{Case 2:} $\frac{\nvote-5}{4} < k < \nvote-1 $: Here, $\cand{2}$ and $\cand{3}$ tie for the fewest top-choice votes. With $\plurWrunoff$, $\cand{1}$ is advanced to the runoff with $\cand{4}$, but then loses to $\cand{4}$ since it has $k+1 \leq \nvote-1$ points in the second round. With $\irun$, if $\cand{3}$ is eliminated first (after all $\cand{\ell}$ for $\ell>4)$, it will be followed by $\cand{2}$ and then $\cand{1}$, once again making $\cand{4}$ the winner. Hence, $\cand{4}$ is among the winners for both rules and thus the final winner picked by $\tieorder$. 
    
    Next, go back to the original profile (with no copies) and consider the case where voter $1$ reports an alternative preference order $\manip{1}$ ranking
    $\cand{3}$ as her top candidate. For Plurality with Runoff ($\plurWrunoff$), $\cand{3}$ and $\cand{4}$ are advanced to the runoff, and in the second round $\cand{3}$ wins. For Instant Runoff $(\irun)$, $\cand{1}$ is eliminated first (after all $\cand{\ell}$ for $\ell>4)$, followed by $\cand{2}$ and then $\cand{4}$, once again making $\cand{3}$ the winner. Therefore, $\plurWrunoff_\tiebreaker\maniprofile{1}=\irun_\tiebreaker\maniprofile{1}=\cand{3} \vote{1} \cand{4} = \plurWrunoff_\tiebreaker(\kprofile{k}{1})=\irun_\tiebreaker(\kprofile{k}{1})$. This shows that neither rule is $\ksp{k}$ for any $0 \leq k < \nvote-1$, completing the proof.
\end{proof}

All three of the rules we have seen so far ($\plurality, \plurWrunoff, \irun$) have manipulation leverage growing linearly with the number of voters $\nvote$, even if the coefficient of growth for $\plurality$ is half as big. Can there be a rule whose manipulation leverage \emph{does not} grow with $\nvote$? We next answer this question in the affirmative.

\subsubsection{Borda Count} As introduced in \Cref{ex:intro}, under Borda Count ($\borda$), each voter contributes $\ncand-j$ points to their $j^\text{th}$-ranked candidate for $j \in [\ncand]$ and the candidate with the highest total ``Borda score'' wins. As we show next, unlike the previous rules, the manipulation leverage of $\borda$ grows linearly with the number of candidates and is independent of the number of voters.

\begin{theorem}\label{thm:borda}
    The manipulation leverage of Borda Count is $\spdeg{\borda}=\ncand-2$ (lower bound assumes $\nvote$ is even). Further, $\borda$ is not $\ksp{k}$ for any $0 \leq k < \ncand-2$.
\end{theorem}

\begin{proof}
     We first prove that Borda Count is not $\ksp{k}$ for any $0 \leq k < \ncand-2$, which also implies $\spdeg{\borda} \geq \ncand-2$. Fix any such $k$ and separate the candidates into two groups: $\{\bcand{1}, \bcand{2}\}$ and $\{\acand{3}, \aldots, \acand{\ncand}\}$ (colored differently for presentation purposes). Consider a profile $\profile$, consisting of:
\begin{center}
    \begin{tabular}{c|l}
    \cellcola Voter 1 & \cellcola $\bcand{1} \vote{1} \bcand{2} \vote{1} \acand{3} \vote{1} \acand{4} \vote{1} \acdots \vote{1}\acand{\ncand}$
    \\ \cellcolb Voter 2 & \cellcolb $\bcand{2} \vote{2} \acand{\ncand} \vote{2} \acand{\ncand-1} \vote{2} \acdots\vote{2} \acand{3} \vote{2} \bcand{1}$
    \\\cellcola $\frac{\nvote-2}{2}$ voters (Group 1) & \cellcola $\bcand{2} \vote{i} \bcand{1} \vote{i} \acand{3} \vote{i} \acand{4} \vote{i} \acdots \vote{i} \acand{\ncand}$
    \\ \cellcolb $\frac{\nvote-2}{2}$ voters (Group 2) & \cellcolb $\bcand{1} \vote{i} \bcand{2} \vote{i} \acand{\ncand} \vote{i} \acand{\ncand-1} \vote{i} \acdots \vote{i} \acand{3}$
    \end{tabular} 
\end{center}
Assume that the tiebreaker is identical to voter 1's preference order, \ie, $\tieorder =\vote{1}$. For each candidate $ \cand{} \in \cands$, let $\bs_k(\cand{})$ denote its total Borda score from all voters in $\kprofile{k}{1}$. Then, we have 
\begin{align*}
 \bs_k(\bcand{1}) &= \underbrace{(k+1)(\ncand-1)}_{\text{Voter 1 and her copies}} + \underbrace{0}_{\text{Voter 2}}  +\underbrace{\frac{\nvote-2}{2} \cdot (\ncand-2)}_{\text{Group 1}} + \underbrace{\frac{\nvote-2}{2} \cdot (\ncand-1)}_{\text{Group 2}},\\
    \bs_k(\bcand{2}) &= \underbrace{(k+1)(\ncand-2)}_{\text{Voter 1 and her copies}} + \underbrace{(\ncand-1)}_{\text{Voter 2}} +\underbrace{\frac{\nvote-2}{2} \cdot (\ncand-1)}_{\text{Group 1}} + \underbrace{\frac{\nvote-2}{2} \cdot (\ncand-2)}_{\text{Group 2}}\text{, and}\\
    \forall \ell \in \{3,\ldots,\ncand\}:&\quad \bs_k(\acand{\ell})  = \underbrace{(k+1)(\ncand-\ell)}_{\text{Voter 1 and her copies}} + \underbrace{(\ell-2)}_{\text{Voter 2}} +  \underbrace{\frac{\nvote-2}{2} \cdot (\ncand-\ell)}_{\text{Group 1}} + \underbrace{\frac{\nvote-2}{2} \cdot (\ell-3)}_{\text{Group 2}}.
\end{align*}

As $\bs_k(\bcand{2}) - \bs_k(\bcand{1}) =(\ncand-1) -(k+1) >0$ and $\bs_k(\bcand{2}) - \bs_k (\acand{\ell})= k(\ell-2)+ \frac{\nvote\ncand}{2}-1>0$ for any $\ell \in \{3,\ldots, \ncand\}$, we get that $\borda_\tiebreaker(\kprofile{k}{1})=\bcand{2}$. Now consider an alternative preference order $\manip{1}$ defined as $\bcand{1} \manip{1} \acand{3} \manip{1} \acand{4} \manip{1} \acdots \manip{1} \acand{\ncand} \manip{1} \bcand{2}$ (\ie, voter 1 moves $\bcand{2}$ to the bottom of her ranking). For each $ \cand{} \in \cands$, let $\bs'(\cand{})$ denote its total Borda score in $\maniprofile{1}$. Then $\bs'(\bcand{1})- \bs'(\bcand{2})=0$ and $\bs'(\bcand{1})- \bs'(\acand{\ell})  = \frac{(\nvote-2)\ncand}{2} \geq 0$ for any $\ell \in \{3,\ldots, \ncand\}$, so $\bcand{1} \in \borda\maniprofile{1}$. Since $\tieorder$ ranks $\bcand{1}$ top, we have $ \borda_\tiebreaker\maniprofile{1}=\bcand{1} \vote{1} \bcand{2}=\borda_\tiebreaker(\kprofile{k}{1})$, proving $\borda$ is not $\ksp{k}$ for any $0 \leq k < \ncand-2$.

For the upper bound, fix some $k \geq \ncand-2$, profile $\profile$, manipulation $\manip{1}$ and tiebreaker $\tieorder$. We will show that $\borda_\tiebreaker(\kprofile{k}{1}) \voteq{1} \borda_\tiebreaker\maniprofile{1}$, proving $\borda$ is $\ksp{k}$. Say $B, B', B''$ are (all) the Borda winners in $\profile$, $\maniprofile{1}$, and $\kprofile{k}{1}$, respectively, prior to tiebreaking. 
For each candidate $\cand{} \in \cands$, let $\bs(\cand{}), \bs'(\cand{}), \bs''(\cand{})$ be its respective Borda scores in these three profiles, and let $\bs_1(\cand{}), \bs'_1(\cand{}), \bs''_1(\cand{})$ denote voter~1’s contributions to these scores.
As all other votes are unchanged, $\bs'(\cand{})= \bs(\cand{})-\bs_1(\cand{})+\bs'_1(\cand{})$ and $\bs''(\cand{})=\bs(\cand{})+k\bs_1(\cand{})$. 
We will show that for any $\cand{}' \in B'$ and $\cand{}'' \in B''$, either $\cand{}'' \voteq{1} \cand{}'$, or $\{\cand{}', \cand{}''\} \subseteq B' \cap B''$.
This will complete the proof, as taking $\cand{}' = \topv_{\tieorder} B'$ and $\cand{}'' = \topv_{\tieorder} B''$ implies $\borda_\tiebreaker(\kprofile{k}{1}) = \topv_{\tieorder} B'' \voteq1 \topv_{\tieorder} B' =  \borda_\tiebreaker\maniprofile{1}$
(either directly or by the consistency of the tiebreaker~$\tieorder$).
To prove our claim on \emph{any} $\cand{}' \in B'$ and $\cand{}'' \in B''$, let $c^* \coloneq \topv_{\vote{1}} B$ (\ie, voter 1's top candidate in $B$), and consider the following two cases:

    \textbf{Case 1:} $c^* \voteq{1} \cand{}'$. Since $\bs''(\cand{}'') \geq \bs''(\cand{}^*)$ and  $\bs(\cand{}^*) \geq \bs(\cand{}'')$, we have $k\bs_1(\cand{}'')=\bs''(\cand{}'')-\bs(\cand{}'') \geq \bs''(\cand{}^*) -\bs(\cand{}^*)=k\bs_1(\cand{}^*)$, implying $\bs_1(\cand{}'') \geq \bs_1(\cand{}^*)$ (since $k \geq \ncand-2 \geq 1$) and therefore $\cand{}'' \voteq{1} c^*$. Since $c^* \voteq{1} \cand{}'$, this implies $\cand{}'' \voteq{1} \cand{}'$, and we are done.     

    \textbf{Case 2:} $\cand{}' \vote{1} c^*$ and thus $\bs_1(\cand{}')-\bs_1(\cand{}^*) \geq 1$. Since $\bs'_1(\cand{}') - \bs'_1(\cand{}^*) \leq \ncand-1$, this implies 
    \begin{align}
    \ncand-2 \geq (\bs'_1(\cand{}') - \bs_1(\cand{}')) - (\bs'_1(\cand{}^*) - \bs_1(\cand{}^*))&= (\bs'(\cand{}') - \bs(\cand{}')) - (\bs'(\cand{}^*) - \bs(\cand{}^*)) \label{eq:firstineq} \\&\geq \bs(\cand{}^*) - \bs(\cand{}'),\label{eq:firstineq2}\end{align}   where the last step follows from $\bs'(\cand{}') \geq \bs'(\cand{}^*)$, since $\cand{}' \in B'$. Assume $\cand{}' \vote{1} \cand{}''$ (otherwise we are done); we will show $\{\cand{}',\cand{}''\} \subseteq B' \cap B''$. Since $\bs''(\cand{}'') \geq \bs''(\cand{}')$ and $\bs_1(\cand{}') \geq \bs_1(\cand{}'')+1$, we have 
    \begin{align}
        \bs(\cand{}'')+k\bs_1(\cand{}'') =\bs''(\cand{}'') \geq \bs''(\cand{}')= \bs(\cand{}')+k\bs_1(\cand{}')\label{eq:ineq1} & \geq \bs(\cand{}^*) -(\ncand-2)+k(1+\bs_1(\cand{}''))\\& \geq \bs(\cand{}^*) +k \bs_1(\cand{}''),\label{eq:ineq2}
    \end{align}
    implying $\bs(\cand{}'') \geq \bs(\cand{}^*)$ and thus $\cand{}'' \in B$. In fact, we must have $\cand{}''=c^*$, since any $c \in B\setminus \{c^*\}$ will have $\bs''(\cand{}) = \bs(\cand{}) +k \bs_1(\cand{}) = \bs(\cand{}^*) +k \bs_1(\cand{}) < \bs(\cand{}^*) +k \bs_1(\cand{}^*)= \bs''(\cand{}^*)$, so $\cand{}$ cannot be in $B''$. This implies \eqref{eq:ineq1}-\eqref{eq:ineq2} is in fact an equality, and therefore $\bs''(\cand{}'')=\bs''(\cand{}')$ and $\bs(\cand{}')=\bs(\cand{}^*) -(\ncand-2)$. The latter implies \eqref{eq:firstineq}-\eqref{eq:firstineq2} is also an equality, and therefore $\bs'(\cand{}')=\bs'(\cand{}^*)$. Since $\bs''(\cand{}')=\bs''(\cand{}'')$ and $\bs'(\cand{}'')=\bs'(\cand{}^*)=\bs'(\cand{}')$, and since $\cand{}' \in B'$ and $\cand{}'' \in B''$, this implies $\{\cand{}', \cand{}''\} \subseteq B' \cap B''$, as desired.

\end{proof}

\begin{figure}
    \begin{center}
        \begin{tikzpicture}[line cap=round, line join=round]

\def\Gap{7.2}

\newcommand{\AsymLever}[9]{
  \begin{scope}[shift={(#1,0)}]

    \def\PivotX{0}
    \def\PivotY{0}
    \def\PivotW{0.60}
    \def\PivotH{0.9}

    \coordinate (P) at (\PivotX,\PivotY);
    \coordinate (L) at (\PivotX-#2,\PivotY);
    \coordinate (R) at (\PivotX+#3,\PivotY);

    \draw[gray!70, line width=0.9pt] (L) -- (R);

    \node[above=2pt] at (L) {$\manip{1}$};
    \node[above=2pt] at (R) {{\color{#6}#9}$(\vote{1})$};

    \draw[<->, line width=0.7pt, color=#6]
      ($(P)+(-0.2,-0.1)$) -- ($(L)+(0,-0.1)$);
    \draw[<->, line width=0.7pt, color=#6]
      ($(P)+(0.2,-0.1)$) -- ($(R)+(0,-0.1)$);
    
    \node[below=1.6pt] at ($(L)!0.5!(P) - (0.1,0)$) {\color{#6} #4};
    \node[below=1.6pt] at ($(R)!0.5!(P) + (0.2,0)$) {\color{#6} 1};
    
    \draw[#6, line width=0.9pt, fill={#6!15}]
      (\PivotX-\PivotW,\PivotY-\PivotH) --
      (\PivotX+\PivotW,\PivotY-\PivotH) --
      (P) -- cycle;

    \node at (\PivotX,{\PivotY-#7}) {\scriptsize {#5}};
    \def\MaybeNode{#8}
    \ifx\relax#8\relax
    \else
      \node at (\PivotX,{\PivotY-#7-0.2}) {\scriptsize {#8}};
    \fi
        
  \end{scope}
}

\AsymLever{0}{1.8}{1}{$\approx \nvote/2$}{$\plurality$}{blue}{0.55}{}{$\frac{\nvote}{2}$}
\AsymLever{6cm}{3.6}{1}{$\approx \nvote$}{$\irun$}{blue}{0.53}{/$\plurWrunoff$}{$\nvote$}
\AsymLever{10cm}{1.6}{1}{$\approx m$}{$\borda$}{orange}{0.55}{}{$\ncand$}
\end{tikzpicture}
        \caption{Visualization of our results from \Cref{sec:warmup}. For each rule, the longer the left arm of the scale, the more truthful copies a single manipulation can ``lift''. That is, the voter will need (in the worst case) a larger number of truthful copies to produce an outcome as desirable as that of the manipulation. The left arm of the middle scale is double that of the left scale, whereas comparing it to the rightmost scale depends on the values of
        ${\color{blue} \nvote}$ and ${\color{orange} \ncand}$.}\label{fig:scale}
    \end{center}
\end{figure}

Our analysis of the rules of this subsection $(\plurality, \plurWrunoff, \irun, \borda)$ already displays that the ``effective weight'' of a single manipulation (in terms of truthful copies needed to create the same impact) varies significantly between rules (\Cref{fig:scale}). In the next two subsections, we shift our focus to characterizations and impossibility results on the manipulation leverage of broad rule \emph{classes}.

\subsection{Characterization for Positional Scoring Rules}\label{sec:pos}

Out of the rules studied so far, the lowest manipulation leverage belongs to either Plurality (if $\left\lceil\frac{\nvote-1}{2}\right\rceil < \ncand-2$) or Borda Count (vice versa). As we will show, this in fact remains the case for much broader classes of rules. To formalize this claim, we first introduce a class that includes $\plurality$ and $\borda$.

\begin{definition}[Positional Scoring Rules]\label{def:pos} A (monotonic) \emph{positional scoring rule} is an SCC $\pos$ associated to a vector  ${\score}\coloneq(\score_{j})_{j \in [\ncand]}$ with $1=\score_1\geq \score_2 \geq \cdots \geq \score_\ncand=0$. On input profile $\profile$, the output of $\pos$ is computed as follows:
each voter $i$ contributes $\score_{\ell}$ points to their ${\ell}^\text{th}$-ranked candidate in $\vote{i}$, and $\pos$ returns the candidate(s) with highest total score.
\end{definition}

Out of the rules we have studied so far, Plurality and Borda Count correspond to the positional scoring rules associated with vectors $\score=(1,0,\ldots,0)$ and $\score=(1, \frac{\ncand-2}{\ncand-1}, \frac{\ncand-3}{\ncand-1}, \ldots, \frac{1}{\ncand-1},0)$, respectively (Note that scaling all the scores by $\frac{1}{\ncand-1}$ for Borda Count does not change the output). Our first result indicates that unlike these two rules, the manipulation leverage of some positional scoring rules can be unbounded!

\begin{theorem}\label{thm:pos-infinity}
    For any positional scoring rule $\pos$ with $\score_1=\score_2$, we have $\spdeg{\pos} =\infty$. In particular, $\pos$ is not $\ksp{k}$ for any $k \geq 0$.
\end{theorem}
\begin{proof}
     We will give a profile $\profile$, manipulation $\manip{1}$, and tiebreaker $\tieorder$ such that $\pos_\tiebreaker\maniprofile{1} \vote{1} \pos_\tiebreaker(\kprofile{k}{1})$ for any $k \geq 0$. Consider a profile $\profile$ where \emph{every} voter votes $\cand{1} \vote{} \cand{2} \vote{} \cdots \vote{}\cand{\ncand}$ and the tiebreaking order ranks $\cand{2} \tieorder \cand{1} \tieorder \cand{}$ for any $\cand{} \in \cands \setminus \{\cand{1}, \cand{2}\}$. For any $\cand{} \in \cands$, say $S_k(\cand{})$ is the total score assigned to $\cand{}$ by $\pos$ in $\kprofile{k}{1}$. Then we have $S_k(\cand{2})=(k + \nvote)\score_2=k+\nvote \geq S_k(\cand{})$ for any $\cand{} \in \cands \setminus \{\cand{2}\}$. Since $\cand{2} \tieorder \cand{}$ for any $\cand{} \in \cands\setminus\{\cand{2}\}$, this implies $\pos_\tiebreaker(\kprofile{k}{1})=\cand{2}$ for all $k \geq 0$. Going back to the original profile (with no copies), assume voter 1 instead reported an alternative preference order $\manip{1}$ such that $\cand{1} \manip{1} \cand{} \manip{1} \cand{2}$ for any $\cand{} \in \cands \setminus \{\cand{1},\cand{2}\}$. For any $\cand{} \in \cands{}$, say $\bs'(\cand{})$ is the total score assigned to $\cand{}$ by $\pos$ in $\maniprofile{1}$. We have $\bs'(\cand{2})=(\nvote-1)\score_2+\score_\ncand = \nvote-1$ and $\bs'(\cand{1})=\nvote \score_1 = \nvote \geq \bs'(\cand{})$ for any $\cand{} \in \cands$. Since $\bs'(\cand{1}) \geq \bs'(\cand{})$ and $\cand{1}\tieorder \cand{}$ for any $\cand{} \in \cands \setminus \{\cand{1},\cand{2}\}$ and since $\bs'(\cand{1})>\bs'(\cand{2})$, this implies $\pos_\tiebreaker\maniprofile{1}=\cand{1} \vote{1} \cand{2} = \pos_\tiebreaker(\kprofile{k}{1})$ for any $k \geq 0$, completing our proof.
\end{proof}

In words, if a positional scoring rule treats the top two positions equally, adding more copies of a voter cannot help her top choice gain an edge over her second top choice, while misreporting can.

How about when $\score_1>\score_2$? Ideally, we would like to generalize our approach to $\plurality$ (\Cref{thm:plur}) and $\borda$ (\Cref{thm:borda}). To do so, we introduce some additional notation. For any $\ell\in [\ncand]$, we define
\[\upave{\ell} \coloneq \frac{1}{\ell}\sum_{\ell'=1}^\ell \score_{\ell'} \quad \text{and} \quad  \downave{\ell} \coloneq \frac{1}{\ell}\sum_{\ell'=\ncand-\ell+1}^\ncand \score_{\ell'}\]
as the average score assigned to the top and bottom $\ell$ positions, respectively. For any $\ell \in \{ 2,3,\ldots,\ncand-1\}$, we write $\posbar{\ell} \coloneq \max\{\ell' \in [\ell-1]: \score_{\ell'} > \score_{\ell+1}\}$,  which is well-defined since $\score_1 > \score_2 \geq \score_{\ell+1}$. For example, for Plurality, $\upave{\ell} =\frac{1}{\ell}$; $\downave{\ell} = \frac{1}{\ncand} $ if $\ell =\ncand$ and 0 otherwise; and $\posbar{\ell}=1$ for any $\ell$. With our notation in place, we now present a general upper bound for positional scoring rules.

\begin{restatable}{theorem}{posup}\label{thm:pos-upper}
    For any $\pos$ with $\score_1 \neq \score_2,$ we have $\displaystyle \spdeg{\pos}\leq \max_{\ell \in [\ncand-1]} \left\lceil\upperpos{\ell} \right\rceil$, where \[ \upperpos{\ell}\coloneq \begin{cases} \frac{1}{\score_\ell-\score_{\ell+1}}-1 &\text{if }\ell=1\\
     \min\left(\frac{1}{\score_\ell-\score_{\ell+1}}-1,  \frac{(\nvote-1)(\upave{2}-\downave{\posbar{\ell}})+\frac{1}{2}}{\upave{\posbar{\ell}}-\score_{\ell+1}} \right) & \text{otherwise}
    \end{cases}.
   \]
\end{restatable}
We interpret $\frac{1}{\score_\ell - \score_{\ell+1}}=\infty$ whenever $\score_\ell =\score_{\ell+1}$. Since $\score_1 > \score_2$, we have $\upperpos{1} < \infty$. For all other $\ell \in \{2,\ldots, \ncand-1\}$, we have $\upave{\posbar{\ell}} > \score_{\ell+1}$ by the definition of $\posbar{\ell}$, so $U_s(\ell) < \infty$ since the second entry of the $\min$ is finite. Hence, \Cref{thm:pos-upper} gives a finite upper bound on $\spdeg{\pos}$, showing that the sufficient condition in \Cref{thm:pos-infinity} $(\score_1=\score_2)$ for unbounded manipulation leverage is also necessary. For $\plurality$ and $\borda$, the upper bound from \Cref{thm:pos-upper} corresponds to $\left\lceil\frac{\nvote}{2}\right\rceil$ and $\ncand-2$, respectively; hence this theorem indeed generalizes the upper bounds from \Cref{thm:plur,thm:borda} (with a minor loss for $\plurality$ if $\nvote$ is odd).

\begin{proof}[Proof sketch]
    The proof of \Cref{thm:pos-upper} can be found in \Cref{app:pos-upper}. We start by considering any manipulation where voter 1 changes the winner from $\cand{\ell+1}$ to $\cand{\ell}$ for some $\ell \in [\ncand-1]$ (We show this is WLOG). Since voter 1 by herself can decrease the total score of $\cand{\ell+1}$ by at most $\score_{\ell+1}$ (by ranking it last) and increase the total score of $\cand{\ell}$ by at most $1-\score_{\ell}$ (by ranking it first), this implies that $\cand{\ell+1}$, prior to the manipulation, is ahead of $\cand{\ell}$ by at most $1-(\score_{\ell}-\score_{\ell+1})$ points. The first term inside the $\min$ in $\upperpos{\ell}$ (and the only term for $\upperpos{1}$) then captures the number of truthful copies it will take (in the worst case) for $\cand{\ell}$ to catch up with $\cand{\ell+1}$, as the difference decreases by $\score_{\ell}-\score_{\ell+1}$ with each added copy. For $\ell>1$, however, we must also consider the possibility of some candidate in $\{\cand{1},\cand{2},\ldots,\cand{\ell-1}\}$ reaching the score of $\cand{\ell+1}$ before $\cand{\ell}$ does (for example, this will definitely be the case if $\score_\ell=\score_{\ell+1}$), as this would also be preferred by voter 1 to the outcome of the manipulation ($\cand{\ell}$). In the worst case, all voters except for voter 1 rank candidates $\cand{\ell}$ and $\cand{\ell+1}$ in the top two positions and also rank the candidates in $\{\cand{1}, \cand{2},\ldots,\cand{\posbar{\ell}}\}$ in the bottom $\posbar{\ell}$ positions (We do not have to worry about any candidate $\cand{\ell'}$ for $\posbar{\ell} < \ell' < \ell$, since $\score_{\ell'}=\score_{\ell+1}$ by definition, so $\cand{\ell'}$ will not come closer to $\cand{\ell+1}$ with additional copies). Then the total score of $\cand{\ell+1}$ in the original profile is roughly $(\nvote-1)\cdot \upave{2}$, whereas the \emph{average} score of candidates 1 to $\posbar{\ell}$ is roughly $(\nvote-1)\cdot \downave{\posbar{\ell}}$. Further, in the worst case, the initial scores of candidates 1 to $\posbar{\ell}$ are distributed in such a way that they all require (roughly) the same number of copies to reach the score of $\cand{\ell+1}$, which occurs when their average score reaches that of $\cand{\ell+1}$ (To see why: if any of these candidates reaches $\cand{\ell+1}$ before others do, we could decrease their initial score and distribute it among the other candidates, thereby increasing the minimum number of copies required for some candidate to reach $\cand{\ell+1}$). Since the average score of the candidates in $\{\cand{1}, \cand{2},\ldots,\cand{\posbar{\ell}}\}$ gets closer to that of $\cand{\ell+1}$ by $\upave{\posbar{\ell}}-\score_{\ell+1}$ with each copy, the number of copies required to reach this threshold is roughly $ \frac{(\nvote-1)(\upave{2}-\downave{\posbar{\ell}})}{\upave{\posbar{\ell}}-\score_{\ell+1}}$, giving us the second term of the $\min$ in $\upperpos{\ell}$ (where the additional $1/2$ results from a more precise analysis). The maximization over all $\upperpos{\ell}$ follows from the fact that we should, in the worst case, consider manipulations that lead to any $\cand{\ell}$ as the new winner.
\end{proof}

\begin{remark}
    \Cref{thm:pos-upper} also shows that having manipulation leverage that does not grow with $\nvote$ is not unique to Borda Count. In particular, for any scoring vector $\score$ such that $\score_i \neq \score_j$ for any $i \neq j$, $\spdeg{\pos}$ will eventually stop growing with $\nvote$ due to the first term inside the $\min$ in each $\upperpos{\ell}$. Still, among all such rules, Borda Count is the one with the lowest manipulation leverage, as we show further below (\Cref{thm:pos-best}).
\end{remark}
We next present a general lower bound on the manipulation leverage of positional scoring rules.

\begin{restatable}{proposition}{posdown}\label{thm:pos-lower}
    For any $\pos$ with $\score_1 \neq \score_2$, we have $\spdeg{\pos}\geq \max_{\ell \in [\ncand-1]} \lowerpos{\ell}$, where
    \[  \lowerpos{\ell}\coloneq \begin{cases}\frac{1}{\score_\ell-\score_{\ell+1}}-1 &\text{if }\ell=1\\
     \min\left(\frac{1}{\score_\ell-\score_{\ell+1}}-1, \min_{t \in [\posbar{\ell}]}\left\lfloor \frac{(\nvote-2)(\upave{2}-\score_{\ncand-t+1})+1-\score_{\ncand-t}}{\score_t-\score_{\ell+1}} \right\rfloor \right) & \text{otherwise}\end{cases},\]
    assuming $\nvote$ is even and at least $\frac{2\score_2}{1-\score_2}$. Further, $\pos$ is not $\ksp{k}$ for any $k < \max_{\ell \in [\ncand-1]} \lowerpos{\ell}$.
\end{restatable}

\begin{proof}[Proof sketch]
    The proof can be found in \Cref{app:pos-lower}. To lower-bound the manipulation leverage, we construct a family of profiles, one for each $\ell \in [\ncand-1]$. The profile corresponding to any $\ell$ follows a similar structure to the one we used for proving the lower bound on Borda Count (\Cref{thm:borda}) except $\cand{1}$ and $\cand{2}$ are replaced with $\cand{\ell}$ and $\cand{\ell+1}$, respectively. Further, the candidates in $\{\cand{1},\cand{2},\ldots,\cand{\ell-1}\}$ are placed in a way that approximates the worst-case situation described in the proof sketch of \Cref{thm:pos-upper} above, while making minimal assumptions on $\nvote, \ncand$, and the scoring vector $\score$.
\end{proof}

It is worth pointing out two shortcomings of \Cref{thm:pos-lower}: First, the second term inside the $\min$ in $\lowerpos{\ell}$ does not exactly match the corresponding term inside $\upperpos{\ell}$ from \Cref{thm:pos-upper}. Second, unlike the minimal lower bound / divisibility requirements we have been making on $\nvote$ or $\ncand$ so far, requiring $\nvote \geq \frac{2\score_2}{1-\score_2}$ can significantly limit the ``feasible region'' of \Cref{thm:pos-lower} for certain rules $\pos$. For example if $\score_2\geq \frac{\ncand-2}{\ncand-1}$ (which is satisfied by $\borda$), this condition then requires $\nvote \geq 2(\ncand-2)$, whereas we would ideally like to find the positional scoring rule(s) with the lowest manipulation leverage regardless of whether we have more voters or candidates. Both of these limitations follow from the fact that constructing the ``worst case'' profile described in the proof sketch of \Cref{thm:pos-upper} is more difficult for some positional scoring rules than others, and in fact impossible for some rules (for example, requiring every candidate in $\{\cand{1},\cand{2},\ldots,\cand{\posbar{\ell}}\}$ to reach the score of $\cand{\ell+1}$ after the same number of truthful copies might require some of them to start from negative scores in the original profile, which is not possible). Overcoming these difficulties to get a tighter bound may require much more complex and unnatural restrictions on $\nvote,\ncand,$ and the scoring vector $\score$. In any case, as we show next, the lower bound established in \Cref{thm:pos-lower} is sufficient for our main goal in this subsection, which is to prove that the positional scoring rule with the smallest manipulation leverage will always be either Plurality or Borda Count, without any assumptions on whether there are more voters or candidates.

\begin{restatable}{theorem}{posbest}\label{thm:pos-best}
    For any positional scoring rule $\pos \notin \{\plurality,\borda\}$, given some $\nvote \geq 4$ that is divisible by 4, we have either $\spdeg{\pos} > \spdeg{\borda}$ or $\spdeg{\pos} \geq \spdeg{\plurality}$. Further, if $\nvote \geq \max(8, \frac{2}{\score_2}+4)$, then either  $\spdeg{\pos} > \spdeg{\borda}$ or $\spdeg{\pos} > \spdeg{\plurality}$ (\ie, the second inequality is also strict).
\end{restatable}
\begin{proof}
    Fix some positional scoring rule $\pos\notin \{\plurality,\borda\}$. If $\score_1=\score_2$, then $\spdeg{\pos}=\infty$ by \Cref{thm:pos-infinity}, and we are done. Otherwise, we will consider two cases based on the value of $\frac{2\score_2}{1-\score_2}$.

    \textbf{Case 1:} $\nvote < \frac{2\score_2}{1-\score_2}$. Then consider a profile $\profile$ where $\frac{\nvote}{4}$ voters (including voter 1) all rank $\cand{1} \succ \cand{2} \succ \cand{}$ for all $\cand{} \in \cands \setminus \{\cand{1},\cand{2}\}$, and the remaining $\frac{3\nvote}{4}$ voters and the tiebreaker $\tieorder$ all rank  $\cand{2} \succ \cand{1} \vote{} \cand{}$ for all $\cand{} \in \cands \setminus \{\cand{1},\cand{2}\}$. Fix any $0 \leq k \leq \frac{\nvote}{2}$ and for any $\cand{} \in \cands$ say $\bs_k(\cand{})$ is the total score assigned by $\pos$ to $\cand{}$ in $\kprofile{k}{1}$. Then $\bs_k(\cand{2})-\bs_k(\cand{1})=\left(\frac{3\nvote}{4}-\left(\frac{\nvote}{4}+k\right)\right)(1-\score_2)= \left(\frac{\nvote}{2}-k\right)(1-\score_2) \geq 0$. Similarly, $\bs_k(\cand{2}) \geq \bs_k(\cand{})$ for all $\cand{} \in \cands \setminus \{\cand{1},\cand{2}\}$ since all voters rank $\cand{2} \vote{} \cand{}$. As $\cand{2}$ is also ranked top by the tiebreaker $\tieorder$, this implies $\pos_\tiebreaker(\kprofile{k}{1})=\cand{2}$. 

    On the other hand, say voter 1 instead submits a preference order $\manip{1}$ that ranks $\cand{1}$ first and $\cand{2}$ last, and for any $\cand{} \in \cands$ say $\bs'(\cand{})$ is the total score assigned to $\cand{}$ by $\pos$ in $\maniprofile{1}$. Then $\bs'(\cand{1})-\bs'(\cand{2})=\left( \frac{\nvote}{4}+\frac{3\nvote}{4}\score_2\right)-\left( \left(\frac{\nvote}{4}-1\right)\score_2+\frac{3\nvote}{4}\right)=\score_2 -\frac{\nvote}{2}(1-\score_2)>0$. Further, for any $\cand{} \in \cands{}\setminus \{\cand{1},\cand{2}\}$, we have $\bs'(\cand{1}) \geq \bs'(\cand{})$ and $\cand{1} \tieorder \cand{}$, implying $\pos_\tiebreaker\maniprofile{1}=\cand{1} \vote{1} \cand{2} =\pos_\tiebreaker(\kprofile{k}{1})$. Therefore, $\pos$ cannot be $\ksp{k}$ for any $k \leq \frac{\nvote}{2}$, implying $\spdeg{\pos} > \frac{\nvote}{2} \geq \spdeg{\plurality}$, as desired.

    \textbf{Case 2:} $\nvote \geq \frac{2\score_2}{1-\score_2}$. In this case, we can invoke \Cref{thm:pos-lower}. We cannot have $\score_{i+1}\leq \score_{i} - \frac{1}{\ncand-1}$ for all $i \in [\ncand-1]$; if we did, this would imply $\score_i \leq \score_1-\frac{i-1}{\ncand-1}=\frac{\ncand-i}{\ncand-1}$ and $\score_i \geq \score_\ncand+\frac{\ncand-i}{\ncand-1}=\frac{\ncand-i}{\ncand-1}$ for all $i$, meaning $\pos= \borda$, which is a contradiction. Hence, let $\ell = \min\{ i\in [\ncand-1]: \score_{i}-\score_{i+1} < \frac{1}{\ncand-1} \}$. If $\ell=1$, \Cref{thm:pos-lower} would imply $\spdeg{\pos} \geq \lowerpos{1}= \frac{1}{\score_1-\score_2}-1>\ncand -2 = \spdeg{\borda}$, and we are done. Otherwise, say $\ell \geq 2$, and define
    \begin{align*}
        \lowerposB{\ell} &\coloneq  \frac{1}{\score_\ell-\score_{\ell+1}}-1\text{, and}\\
       \forall t \in [\posbar{\ell}]: \lowerposP{\ell, t}&\coloneq \frac{(\nvote-2)(\upave{2}-\score_{\ncand-t+1})+1-\score_{\ncand-t}}{\score_t-\score_{\ell+1}}.
    \end{align*}
    By assumption, we have $  \lowerposB{\ell} > \ncand-2=\spdeg{\borda}$. To complete the proof, we will show that for each $t \in [\posbar{\ell}]$ we have $\lowerposP{\ell,t} \geq \spdeg{\plurality}$ (and $ \lowerposP{\ell,t} \geq \spdeg{\plurality}+1$ if $\nvote \geq\max(8, \frac{2}{\score_2}+4)$). This is sufficient to prove the theorem statement since $\spdeg{\pos} \geq \lowerpos{\ell} = \min( \lowerposB{\ell}, \min_t \lfloor \lowerposP{\ell,t} \rfloor)$ by \Cref{thm:pos-lower}. By assumption, we have $\score_i - \score_{i+1} \geq \frac{1}{\ncand-1}$ for any $i < \ell$. Also recall $\posbar{\ell} \leq \ell -1$ by definition. Fix any $t \in [\posbar{\ell}]$, and consider three subcases.

    \textbf{Case 2a:} $t=1$. Then we have 
    \begin{align*}
        \lowerposP{\ell, 1} - \spdeg{\plurality} &\geq \lowerposP{\ell, 1} - \frac{\nvote}{2} = \frac{(\nvote-2)(\score_1+\score_2-2\score_{\ncand})+2-2\score_{\ncand-1}-\nvote(\score_1-\score_{\ell+1})}{2(\score_1-\score_{\ell+1})}  \\
        &= \frac{(\nvote-2)\score_2-2\score_{\ncand-1}+\nvote\score_{\ell+1}}{2(1-\score_{\ell+1}) } \geq  \frac{(\nvote-4)\score_2}{2}\geq 0 \quad\left(\text{and }\geq 1\text{ if } \nvote\geq \frac{2}{\score_2} +4\right).
    \end{align*}

    \textbf{Case 2b:} $t\geq 2$, $\ell \leq \ncand-t$. In this case we have $\score_{\ell} \geq \score_{\ell+1} \geq \score_{\ncand-t+1}$ and $\score_{\ell} \geq \score_{\ncand-t} \geq \score_{\ncand-t+1}$. Thus 
        \begin{align*}
        \lowerposP{\ell, t} - \spdeg{\plurality} \geq \lowerposP{\ell, t} - \frac{\nvote}{2} &= \frac{(\nvote-2)(\score_1+\score_2-2\score_{\ncand-t+1})+2-2\score_{\ncand-t}-\nvote(\score_t-\score_{\ell+1})}{2(\score_t-\score_{\ell+1})}\\
        &\geq \frac{(\nvote-2)(1+\score_2-2\score_{\ell+1})+2-2\score_{\ell}-\nvote(\score_2-\score_{\ell+1})}{2(\score_2-\score_{\ell+1})}\\&= \frac{\nvote-2\score_2 -(\nvote-4)\score_{\ell+1} -2\score_{\ell}}{2(\score_2-\score_{\ell+1})}\geq \frac{\nvote-2\score_2  -2\score_{\ell}}{2\score_2}\\& \geq \frac{\nvote(\ncand-1)-2(\ncand-\ell)}{2(\ncand-2)} -1 \geq \frac{4(\ncand-1)-2(\ncand-3)}{2(\ncand-2)} -1 \geq 0,
    \end{align*}
    since $\score_\ell \leq \score_1 -\frac{\ell-1}{\ncand-1}= \frac{\ncand-\ell}{\ncand-1}$,  $\score_2 \leq \score_1 -\frac{1}{\ncand-1}= \frac{\ncand-2}{\ncand-1}$, $\nvote \geq 4$, and $\ell \geq 3$ (as $\posbar{\ell} \geq t\geq 2$). Further, if $\nvote \geq 6$, we get a lower bound of 1 instead.

    \textbf{Case 2c:} $t\geq 2$, $\ell \geq \ncand-t+1$. Since $\score_i \leq \frac{\ncand-i}{\ncand-1}$ for all $i \leq \ell$ and since $\score_2-\score_{\ncand-t+1} \geq \frac{\ncand-t-1}{\ncand-1}$, we get
        \begin{align*}
        \lowerposP{\ell, t} - \spdeg{\plurality} \geq \lowerposP{\ell, t} - \frac{\nvote}{2} &= \frac{(\nvote-2)(\score_1+\score_2-2\score_{\ncand-t+1})+2-2\score_{\ncand-t}-\nvote(\score_t-\score_{\ell+1})}{2(\score_t-\score_{\ell+1})}\\
        &\geq \frac{\nvote+(\nvote-2)(\score_2-2\score_{\ncand-t+1})-2\score_{\ncand-t}-\nvote\score_t}{2(\score_t-\score_{\ell+1})}\\
        &\geq \frac{\nvote(\ncand-1)+(\nvote-2)( (\ncand-t-1) -(t-1))-2t-\nvote(\ncand-t)}{2(\score_t-\score_{\ell+1})(\ncand-1)}\\
        & = \frac{(\nvote-2)( \ncand-t-1)-2}{2(\score_t-\score_{\ell+1})(\ncand-1)} \geq 0,
    \end{align*}
    where the last inequality follows from the fact that $t \leq \posbar{\ell} \leq \ell-1 \leq \ncand-2$ and $\nvote \geq 4$, and therefore $(\nvote-2)( \ncand-t-1) \geq 2$. Further, if $\nvote \geq 8$, we get
    \begin{align*}
        \lowerposP{\ell, t} - \spdeg{\plurality} &\geq  \frac{(\nvote-2)( \ncand-t-1)-2}{2(\score_t-\score_{\ell+1})(\ncand-1)} \geq \frac{(\nvote-2)( \ncand-t-1)-2}{2\score_t(\ncand-1)}\\&\geq\frac{(\nvote-2)( \ncand-t-1)-2}{2(\ncand-t)} 
        \geq \frac{6( \ncand-t-1)-2}{2(\ncand-t)} = 3 -\frac{8}{2(\ncand-t)}\geq 1 
    \end{align*}
    since $\ncand-t \geq 2$, completing the proof for Case 2c.

    Hence, for each possible case, we have shown that $\lowerposP{\ell, t} \geq \spdeg{\plurality}$ (and $\lowerposP{\ell, t} \geq \spdeg{\plurality} + 1$ whenever $\nvote \geq \max(8, \frac{2}{\score_2}+4)$) for all $t \in [\posbar{\ell}]$, completing the proof of the theorem.
\end{proof}

Having shown that either Plurality or Borda Count (depending on $\nvote$ and $\ncand$) has the lowest manipulation leverage of an infinite-sized class of rules they belong to, we next move on to another class of rules---those that are \emph{Condorcet-consistent}---which contains neither $\plurality$ nor $\borda$. Below, we will show that, once again, any of the rules in this class will have higher manipulation leverage than one or both of $\plurality$ and $\borda$.
 
\subsection{Condorcet- (and Biranking-Majority-) Consistent Rules}\label{sec:condorcet}

In this section, we study the manipulation leverage of Condorcet extensions, as well as that of wider rule classes using weaker notions than Condorcet consistency. After formally introducing Condorcet consistency, we first analyze two example Condorcet-extensions: Black's Rule and Maximin.

We then prove general lower bounds for two broad classes of rules. We first consider \emph{majoritarian} rules, which determine the winner solely based on pairwise defeats;\footnote{Majoritarian rules are sometimes referred to as \emph{C1 functions}~\citep{Fishburn77:Condorcet}, whereas a neutral majoritarian rule is often called a \emph{tournament solution} under an odd number of voters (so that no margin is exactly zero)~\citep{Brandt:Barriers}.} indeed, the manipulability of majoritarian Condorcet extensions has been widely studied in prior literature (see, for example, \citep{brandt2015setmono,brandt2023characterizing}). We show that no majoritarian rule can have manipulation leverage lower than $\nvote-2$ under the minimal assumptions of neutrality and faithfulness (the winner of a single-voter profile is the top choice of that voter), the latter of which is much weaker than Condorcet consistency. For majoritarian Condorcet extensions, this nearly matches the upper bound of $\nvote-1$ from \Cref{remark:majority-consistent}. We also show that any neutral and nontrivial (\ie, it does not always return every candidate) majoritarian rule that fails faithfulness has unbounded manipulation leverage. We then turn to Condorcet-consistent rules without the majoritarian restriction and show that, when $\nvote$ is odd, their manipulation leverage is at least as high as that of Plurality. In fact, this lower bound holds for any rule satisfying a weaker form of majority consistency we call \emph{biranking majority consistency}, which only dictates the behavior of the rule on profiles with two voter types. Biranking-majority-consistent rules include all Condorcet extensions and, as we show, the infinite-sized class of runoff rules. Consequently, the manipulation leverage of Borda Count is significantly lower than that of these rules whenever $\nvote \gg \ncand$.

\paragraph{Condorcet Consistency} A \emph{Condorcet winner} is a candidate who pairwise defeats every other candidate. Formally, for a profile $\profile$, its \emph{margin matrix} $\margin{\profile}$ is a $\ncand \times \ncand$ matrix with its entry corresponding to candidates $\cand{},\cand{}'$ equal to $\margin{\profile}[\cand{},\cand{}']\coloneq |\{i \in \voters: \cand{} \vote{i} \cand{}'\}|- |\{i \in \voters: \cand{}' \vote{i} \cand{}\}|$. Then $\cand{}$ is the Condorcet winner if $\margin{\profile}[\cand{},\cand{}'] > 0$ for all $\cand{}' \in \cands \setminus \{\cand{}\}$. An SCC/SCF is \emph{Condorcet-consistent} (or a \emph{Condorcet extension}) if it only picks the Condorcet winner whenever one exists.

Any Condorcet extension is majority-consistent, since a majority winner is also a Condorcet winner. Consequently, the  upper bound of $\nvote-1$ from \Cref{remark:majority-consistent} also applies to Condorcet extensions. On the other hand, not every majority-consistent rule is a Condorcet extension, \eg, $\plurality, \irun,$ and $\plurWrunoff$ all fail Condorcet consistency. We now present two example Condorcet extensions.

\subsubsection{Black's Rule} 
As we have seen in \Cref{sec:warmup}, Borda Count ($\borda$) stands out among the rules we have studied so far by having manipulation leverage that does not grow with the number of voters, which, in practice, can be significantly larger than the number of alternatives. However, $\borda$ fails even majority-consistency. Could there be a Condorcet extension that shares the same manipulation leverage as $\borda$? A natural candidate is Black's Rule ($\black$), which is simply defined as follows: if there is a Condorcet winner, output that. Otherwise, output the Borda Count winner(s). However, as we show next, $\black$ has manipulation leverage as large as any Condorcet extension can have.

\begin{restatable}{theorem}{blackthm}\label{thm:black}
    The manipulation leverage of Black's Rule is $\spdeg{\black}=\nvote-1$ (lower bound assumes $\ncand \geq 4$ and odd $\nvote$). Further, $\black$ is not $\ksp{k}$ for any $0 \leq k < \nvote -1$.
\end{restatable}
\begin{proof}[Proof sketch]
    The proof can be found in \Cref{app:black} and uses the following profile:
        \begin{center}
        \begin{tabular}{c|l}
            \cellcola Voter 1 & \cellcola $\cand{1} \vote{1} \cand{2} \vote{1} \cand{3} \vote{1} \cdots$
            \\ \cellcolb $\frac{\nvote-1}{2}$ voters (Group 2) & \cellcolb $\cand{2} \vote{i} \cand{3} \vote{i} \cdots \vote{i} \cand{1}$
            \\\cellcola $\frac{\nvote-1}{2}$ voters (Group 3) &  \cellcola $\cand{3} \vote{i} \cdots \vote{i} \cand{1} \vote{i} \cand{2}$
        \end{tabular} 
    \end{center}
    Voter 1 can make $\cand{2}$ the Condorcet (and thus $\black$) winner by ranking it first. However, with $k < \nvote-1$ copies, there is no Condorcet winner, and the Borda Count (and thus $\black$) winner is $\cand{3}$.
\end{proof}

As \Cref{thm:black} shows, adding the conditional on Condorcet winner can be sufficient to go from manipulation leverage that grows linearly with $\ncand$ (Borda Count) to leverage that grows linearly with $\nvote$ instead (Black's Rule). Every rule we have studied so far has fallen into exactly one of these two categories (except for the general characterization of positional scoring rules in \Cref{sec:pos}). To show this is not always the case, we next present a Condorcet extension with manipulation leverage that depends on both $\ncand$ and $\nvote$.

\subsubsection{Maximin} The SCC known as Maximin\footnote{Maximin is also referred to as minimax (as it minimizes the worst defeat) or the Simpson-Kramer rule~\citep{brandt2025condorcet}.} is defined using the margin matrix $\margin{\profile}$ as \[\maximin(\profile)= \argmax_{\cand{} \in \cands} \min_{\cand{}' \in \cands} \margin{\profile}[\cand{},\cand{}'].\]
In words, Maximin scores candidates according to their lowest margin against other candidates and then outputs the candidate(s) with the highest score.

\begin{restatable}{theorem}{maximinthm}\label{thm:maximin}
    The manipulation leverage of Maximin is at most $\frac{(\ncand-2)\nvote}{\ncand-1}+2$ and (assuming $\nvote-2$ is divisible by $\ncand-1$ and $\ncand \geq 4$) at least $\frac{(\ncand-2)(\nvote-2)}{\ncand-1}+1$. Further, $\maximin$ is not $\ksp{k}$ for any $k<\frac{(\ncand-2)(\nvote-2)}{\ncand-1}+1$.
\end{restatable}
\begin{proof}[Proof sketch]
    The proof can be found in \Cref{app:maximin} and formalizes the ``worst-case'' for Maximin as a profile where the manipulator's bottom $\ncand-1$ candidates form a ``Condorcet cycle'' (\ie, each candidate pairwise defeats the next one in the cycle) of roughly equal edge weights, and the manipulator's top candidate is ranked bottom by almost every other voter. The manipulator can make her second-favorite candidate win by giving it a slight edge in the Condorcet cycle, whereas adding her truthful copies to the original profile increases multiple edges in the cycle by the same amount, so gives her no improvement (until her top choice eventually becomes the winner despite being overwhelmingly ranked bottom by other voters).
\end{proof}

\subsubsection{Lower Bounds for Majoritarian and Biranking-Majority-Consistent Rules}\label{sec:biranking-majoritarian}

Both of the example Condorcet extensions above have manipulation leverage that is at least as large as that of Plurality (for $\maximin$, note $\frac{\ncand-2}{\ncand-1} \geq \frac{1}{2}$) and thus larger than that of Borda Count whenever $\nvote \geq 2 \ncand$. To study this trend more broadly, we now turn our attention to general impossibility results. The main question we seek to answer in this subsection is: \emph{Can there be a Condorcet extension with lower manipulation leverage than (both) Plurality and Borda Count?} The answer, as we will show, is no.

We first study \emph{majoritarian} rules, which decide on a winner solely based on the majority defeats between pairs of candidates. Indeed, past literature has studied the manipulability of several majoritarian Condorcet extensions, such as the top cycle, uncovered set, and minimal covering set, which achieve weaker forms of strategyproofness when the voters are unaware of the tiebreaking mechanism being used and act conservatively~\citep{brandt2015setmono,brandt2023characterizing}. Nevertheless, we show that relying solely on pairwise majority defeats imposes strong limitations on how low the manipulation leverage can be, both for Condorcet extensions and for a much broader class of majoritarian rules.

Formally, an SCC $\scc$ is majoritarian if for any two profiles $\profile$, $\profile'$ such that $\sign\left(\margin{\profile}[\cand{},\cand{}']\right)=\sign\left(\margin{\profile'}[\cand{},\cand{}']\right)$ for all pairs of candidates $\cand{},\cand{}' \in \cands$, we have $\scc(\profile)=\scc(\profile')$.
We additionally introduce two minimal properties. An SCC $\scc$ is \emph{faithful} if when given a profile consisting of a single vote $\vote{1} \in \linprefs$, it outputs the top choice of that vote; \ie, $\scc(\vote{1})=\{\topv_{\vote{1}}\cands\}$~\citep{Young74:Axiomatization}. Indeed, faithfulness is implied by Condorcet consistency. Finally, we say an SCC $\scc$ is \emph{nontrivial} if for every electorate size $\nvote$, there is some profile $\profile \in \linprefs^\nvote$ on which it does not declare all candidates winners; \ie, $\scc(\profile) \neq \cands$.

For nontrivial and neutral majoritarian rules, \Cref{prop:majoritarian-nonfaithful} below shows that faithfulness is necessary for finite manipulation leverage. We then prove in \Cref{thm:majoritarian} that even under faithfulness, the manipulation leverage is at least $\nvote-2$. Together, these results show that the manipulation leverage of nontrivial neutral majoritarian rules (and thus, in particular, neutral majoritarian Condorcet extensions) cannot get much better than the worst case for Condorcet extensions given in \Cref{remark:majority-consistent}.
\color{black}

\begin{proposition}\label{prop:majoritarian-nonfaithful}
    Any nontrivial neutral majoritarian SCC $\scc$ that is not faithful will have $\spdeg{F}=\infty$ for any $\nvote \geq 3$.
\end{proposition}
\begin{proof}
    As always, fix the preference order of voter 1 to $\cand{1} \vote{1} \cand{2} \vote{1} \cdots \vote{1} \cand{\ncand}$.
    Since $\scc$ is not faithful and neutral, there exists some $\cand{}^*\in \scc(\vote{1}) \setminus \{c_1\}$ (By neutrality, $\scc$ violates faithfulness with some input ranking if and only if it violates it with each possible input ranking).  Set $\vote{2}$ to be the reverse of $\vote{1}$, and $\voters'=\voters \setminus [2]= \{3,\ldots,\nvote\}$. Since $\scc$ is nontrivial and neutral, there exists a profile $\vote{\voters'} \in \linprefs^{\nvote-2}$ with $\nvote-2$ voters such that $\cand{1} \in \scc(\vote{\voters'})$ and $\cand{}^* \notin \scc(\vote{\voters'})$ (Start from any $\nvote-2$ voter profile in which not every candidate wins; then permute the candidates so that $\cand{1}$ is among the winners and $\cand{}^*$ is among the losers). Set the tiebreaker so that $\cand{}^* \tieorder \cand{1} \tieorder \cand{}$ for all $\cand{} \in \cands \setminus \{\cand{1}, \cand{}^*\}$ and set $\profile= (\vote{1},\vote{2},\vote{\voters'})$. Since $\vote{1}$ and $\vote{2}$ are reverse orders, we have $\margin{\profile}=\margin{\vote{\voters'}}$ and therefore $\scc_\tiebreaker(\profile) = \scc_\tiebreaker(\vote{\voters'}) = \cand{1}$. On the other hand, for any $k\geq \nvote-1$, we have  $\sign\left(\margin{\kprofile{k}{1}}[\cand{},\cand{}']\right)=\sign\left(\margin{\vote{1}}[\cand{},\cand{}']\right)$ for all pairs of candidates $\cand{},\cand{}' \in \cands$. This implies $\scc_\tiebreaker(\kprofile{k}{1})=\cand{}^*$ since $\cand{}^*$ is in $\scc(\kprofile{k}{1})=\scc(\vote{1})$ and is ranked first by $\tieorder$. Setting $\manip{1}$ to be the same order as $\vote{1}$, we get $\scc_\tiebreaker\maniprofile{1}= \scc_\tiebreaker(\profile)= \cand{1} \vote{1} \cand{}^* =\scc_\tiebreaker(\kprofile{k}{1})$, showing $F$ is not $\ksp{k}$. Since $k$ can be arbitrarily large, this proves $\spdeg{\scc}=\infty$.
\end{proof}

Having established the manipulation leverage of neutral majoritarian rules that violate faithfulness, we next turn to those that satisfy this property.

\begin{theorem}\label{thm:majoritarian}
    Any faithful and neutral majoritarian SCC $\scc$ has $\spdeg{\scc} \geq \nvote-2$. Further, $\scc$ is not $\ksp{k}$ for any $k < \nvote-2$. 
\end{theorem}
\begin{proof}
    Fix any $1 \leq k < \nvote-2$ and an SCC $\scc$ satisfying the conditions in the theorem statement. We will show that $\scc$ is not $\ksp{k}$ (the claim for $k=0$ follows from the Gibbard-Satterthwaite theorem, as after fixing a tiebreaker $\tieorder$, $\scc_\tiebreaker$ is nondictatorial because majoritarianism implies anonymity, and it is onto because faithfulness and majoritarianism together ensure that if all voters submit the same vote, the top choice of that vote wins).
    Consider the following profile:
    \begin{center}
        \begin{tabular}{c|l}
            \cellcola Voter 1 & \cellcola $\bcand{1} \vote{1} \bcand{2} \vote{1} \bcand{3} \vote{1} \acand{4} \vote{1} \acdots \vote{1} \acand{\ncand}$
            \\ \cellcolb  $\left\lceil (\nvote-1)/2 \right\rceil$ voters (Group 1) & \cellcolb $\bcand{2} \vote{i} \bcand{3} \vote{i} \bcand{1} \vote{i} \acand{4} \vote{i} \acdots \vote{i} \acand{\ncand}$
            \\\cellcola $\left\lfloor (\nvote-1)/2 \right\rfloor$ voters (Group 2)&  \cellcola $\bcand{3} \vote{i} \bcand{1} \vote{i} \bcand{2} \vote{i} \acand{4} \vote{i} \acdots \vote{i} \acand{\ncand}$
        \end{tabular} 
    \end{center}
    Assume that the tiebreaker ranks $\bcand{3} \tieorder \bcand{2} \tieorder \bcand{1} \tieorder \acand{4} \tieorder \acdots \tieorder \acand{\ncand}$. Then we have
    \begin{align*}
        \margin{\kprofile{k}{1}}[\bcand{1}, \bcand{2}]&= (k+1) + \left\lfloor (\nvote-1)/2 \right\rfloor-\left\lceil (\nvote-1)/2 \right\rceil \geq 1,\\
        \margin{\kprofile{k}{1}}[\bcand{2}, \bcand{3}]&= (k+1) +\left\lceil (\nvote-1)/2 \right\rceil- \left\lfloor (\nvote-1)/2 \right\rfloor \geq 2\text{, and}\\
        \margin{\kprofile{k}{1}}[\bcand{3}, \bcand{1}]&= \left\lceil (\nvote-1)/2 \right\rceil+ \left\lfloor (\nvote-1)/2 \right\rfloor - (k+1) \geq 1.
    \end{align*}
    Thus, $\bcand{1}, \bcand{2}$, and $\bcand{3}$ form a 3-cycle in terms of majority defeats (\ie, a Condorcet cycle), and each of them defeats $\acand{\ell}$ for any $\ell \geq 4$. By neutrality and majoritarianism, this implies that $\{\bcand{1},\bcand{2},\bcand{3}\} \cap \scc(\kprofile{k}{1}) \neq \emptyset$ if and only if $\{\bcand{1},\bcand{2},\bcand{3}\} \subseteq \scc(\kprofile{k}{1})$, as we can permute these three candidates in a cycle without changing the graph of majority defeats. As $\tieorder$ ranks $\cand{3}$ first, this implies that we have $\scc_\tiebreaker(\kprofile{k}{1}) \notin \{\cand{1},\cand{2} \}$ and so $\cand{2} \vote{1} \scc_\tiebreaker(\kprofile{k}{1})$. Consider a manipulation where voter 1 also reports $\bcand{2} \manip{1} \bcand{3} \manip{1} \bcand{1} \manip{1} \acand{4} \manip{1} \acdots \manip{1} \acand{\ncand}$, joining Group 1. Since Group 1 now makes up strictly more than half of the electorate in $\maniprofile{1}$, we have $\sign\left(\margin{\maniprofilenp{1}}[\cand{},\cand{}']\right)=\sign\left(\margin{\manip{1}}[\cand{},\cand{}']\right)$ for all pairs of candidates $\cand{},\cand{}' \in \cands$. Since $\scc$ is faithful and majoritarian, $\scc_\tiebreaker\maniprofile{1}= \scc_\tiebreaker(\manip{1})= \cand{2} \vote{1} \scc_\tiebreaker(\kprofile{k}{1})$, completing the proof.
\end{proof}

Can we hope for much lower manipulation leverage if we remove the majoritarian requirement? After all, Plurality, which is not majoritarian, outperforms the lower bound of \Cref{thm:majoritarian} by a factor of (roughly) $\frac{1}{2}$.
As we will show in \Cref{thm:biranking} below, this is also the best we can achieve with any Condorcet extension, as the manipulation leverage of the entire class is lower-bounded by that of plurality. In fact, we prove this result under a significant relaxation of Condorcet consistency that we call \emph{biranking majority consistency}, which is weaker than majority consistency. To the best of our knowledge, this property has not been studied before.

\begin{definition}[Biranking-Majority-Consistent]\label{def:biranking-maj}
    Consider a profile consisting of only two types of rankings, \ie, there exist rankings $\vote{a},\vote{b} \in \linprefs$ such that for any voter $i \in \voters$, we have $\vote{i} \in \{\vote{a},\vote{b}\}$. We say an SCC/SCF is \emph{biranking-majority-consistent} if for any such profile where one of the rankings appears strictly more times than the other, it (only) picks the top choice of the more-frequent ranking.
\end{definition}

Majority consistency (and thus Condorcet consistency) implies biranking majority consistency (which in turn implies faithfulness). To demonstrate that biranking-majority-consistent rules are a notably broader class than those that are just majority-consistent, we turn our attention to another infinite-sized class, \emph{runoff rules}, which generalizes $\irun$.

\begin{definition}[Runoff rules~\citep{freeman2014axiomatic}]\label{def:runoff}
    A \emph{runoff rule} is an SCC $\runoff$ associated with a \emph{sequence} of scoring vectors $r=(\score^t)_{t \in \{2,\ldots,\ncand\}}$, where each $\score^t=(\score^t_j)_{j \in [t]}$ has length $t$ with $1=\score^t_1 \geq \score^t_2 \geq \cdots \geq \score^t_t=0$. On input profile $\profile$, the output of $\runoff$ is computed by iteratively eliminating candidates as follows. Whenever there are $t$ candidates remaining (for $t \in \{2,\ldots,\ncand\}$), each voter contributes $\score^t_j$ points to the candidate ranked $j^\text{th}$ among the remaining candidates. A candidate among those with the lowest total score is then eliminated. The final remaining candidate is the winner. As with other SCCs, ties are handled using parallel-universes tiebreaking: a candidate belongs to $\runoff(\profile)$ if there exists some way to break ties among the lowest-scoring candidates at each elimination step that leaves this candidate as the final winner.
\end{definition}

It is straightforward to see that Instant Runoff ($\irun$) is simply the runoff rule with $\score^t_j=\mathbb{I}[j=1]$ for all $t \in \{2,\ldots,\ncand\}$.\footnote{Despite its name, Plurality with Runoff ($\plurWrunoff$) is \emph{not} a member of runoff rules as defined in \Cref{def:runoff} for any $\ncand \geq 4$, since it eliminates $\ncand-2$ candidates in its first round rather than a single candidate.} We next show that \emph{all} runoff rules, regardless of $r$, are biranking-majority-consistent.

\begin{proposition}\label{prop:runoff-biranking}
    Any runoff rule is biranking-majority-consistent.
\end{proposition}
\begin{proof}
    Fix a runoff rule $\runoff$ with $r=(\score^t)_{t \in \{2,\ldots,\ncand\}}$ and consider a profile $\profile$ consisting of only two types of rankings, $\{\vote{a},\vote{b}\}$, such that $\ell > \frac{\nvote}{2}$ voters submit the ranking $\vote{a}$. Let $\cand{a}$ be the top choice of  $\vote{a}$. We will show that at each round of running $\runoff$, there will always be a candidate with a \emph{strictly} lower total score than $\cand{a}$; hence, $\cand{a}$ can never be eliminated, and we have $\runoff(\profile)=\{\cand{a}\}$. Say there are $t$ candidates remaining, and $\cand{a}'$ is the bottom choice of  $\vote{a}$ among them. Since $\cand{a}'$ gets $\score^t_t=0$ points from each of the $\ell$ voters ranking $\vote{a}$, its total score is at most $(\nvote-\ell)\score^t_1=\nvote-\ell$. Since $\cand{a}$ gets $\score^t_1=1$ points from each of the $\ell$ voters ranking $\vote{a}$, its total score is at least $\ell \score^t_1+(\nvote-\ell)\score^t_t=\ell$. As $\ell>\frac{\nvote}{2}$, we have $\ell> \nvote-\ell$, completing the proof.
\end{proof}
    
Consequently, all of our impossibility results below for biranking-majority-consistent rules apply to runoff rules. \Cref{prop:runoff-biranking} also demonstrates that biranking majority consistency is strictly weaker than majority (and therefore Condorcet) consistency: As an example, consider Coombs' Rule, which is the runoff rule under which, in each round, the candidate with the most bottom-choice votes is eliminated (\ie, $\score^t_j=\mathbb{I}[j\neq t]$ for all $t \in \{2,\ldots,\ncand\}$ and $j \in [t]$). Using this rule, even if a majority of voters agree on their top-choice candidate, this candidate can get eliminated if the same group of voters disagree about their bottom choices; however, this will not happen if they all submit the same ranking. Hence, Coombs' Rule is biranking-majority-consistent but not majority-consistent.

We are now ready to prove that the manipulation leverage of Plurality is essentially the best we can achieve with \emph{any} biranking-majority-consistent rule (and thus with any runoff rule or Condorcet extension).

\begin{restatable}{theorem}{birankingodd}\label{thm:biranking}
    When $\nvote$ is odd, any anonymous biranking-majority-consistent SCF $\scf$ has $\spdeg{\scf} \geq \frac{\nvote-1}{2}$.
\end{restatable}
\begin{proof}
    Fix some odd $\nvote$, and consider a profile with $\frac{3(\nvote-1)}{2}$ voters divided into three groups:
    \begin{center}
        \begin{tabular}{c|l}
            \multicolumn{2}{c}{\cellcolor{green!10} Profile 1}\\
            \cellcola $\frac{\nvote-1}{2}$ voters (Group 1) & \cellcola $\bcand{1} \vote{i} \bcand{2} \vote{i} \bcand{3} \vote{i} \acand{4} \vote{i} \acdots \vote{i} \acand{\ncand}$
            \\ \cellcolb $\frac{\nvote-1}{2}$ voters (Group 2) & \cellcolb $\bcand{2} \vote{i} \bcand{3} \vote{i} \bcand{1} \vote{i} \acand{4} \vote{i} \acdots \vote{i} \acand{\ncand}$
            \\\cellcola $\frac{\nvote-1}{2}$ voters (Group 3) &  \cellcola $\bcand{3} \vote{i} \bcand{1} \vote{i} \bcand{2} \vote{i} \acand{4} \vote{i} \acdots \vote{i} \acand{\ncand}$
        \end{tabular} 
    \end{center}
    Since $\scf$ will output only one winner, at least one out of the three groups of voters will get neither their first nor their second choice as the winner. WLOG, say this is Group 1, \ie, the winner that $\scf$ outputs on Profile 1 is not in $\{\cand{1},\cand{2}\}$. Then remove $k=\frac{\nvote-1}{2}-1$ copies of Group 1, in order to get a second profile:
    \begin{center}
        \begin{tabular}{c|l}
            \multicolumn{2}{c}{\cellcolor{green!10} Profile 2}\\
            \cellcola Voter 1 & \cellcola $\bcand{1} \vote{1} \bcand{2} \vote{1} \bcand{3} \vote{1} \acand{4} \vote{1} \acdots \vote{1} \acand{\ncand}$
            \\ \cellcolb $\frac{\nvote-1}{2}$ voters (Group 2) & \cellcolb $\bcand{2} \vote{i} \bcand{3} \vote{i} \bcand{1} \vote{i} \acand{4} \vote{i} \acdots \vote{i} \acand{\ncand}$
            \\\cellcola $\frac{\nvote-1}{2}$ voters (Group 3) &  \cellcola $\bcand{3} \vote{i} \bcand{1} \vote{i} \bcand{2} \vote{i} \acand{4} \vote{i} \acdots \vote{i} \acand{\ncand}$
        \end{tabular} 
    \end{center}
    If we label Profile 2 as $\profile$, then Profile 1 is just $\kprofile{k}{1}$. Consider a manipulation in $\profile$ (Profile 2) where voter 1 also reports the ranking of Group 2. Then, by biranking-majority-consistency, we have $\scf\maniprofile{1}= \cand{2} \vote{1} \scf(\kprofile{k}{1})$, showing that $\scf$ is not $\ksp{\left(\frac{\nvote-1}{2}-1\right)}$.
\end{proof}

Since a biranking-majority-consistent SCC equipped with any tiebreaker induces a biranking-majority-consistent SCF, \Cref{thm:biranking} also applies to multi-valued rules (\cf.\ discussion on ties in \Cref{sec:manipot}).

\begin{remark}\label{remark:odd-n}
    It is worth pointing out that the odd $\nvote$ assumption is quite essential to \Cref{thm:biranking}. This is unlike our previous results, in which divisibility assumptions were mostly employed to construct simple profiles achieving the desired lower bounds and it is likely possible to achieve similar bounds with minimal loss using a more cumbersome analysis with floor and ceiling functions (see, for example, \Cref{thm:plur,thm:majoritarian}). A simple adaptation of the proof of \Cref{thm:biranking} to even $\nvote$ would enable the manipulator to make her second-favorite candidate a \emph{weak Condorcet winner} (\ie, it defeats \emph{or ties with} any other candidate). Thus, a similar bound can be proven, regardless of the parity of $\nvote$, for rules that pick the unique weak Condorcet winner whenever one exists. We formalize this in \Cref{app:biranking-even}, where we also prove a slightly weaker lower bound of $\approx \frac{\nvote}{3}$ on the manipulation leverage of Condorcet extensions under even $\nvote$. Thus, the manipulation leverage of any Condorcet extension is $\Omega(\nvote)$ regardless of the parity of $\nvote$, whereas it is currently an open problem whether we can lower-bound the manipulation leverage of all biranking-majority-consistent rules for every even $\nvote$. The proof of \Cref{thm:biranking} also does not make any assumptions about the behavior of the rule when we have strictly less than $k=\frac{\nvote-1}{2}-1$ copies, since a violation of $\ksp{k}$ is sufficient to prove $\spdeg{\scf} \geq k+1$ (see the discussion on non-monotonicities in \Cref{sec:manipot}).
\end{remark}

\section{The Benefits of an $\nvote$-Independent \texttt{MLev}: A Game-Theoretic Perspective}
\label{sec:game-theory}

We now take a short intermission from analyzing the manipulation leverage of rules to emphasize the qualitative differences between the results from the preceding section. As summarized in \Cref{tab:rules-bounds}, all of the rules and rule classes we studied so far, with the exception of positional scoring rules with $\score_i \neq \score_j$ for all distinct $i,j \in [\ncand]$ (among which Borda Count has the smallest manipulation leverage), have a lower bound on their manipulation leverage that grows linearly with the number of voters,~$\nvote$. Previously, we have argued that a comparative disadvantage of rules whose manipulation leverage is $\Omega(\nvote)$ is the fact that $\nvote \gg \ncand$ in many standard settings. In this section, we will demonstrate that rules whose manipulation leverage is independent of $\nvote$ (such as that of Borda Count) have a further qualitative advantage, one regarding the game-theoretic \emph{stability} of voters choosing truthful augmentation over manipulation.

Recall from \Cref{sec:how-many} that our definition for manipulation leverage is motivated by information-scarce settings where voters lack any prior knowledge of others' preferences. In such settings, a voter can invest her resources either in manipulating (\eg, by gaining information on other voters and computing her optimal vote) or in increasing her voter weight (\eg, by increasing turnout from like-minded voters, or putting in money or effort to increase her own weight). So far, we have considered this choice for a single voter while holding the other votes fixed. Here, we formalize a game-theoretic setting in which all voters face this choice simultaneously. In this setting, an $\nvote$-independent bound on the manipulation leverage ensures voting truthfully is a weakly dominant strategy if voters are offered sufficiently many truthful copies (\Cref{prop:game-n-indep}), while we show that without such a bound, all voters voting truthfully may fail to be a Nash equilibrium regardless of the amount of truthful augmentation, even with just three voters (\Cref{prop:game-n-dep}).

\paragraph{Game Definition}
We introduce a family of normal-form games $\Gamma(\scf, \profile, k)$, parameterized by an anonymous SCF $\scf$, a profile $\profile$ (representing the true preferences of the voters), and a nonnegative integer $k$. The players are the voters, each of whom must choose between augmenting their truthful vote by $k$, or reporting a single untruthful vote instead. Formally, the action space of each voter is $\linprefs$; if each voter $i \in \voters$ plays the pure action $\vote{i}^*$ and we denote the resulting ``reported profile'' as $\profile^*$, the winner of the election is \[f\left( \profile^*+\sum_{i \in \voters: \vote{i}^*=\vote{i}}k(\vote{i}^*)\right),\]
\ie, all the voters that voted truthfully get $k$ additional votes.
The resulting payoff of each voter $i \in \voters$ is solely determined by the winner and is consistent with their true preferences $\vote{i}$.

The following is immediate from the definitions of augmentation strategyproofness (\Cref{def:kmanip}) and the manipulation leverage (\Cref{def:manipot}).

\begin{proposition}\label{prop:game-n-indep}
    Given an anonymous SCF $\scf$, say there exists some $K$ (that may depend on $\ncand$ but not $\nvote$) such that $\spdeg{f} \leq K$ for any electorate size $\nvote$. Then truthful voting is a weakly dominant strategy for all players in $\Gamma(\scf, \profile, K)$ for any profile $\profile$ with any number of voters; consequently, all voters voting truthfully is a Nash equilibrium.
\end{proposition}

The main idea is that for each $i \in \voters$, fixing the choices of all other voters, we can treat the resulting (possibly augmented) profile as $\vote{-i}$ and then use augmentation strategyproofness to show $\vote{i}$ weakly dominates any other choice.

In particular, the result applies to Borda Count with $K=\ncand-2$, since this bound holds for every electorate size and depends only on the number of candidates (which, unlike the number of votes, is unaffected by other voters choosing between manipulation and truthful augmentation).
The same cannot necessarily be said, however, for rules with a manipulation leverage that grows with $\nvote$: here, the choices of the other voters affect the effective number of votes in the resulting profile, and thus the augmentation needed to dissuade manipulation. As a result, regardless of how large $k$ is, it may not be enough to ensure truthful augmentation weakly dominates manipulation for a voter, since the same choice is offered to every other voter as well. In fact, no $k$ may suffice to even ensure the stability of the truthful outcome, as we show next.

\begin{proposition}\label{prop:game-n-dep}
    Let $f$ be any anonymous biranking-majority-consistent (\Cref{def:biranking-maj}) SCF; then there exists a profile $\profile$ with $\nvote=3$ where for \emph{any} $k \geq 0$, all voters voting truthfully is not a Nash equilibrium in $\Gamma(f,\profile,k)$.  
\end{proposition}

\begin{proof}
    Let the true preferences of the three voters be:
    \begin{center}
        \begin{tabular}{c|l}
            \cellcola Voter 1 & \cellcola $\bcand{1} \vote{1} \bcand{2} \vote{1} \bcand{3} \vote{1} \acand{4} \vote{1} \acdots \vote{1} \acand{\ncand}$
            \\ \cellcolb Voter 2 & \cellcolb $\bcand{2} \vote{2} \bcand{3} \vote{2} \bcand{1} \vote{2} \acand{4} \vote{2} \acdots \vote{2} \acand{\ncand}$
            \\\cellcola Voter 3 &  \cellcola $\bcand{3} \vote{3} \bcand{1} \vote{3} \bcand{2} \vote{3} \acand{4} \vote{3} \acdots \vote{3} \acand{\ncand}$
        \end{tabular} 
    \end{center}
   We show that all voters voting truthfully is not a Nash equilibrium by showing that at least one voter has a profitable unilateral deviation. If all three voters choose truthful augmentation ($\vote{i}^*=\vote{i}$ for all $i \in [3]$), the resulting profile (after augmentation) has $k+1$ copies of each vote. As in the proof of \Cref{thm:biranking}, at least one voter ranks the winner third or lower. WLOG, say this is voter 1, who we will show has a profitable deviation. Consider voter 1 reporting $\vote{1}^*=\vote{2}$ instead, giving up her $k$ additional truthful copies in the process. The new resulting profile (after augmentation) now contains only two rankings, $k+2$ copies of $\vote{2}$ and $k+1$ copies of $\vote{3}$. By biranking majority consistency, $\cand{2}$ therefore wins, making this a profitable deviation for voter 1.
\end{proof}

As \Cref{prop:game-n-indep,prop:game-n-dep} demonstrate, a key advantage of rules whose manipulation leverage is independent of $\nvote$ is that they remove the need to reason about the choices of other voters, even if augmentation is available to them too. Motivated by this, we next investigate how small $\nvote$-independent manipulation leverage can get. As we will see, minimal assumptions on the rule give us an impossibility result, which, in turn, proves the optimality of Borda Count in terms of low ($\nvote$-independent) manipulation leverage.
\section{A Universal Lower Bound on the Manipulation Leverage}\label{sec:universal}
Leaving the game theoretic discussion of \Cref{sec:game-theory} behind, we return to our main task of comparing the manipulation leverage of different rules. In \Cref{sec:rules}, we have shown that Borda Count achieves a manipulation leverage of $\ncand-2$ (\Cref{thm:borda}), which, under large electorates, outperforms several broad classes of rules, including positional scoring rules (\Cref{thm:pos-best}), majoritarian rules (\Cref{thm:majoritarian}), and biranking-majority-consistent rules (\Cref{thm:biranking}). Moreover, the game-theoretic framework of the preceding section reveals the advantages of having a manipulation leverage that is independent of $\nvote$, even for small electorates. These results motivate a natural next question: can we design a reasonable rule that achieves an $\nvote$-independent manipulation leverage lower than that of Borda Count? Better yet, is there any reasonable rule $\scf$ that achieves a \emph{constant} manipulation leverage, \ie, an $\spdeg{f}$ that is independent of both $\nvote$ and $\ncand$? In this section, we answer both of these questions in the negative.

The Gibbard-Satterthwaite theorem already shows that, for any fixed electorate size $\nvote$, an SCF $\scf$ with $\spdeg{f}=0$ that has at least three candidates in its range for $\nvote$-voter profiles must be dictatorial (\ie, on those profiles it always outputs a specific voter's top candidate from that range). 
These conditions (nondictatorship and a range of size at least three) alone are not sufficient to rule out constant manipulation leverage: an SCF that, on any input profile, eliminates all but three fixed candidates (\eg, $\{\cand{1},\cand{2},\cand{3}\}$) and runs Borda Count on the remaining profile achieves a manipulation leverage of 1, by the same proof as that of \Cref{thm:borda}.\footnote{More generally, the result on Borda Count shows that, for any $2 \leq t \leq m$, we can achieve a manipulation leverage of $t-2$ by ignoring all but $t$ candidates. At $t=2$, this reduces to running majority rule between two fixed candidates, which is strategy-proof.} Such a rule is clearly not that satisfactory: even if $\ncand$ grows, the \emph{effective} number of candidates remains three. We therefore turn to rules that are \emph{onto}: for each $\cand{} \in \cands$, there exists some profile in which $\cand{}$ wins.\footnote{While both \citep{gibbard1973manipulation} and \citep{satterthwaite1975strategy} focus on rules that have at least three candidates in their range, later work commonly presents their theorem as applying to onto rules with $\ncand \geq 3$ (\eg, \citep{Reny01:Arrow,Conitzer16:Barriers}), effectively treating $\ncand$ as the size of the range.} To stay consistent with our model involving changing electorates, we also assume \emph{anonymity} (all voters are treated equally), which implies non-dictatorship. These two minimal assumptions, as the next theorem shows, suffice to rule out a manipulation leverage lower than that of Borda Count, giving rise to our most general lower bound.

\begin{theorem}\label{thm:universal} Let $\scf$ be an anonymous and onto social choice function. If there exists an integer $K \geq 0$, possibly depending on $\ncand$ but independent of $\nvote$, such that $\spdeg{f} \leq K$ uniformly over all electorate sizes, then $K \geq \ncand-2$.
\end{theorem}

Before proceeding with the proof of \Cref{thm:universal}, it is worth briefly discussing its implications. First, it firmly establishes the special place occupied by Borda Count in the hierarchy induced by our metric, showing that it has the lowest possible $\nvote$-independent manipulation leverage among not just positional scoring rules, but essentially all reasonable rules. Second, since Borda Count exactly achieves the bound from \Cref{thm:universal}, we immediately recover that the bound is tight. Lastly, \Cref{thm:universal} can be seen as a generalization of the Gibbard-Satterthwaite theorem (with the caveat that non-dictatorship is strengthened to anonymity, which we assume in our variable-electorate model): it shows that any onto rule treating all voters equally will induce incentives for misreporting preferences, \emph{even if} the voter is offered an additional $k$ truthful votes for some $k \geq \ncand-3$ as an alternative to manipulating. As such, \Cref{thm:universal} can be thought of as a universal lower bound on the (worst-case) strategic advantage of a manipulator under any reasonable rule.

The rest of this section is dedicated to the proof of \Cref{thm:universal}. On a high level, we will first ``stretch'' $\scf$ into a rule that is defined on rational profiles that inherits a weighted form of its manipulation-leverage bound, while also exhibiting additional properties like group faithfulness and conditional participation. By studying families of cyclic profiles and applying Sperner's lemma, we then construct a Lipschitz family of \emph{common junctions}: points approached arbitrarily closely by rational profiles electing each candidate. The local geometry of this family translates the manipulation-leverage bound on $\scf$ into linear constraints on pairwise comparison vectors lying in a common $(\ncand-1)$-dimensional space. Finally, a linear dependence among these comparison vectors around a cycle of $\ncand$ candidates forces $K \geq \ncand-2$. For reference, the main notation used in the proof is gathered in \Cref{tab:universal-notation} at the end of the section.

\paragraph{From integer to rational profiles} So far, we have focused on social choice functions that are defined for any integer number of voters. For the proof of \Cref{thm:universal}, we will consider a more general class of rules that are defined on \emph{rational profiles}, where the multiplicity of each ranking is a nonnegative rational number. Formally, we represent a rational profile $P \in \mathbb{Q}_{\geq 0}^{\ncand!} \setminus \{0\}$ as a vector $(P_r)_{r \in \linprefs}$, where $P_r$ is the multiplicity of ranking $r$ in the profile (the rules we consider all satisfy anonymity; therefore we only need the coefficient of each ranking). By a slight abuse of notation, we also use each ranking $r \in \linprefs$ to denote its corresponding unit vector $\vecr{r}$.
While $\scf$ from \Cref{thm:universal} only needs to be defined for integer profiles, as the next lemma shows, it can be used to construct a rule $\flam$ on rational profiles that inherits (a weighted version of) its manipulation leverage bound, while also exhibiting other useful properties. Intuitively, we will ``stretch'' $\scf$ into $\flam$ so that for any input rational profile $P$, we have $g(P) \coloneq \scf(tP)$ for some large $t$.
For the rest of the proof, we let $\lambda \coloneq K+1$.

\begin{lemma}[Axioms for $\flam$]\label{lemma:flam-properties} Under the assumptions of \Cref{thm:universal}, there exists a rule $\flam: \mathbb{Q}_{\geq 0}^{\ncand!} \setminus \{0\} \to \cands$ satisfying
    \begin{itemize}
        \item \textit{(Weighted) Faithfulness:} For any ranking $r \in \linprefs$ and any $a \in \mathbb{Q}_{>0}$, we have $\flam(ar)=\topv_{r} \cands$.
        
        \item \textit{Bounded Manipulation Leverage:} For any profile $P \in \mathbb{Q}_{\geq 0}^{\ncand!}$ (referred to as the ``background profile''), any two rankings $r$ and $s$, and any positive rationals $a$ and $b$ with $a \geq \lambda b$, we have
        \begin{align}\label{eq:bounded-manip-potential}
            \flam(P + ar) \voteq{r} \flam(P + bs).
        \end{align}

        \item \textit{Conditional Participation:} For any profile $P$, any ranking $r \in \linprefs$, and any positive rational $a>0$, if $P_r>0$ then $\flam(P + ar) \voteq{r} \flam(P)$.
    \end{itemize}
\end{lemma}
\begin{proof}
    For any \emph{integer} profile $P \in \mathbb{Z}_{\geq 0}^{\ncand!}$, rankings $r,s \in \linprefs$, and integer $t \geq \lambda=K+1$, the bound $\spdeg{\scf} \leq K$ (for all $\nvote$) implies
    \begin{align}
        \scf(P+tr) \voteq{r} \scf(P+s)\label{eq:og-asp}
    \end{align}
    by \Cref{def:kmanip,def:manipot}. More generally, for any positive integers $a$ and $b$ with $a \geq \lambda b$, we have
    \begin{align}
        \scf(P+ar) \voteq{r} \scf(P+bs).\label{eq:integer-block}
    \end{align}
    To see this, start from $P+bs$ and replace the $b$ copies of $s$ one at a time, using $\lambda$ copies of $r$ for each of the first $b-1$ replacements and $a-\lambda(b-1) \geq \lambda$ copies for the last. Each replacement moves the winner weakly up in $r$ by \eqref{eq:og-asp}, and transitivity gives \eqref{eq:integer-block}.

    We now construct $\flam$ from $\scf$ by choosing a common sequence of integer scalings along which the winner eventually stabilizes for every rational profile. Enumerate the nonzero rational profiles as $P^{(1)},P^{(2)},\ldots$. For each fixed profile $P^{(i)}$, the profile $j!P^{(i)}$ is a nonzero integer profile for all sufficiently large integers $j$.

    Starting with $J_0=\mathbb{Z}_{>0}$, recursively choose an infinite subset $J_i \subseteq J_{i-1}$ on which $\scf(j!P^{(i)})$ is well-defined and constant. Such a subset exists because only finitely many indices give undefined evaluations (due to noninteger coefficients), and there are finitely many possible winners, so at least one candidate wins for infinitely many indices in $J_{i-1}$. Now choose $j_1<j_2<\cdots$ with $j_t \in J_t$. For every fixed $i$, all indices $j_t$ with $t \geq i$ belong to $J_i$, so $\scf(j_t!P^{(i)})$ is constant over this tail of the sequence. We therefore define
    \[
        \flam(P^{(i)}) \coloneq \scf(j_t!P^{(i)}) \qquad\text{for any } t \geq i.
    \]
    This defines $\flam$ on every nonzero rational profile, using the same sequence of scalings $(j_t!)_{t \geq 1}$ throughout. We now prove each property in the lemma statement separately.

    \textit{(Weighted) Faithfulness:} Fix a ranking $r$ and let $\candni=\topv_r \cands$. Since $\scf$ is onto, there exists an integer profile $P_{\candni}$ with $\scf(P_{\candni})=\candni$; let $\nvote_{\candni} \geq 1$ be its number of voters. Fix some integer $\nvote_r \geq \lambda\nvote_{\candni}$ and replace the original rankings in $P_{\candni}$ one at a time, using $\lambda$ copies of $r$ for each of the first $\nvote_{\candni}-1$ replacements and $\nvote_r-\lambda(\nvote_{\candni}-1) \geq \lambda$ copies of $r$ for the last. By \eqref{eq:og-asp}, each replacement preserves $\candni$ as the winner, since $\candni$ is the top candidate in $r$. Thus, $\scf(\nvote_r r)=\candni$ whenever $\nvote_r \geq \lambda\nvote_{\candni}$. For any positive rational $a$, the integer $j_t!a$ exceeds this threshold for all sufficiently large $t$, giving $\flam(ar)=\candni$.

    \textit{Bounded Manipulation Leverage:} Fix a rational profile $P \in \mathbb{Q}_{\geq 0}^{\ncand!}$, rankings $r$ and $s$, and positive rationals $a$ and $b$ with $a \geq \lambda b$. For all sufficiently large $t$, the profile $j_t!P$ and weights $j_t!a$ and $j_t!b$ are integral, so \eqref{eq:integer-block} gives
    \[
        \scf\bigl(j_t!(P+ar)\bigr) \voteq{r} \scf\bigl(j_t!(P+bs)\bigr).
    \]
    Taking the eventual values of both sides yields $\flam(P+ar) \voteq{r} \flam(P+bs)$.

    \textit{Conditional Participation:} Fix a rational profile $P$, a ranking $r$ with $P_r>0$, and a positive rational $a$. Choose a positive rational $b \leq P_r$ small enough that $a \geq (\lambda-1)b$. Applying bounded manipulation leverage \eqref{eq:bounded-manip-potential} with background profile $P-br$, both rankings equal to $r$, and weights $a+b \geq \lambda b$ and $b$ gives
    \[
        \flam(P+ar)=\flam\bigl((P-br)+(a+b)r\bigr) \voteq{r} \flam\bigl((P-br)+br\bigr)=\flam(P).
    \]
\end{proof}

\paragraph{Cyclic profiles} Having defined the rule $\flam$ on rational profiles, we now turn our attention to its behavior on two families of \emph{cyclic} profiles. We fix two sets of rankings, $\Rset=\{\rrank{j}\}_{j \in [\ncand]}$ and $\Lset=\{\lrank{j}\}_{j \in [\ncand]}$, so that for any index $j \in [\ncand]$, $\rrank{j}$ and $\lrank{j}$ rank
\begin{align*}
    \rrank{j} &: c_j \vote{} c_{j+1} \vote{} \cdots \succ c_{j-1},\\
    \lrank{j} &: c_j \succ c_{j-1} \vote{} \cdots \succ c_{j+1},
\end{align*}
where indices are taken modulo $\ncand$. In words, starting from the candidate $\cand{j}$, ranking $\rrank{j}$ (resp.\ $\lrank{j}$) ranks all candidates in increasing (resp.\ decreasing) index number, for some fixed labeling. If every ranking that appears with a positive coefficient in a profile $P$ is contained in
$\Rset$ (resp.\ in $\Lset$), we refer to $P$ as an $\Rset$-profile (resp.\ $\Lset$-profile). Given a nonzero vector of rational coefficients $x \in \mathbb{Q}^{\ncand}_{\geq 0}$, we write $x\Rset\coloneq\sum_{j=1}^m{x_j \rrank{j}}$ and $x\Lset\coloneq\sum_{j=1}^m{x_j \lrank{j}}$. We next show that on any profile with a sufficiently large cyclic component, $\flam$ picks among the top choices of the cyclic rankings present in the profile.

\begin{lemma}[Support of $\flam$ for ``sufficiently'' cyclic profiles]\label{lemma:support} Suppose a rational profile $P \geq0$ and a rational number $u_0\geq0$ satisfy $\flam(P+u \lrank{j})=\cand{j}$ for every $j \in [\ncand]$ and every rational $u>u_0$. Then for any nonzero rational coefficient vector $x \in \mathbb{Q}^{\ncand}_{\geq 0}$ with
$S=\supp(x)$, we have 
\[
 \norm{x}_1>\lambda^{|S|-1}u_0
 \quad\Longrightarrow\quad
 \flam\left(P+x\Lset\right)\in \{\cand{j}:j\in S\}.
\]
The same assertion holds with $\rrank{j}$ and $\Rset$ in place of $\lrank{j}$ and $\Lset$ throughout.

\end{lemma}
\begin{proof}
    We prove the claim for $\Lset$ by induction on $k=|S|$, the number of cyclic rankings with positive coefficients in $x\Lset$. If $k=1$, say $S=\{j\}$, then $\norm{x}_1=x_j>u_0$, so $\flam(P+x\Lset)=\flam(P+x_j\lrank{j})=\cand{j}$ by the assumption on $u_0$.

    For the inductive case, assume $k\geq 2$ and fix some $j \in [\ncand]\setminus S$ (if no such $j$ exists, we are done); we will show that $\flam(P+x\Lset)\neq\cand{j}$. Starting from $j-1$, keep decrementing the index by one (modulo $\ncand$) until reaching some $j'\in S$. By the choice of $j'$, ranking $\lrank{j'}$ places each candidate $\cand{h}$ with $h\in S\setminus\{j'\}$ above $\cand{j}$.

    Choose any $j''\in S\setminus\{j'\}$ and replace the block $x_{j'}\lrank{j'}$ with $(x_{j'}/\lambda)\lrank{j''}$. Write the resulting profile as $P+x'\Lset$, so that $x'\Lset=x\Lset-x_{j'}\lrank{j'}+\frac{x_{j'}}{\lambda}\lrank{j''}$. The new coefficient vector $x'$ has support $S\setminus\{j'\}$ and satisfies
    \[
        \norm{x'}_1=\norm{x}_1-x_{j'}+\frac{x_{j'}}{\lambda}
        \geq \frac{\norm{x}_1}{\lambda}>\lambda^{k-2}u_0.
    \]
    By the inductive hypothesis, $\flam(P+x'\Lset)$ has an index in $S\setminus\{j'\}$, so it is ranked above $\cand{j}$ in $\lrank{j'}$. Applying bounded manipulation leverage \eqref{eq:bounded-manip-potential} with background profile $P+x\Lset-x_{j'}\lrank{j'}$, rankings $\lrank{j'}$ and $\lrank{j''}$, and weights $x_{j'}$ and $x_{j'}/\lambda$ gives
    \[
        \flam(P+x\Lset)\voteq{\lrank{j'}}\flam(P+x'\Lset)\vote{\lrank{j'}}\cand{j}.
    \]
    Thus, $\flam(P+x\Lset)\neq\cand{j}$. Since this holds for every $j\notin S$, the proof for $\Lset$ is complete.

    The proof for $\Rset$ is analogous, following the cyclic order of $\Rset$-rankings instead: starting from $j+1$, keep incrementing the index until reaching the first member of $S$.
\end{proof}

In particular, setting $P=0$ (the zero profile) and $u_0=0$ satisfies the conditions of \Cref{lemma:support} by faithfulness (\Cref{lemma:flam-properties}). As a result, on any purely $\Lset$-profile or purely $\Rset$-profile, $\flam$ returns the top candidate of a ranking that is in the support of the profile.

\begin{corollary}[Support of $\flam$ on purely cyclic profiles]\label{cor:support}
    For any rational coefficients $x \in \mathbb{Q}_{\geq0}^\ncand \setminus \{0\}$ with $S = \supp(x)$, we have $\flam(x\Lset) \in \{\cand{j}:j\in S\}$ and $\flam(x\Rset)\in \{\cand{j}:j\in S\}$.
\end{corollary}

To prove our desired bound of $\lambda \geq \ncand-1$, we will pay special attention to \emph{binary} $\Lset$-profiles: $P=x_i\lrank{i}+x_j\lrank{j}$ for some $x_i,x_j \geq 0$ (with $x_i+x_j>0$) and $i,j \in [\ncand]$. By \Cref{cor:support}, we know $\flam(P) \in \{\cand{i},\cand{j}\}$. Further, any (possibly nonbinary) $\Lset$-profile or binary $\Rset$-profile can be ``compressed'' into a specific binary $\Lset$-profile with adjacent indices while preserving the winner, as formalized in the next lemma.

\begin{lemma}[Cyclic compressions]\label{lemma:cyclic-compress}
    Fix any $x \in \mathbb{Q}_{\geq 0}^\ncand \setminus \{0\}$ and write $X = \sum_{j=1}^\ncand x_j >0$. Then,
    \begin{align}
            \flam\left(x\Lset\right)=\cand{j}
            &\ \Longrightarrow\
            \flam\left(x_j\lrank{j}+\frac{X-x_j}{\lambda}\lrank{j+1}\right)=\cand{j},
            \label{eq:L-compress} \tag{$\Lset$-Compression}
    \end{align}
    Also, for distinct $i,j \in [\ncand]$ and positive rationals $x_i,x_j$, we have
    \begin{align}
            \flam\left(x_i\rrank{i}+x_j\rrank{j}\right)=\cand{i}
            &\ \Longrightarrow\
            \flam\left(\lambda x_i \lrank{i}+ \frac{x_j}{\lambda}\lrank{i-1}\right)=\cand{i}.
            \label{eq:R-compress}  \tag{$\Rset$-Compression}
    \end{align}
\end{lemma}

\begin{proof}
    We prove each compression separately, using bounded manipulation leverage \eqref{eq:bounded-manip-potential} and the support of $\flam$ on purely cyclic profiles (\Cref{cor:support}).

    \textit{$\Lset$-Compression:} By \Cref{cor:support}, $x_j>0$. Keep the block $x_j\lrank{j}$ fixed and process the other original blocks in the order $\lrank{j-1},\lrank{j-2},\ldots,\lrank{j+1}$. When processing each $h\neq j$ with $x_h>0$, replace the original block $x_h\lrank{h}$ with $(x_h/\lambda)\lrank{j+1}$, skipping the zero blocks. When processing the final index $h=j+1$, replace only the original $x_{j+1}$ copies of $\lrank{j+1}$, retaining all newly created rankings. 

    We show that each replacement preserves $\cand{j}$ as the winner. Let $P$ and $P'$ be the profiles before and after the replacement for index $h$, and suppose $\flam(P
    )=\cand{j}$. Every candidate below $\cand{j}$ in $\lrank{h}$ is the top candidate of an original block that has already been processed (and thus replaced). Hence, by \Cref{cor:support}, all possible winners of $P'$ other than $\cand{j}$, including $\cand{j+1}$ (the top candidate of the newly added rankings), are ranked above $\cand{j}$ in $\lrank{h}$. On the other hand, bounded manipulation leverage \eqref{eq:bounded-manip-potential} with background profile $P-x_h\lrank{h}$, rankings $\lrank{h}$ and $\lrank{j+1}$, and weights $x_h$ and $x_h/\lambda$ gives
    \[
        \cand{j}=\flam(P)\voteq{\lrank{h}}\flam(P').
    \]
    Thus, $\flam(P')=\cand{j}$. After all original blocks other than $x_j\lrank{j}$ have been processed, the resulting profile is $x_j\lrank{j}+((X-x_j)/\lambda)\lrank{j+1}$, proving \eqref{eq:L-compress}.

    \textit{$\Rset$-Compression:} First replace $x_i\rrank{i}$ with $\lambda x_i\lrank{i}$ to obtain the profile $Q=\lambda x_i\lrank{i}+x_j\rrank{j}$. Applying bounded manipulation leverage with background profile $x_j\rrank{j}$, rankings $\lrank{i}$ and $\rrank{i}$, and weights $\lambda x_i$ and $x_i$ gives
    \[
        \flam(Q)\voteq{\lrank{i}}\flam(x_i\rrank{i}+x_j\rrank{j})=\cand{i}.
    \]
    Since $\cand{i}$ is the top candidate in $\lrank{i}$, we have $\flam(Q)=\cand{i}$.

    Next, replace $x_j\rrank{j}$ with $(x_j/\lambda)\lrank{i-1}$ to obtain $Q'=\lambda x_i\lrank{i}+(x_j/\lambda)\lrank{i-1}$. Bounded manipulation leverage with background profile $\lambda x_i\lrank{i}$, rankings $\rrank{j}$ and $\lrank{i-1}$, and weights $x_j$ and $x_j/\lambda$ gives
    \[
        \cand{i}=\flam(Q)\voteq{\rrank{j}}\flam(Q').
    \]
    By \Cref{cor:support}, $\flam(Q')\in\{\cand{i},\cand{i-1}\}$. Since $j\neq i$, ranking $\rrank{j}$ places $\cand{i-1}$ above $\cand{i}$, so the displayed comparison rules out $\cand{i-1}$. Thus, $\flam(Q')=\cand{i}$, proving \eqref{eq:R-compress}.
\end{proof}

The two compressions from \Cref{lemma:cyclic-compress} allow us to strengthen \Cref{cor:support}, showing that in any purely $\Lset$-profile (binary or not), the ranking whose top candidate wins accounts for at least a certain positive fraction of the profile's total weight. This fraction depends only on $\lambda$ and $\ncand$, as the next lemma formalizes.

\begin{lemma}[Positive cyclic share of the winner ranking]\label{lemma:positive-share}
    Fix any $i \in [\ncand]$ and rationals $x_i>0, x_{i-1}\geq0$. If $\flam(x_i\lrank{i}+x_{i-1}\lrank{i-1})=\cand{i-1}$, then
    \begin{align}
        (\ncand-2)x_i\leq\lambda^5x_{i-1}.\label{eq:binary-ratio}
    \end{align}
    Consequently, for any rational coefficients $x \in \mathbb{Q}_{\geq0}^{\ncand}\setminus\{0\}$ with $X=\sum_{j=1}^{\ncand}x_j>0$, if $\flam(x\Lset)=\cand{i}$, then
    \begin{align}
        \frac{x_i}{X}\geq\delta\coloneq\frac{\ncand-2}{\lambda^6+\ncand-2}>0.\label{eq:share-bound}
    \end{align}
\end{lemma}

\begin{proof}
    We first prove \eqref{eq:binary-ratio}. By faithfulness (\Cref{lemma:flam-properties}), $x_{i-1}>0$, since otherwise $\flam(x_i\lrank{i})=\cand{i}$. For each $j\neq i$, consider the profile $(x_i/\lambda)\rrank{i}+\lambda x_{i-1}\rrank{j}$. By \Cref{cor:support}, its winner is either $\cand{i}$ or $\cand{j}$. If $\cand{i}$ wins, \eqref{eq:R-compress} gives $\flam(x_i\lrank{i}+x_{i-1}\lrank{i-1})=\cand{i}$, contradicting the assumption. Thus, $\cand{j}$ wins. Applying \eqref{eq:R-compress} again, with the two indices interchanged, then gives
    \begin{align}
        \flam\left(\lambda^2x_{i-1}\lrank{j}+\frac{x_i}{\lambda^2}\lrank{j-1}\right)=\cand{j}
        \qquad(\forall j\neq i).\label{eq:propagated-wins}
    \end{align}

    Now consider the $\Lset$-profile
    \[
        P=\frac{x_i}{\lambda^2}\sum_{j\neq i-1}\lrank{j},
    \]
    and let $\flam(P)=\cand{h}$. By \Cref{cor:support}, $h\neq i-1$. Applying \eqref{eq:L-compress} gives
    \[
        \flam\left(\frac{x_i}{\lambda^2}\lrank{h}+\frac{(\ncand-2)x_i}{\lambda^3}\lrank{h+1}\right)=\cand{h}.
    \]
    On the other hand, \eqref{eq:propagated-wins} with $j=h+1\neq i$ gives
    \[
        \flam\left(\frac{x_i}{\lambda^2}\lrank{h}+\lambda^2x_{i-1}\lrank{h+1}\right)=\cand{h+1}.
    \]
    If $(\ncand-2)x_i>\lambda^5x_{i-1}$, we can obtain the first of these two profiles from the second by adding $(\ncand-2)x_i/\lambda^3-\lambda^2x_{i-1}>0$ copies of $\lrank{h+1}$. Since this ranking already has positive weight and ranks $\cand{h+1}$ first, conditional participation (\Cref{lemma:flam-properties}) preserves $\cand{h+1}$ as the winner, a contradiction. This proves \eqref{eq:binary-ratio}.

    For \eqref{eq:share-bound}, fix $x$ as in the statement and suppose $\flam(x\Lset)=\cand{j}$. By \Cref{cor:support}, $x_j>0$. If $x_j=X$, the bound is immediate. Otherwise, $\Lset$-Compression gives
    \[
        \flam\left(x_j\lrank{j}+\frac{X-x_j}{\lambda}\lrank{j+1}\right)=\cand{j}.
    \]
    Applying \eqref{eq:binary-ratio} to this binary profile, with weights $(X-x_j)/\lambda>0$ and $x_j$, gives
    \[
        (\ncand-2)\frac{X-x_j}{\lambda}\leq\lambda^5x_j.
    \]
    Multiplying by $\lambda$ and rearranging gives $(\ncand-2)X\leq(\lambda^6+\ncand-2)x_j$, proving \eqref{eq:share-bound}.
\end{proof}

\Cref{lemma:positive-share} recovers the guarantee from the Gibbard-Satterthwaite theorem that $\lambda=1$ (\ie, $K=0$) is impossible: the equal-weight $\Lset$-profile has $x_i/X=1/\ncand$ for each $i \in [\ncand]$, while $\delta=(\ncand-2)/(\ncand-1) > 1/m$ for $\lambda=1$ and $\ncand \geq 3$, showing no candidate can satisfy the requirement of \eqref{eq:share-bound} from \Cref{lemma:positive-share}. Hence, for the rest of the proof, we can assume $\lambda >1$.

The $\Rset$-rankings will not play a separate role in the rest of the proof. As such, from now on, the \emph{cyclic coordinates} of a profile refer to the coefficients of the rankings in $\Lset$; the coefficients of all other rankings are its \emph{external coordinates}. WLOG, we reorder rankings in $\linprefs$ and write a profile as $P=(x,y)$, where $x \in \mathbb{R}_{\geq0}^{\ncand}$ records the coefficients of $\Lset=\{\lrank{j}\}_{j \in [\ncand]}$ (the cyclic coordinates), and $y \in \mathbb{R}_{\geq0}^{\ncand!-\ncand}$ records the coefficients of the rankings in $\linprefs\setminus\Lset$ (the external coordinates); \ie, $P=\sum_{j=1}^\ncand x_j \lrank{j}+\sum_{s \in \linprefs \setminus \Lset}y_ss$. We now extend the bound from \Cref{lemma:positive-share} to profiles that also contain external rankings.

\begin{corollary}[Positive cyclic share with external rankings]\label{cor:extended-share}
    For any rational profile $P=(x,y)$ with $X=\sum_{i=1}^{\ncand}x_i>0$, if $\flam(P)=\cand{j}$, then
    \begin{align}
        x_j\geq\delta X-\lambda(1-\delta)\norm{y}_1,\label{eq:extended-share}
    \end{align}
    where $\delta=\frac{\ncand-2}{\lambda^6+\ncand-2}$, as defined in \eqref{eq:share-bound} from \Cref{lemma:positive-share}.
\end{corollary}
\begin{proof}
    Replace each external block $y_s s$, with $y_s>0$ and $s\in\linprefs\setminus\Lset$, with $\lambda y_s$ copies of $\lrank{j}$, skipping zero blocks. We show that each replacement preserves $\cand{j}$ as the winner. Let $Q$ and $Q'$ be the profiles before and after one such replacement, and suppose $\flam(Q)=\cand{j}$. Applying bounded manipulation leverage \eqref{eq:bounded-manip-potential} with background profile $Q-y_s s$, rankings $\lrank{j}$ and $s$, and weights $\lambda y_s$ and $y_s$ gives
    \[
        \flam(Q')\voteq{\lrank{j}}\flam(Q)=\cand{j}.
    \]
    Since $\cand{j}$ is the top candidate in $\lrank{j}$, we have $\flam(Q')=\cand{j}$. Thus, after all external blocks have been replaced, the purely $\Lset$-profile $x\Lset+\lambda\norm{y}_1\lrank{j}$ still elects $\cand{j}$. Its total weight is $X+\lambda\norm{y}_1$, and the coefficient of $\lrank{j}$ is $x_j+\lambda\norm{y}_1$. Applying \eqref{eq:share-bound} from \Cref{lemma:positive-share} gives
    \[
        x_j+\lambda\norm{y}_1\geq\delta\bigl(X+\lambda\norm{y}_1\bigr).
    \]
    Rearranging yields \eqref{eq:extended-share}.
\end{proof}

\paragraph{Common junctions and Sperner’s lemma} A key concept that will be useful in the rest of our analysis is the notion of (possibly real-valued) points approached arbitrarily closely by rational profiles electing each candidate, which we refer to as \emph{common junctions}. Formally, let $\Omega \coloneq \R_{>0}^{m!}$ and for any $A \subseteq \Omega$ let $\cl_{\Omega}(A)$ denote the closure of $A$ in $\Omega$. For any $j \in [\ncand]$, define the closure of the rational profiles in which $\cand{j}$ is the winner as
\begin{equation}
 A_j \coloneq \cl_{\Omega}\{(x,y)\in\Omega\cap\Q^{m!}:\flam(x,y)=\cand{j}\}.
 \label{eq:closed-winners}
\end{equation}
$\{A_j\}_{j \in [\ncand]}$ covers $\Omega$, by rational approximation and a constant-winner subsequence. Inequality~\eqref{eq:extended-share} from \Cref{cor:extended-share} holds throughout $A_j$,
because it is continuous in the profile coordinates. A \emph{common
junction} is a point in $\bigcap_{j \in [\ncand]} A_j$.

Our exploration of common junctions will rely on \emph{Sperner’s Lemma}~\citep{Sperner28:Neuer}, where the labelings will correspond to membership in the sets $\{A_j\}_{j \in [\ncand]}$, and fully-labeled simplices will lead to a common junction.

We now formally present the $d$-dimensional Sperner’s Lemma~\citep[Section 2]{Su99:Rental}. Let $S$ be a $d$-simplex for some integer $d\geq1$, with its $d+1$ facets indexed by $[d+1]$ (here, a facet is a face spanned by $d$ of the $d+1$ vertices of $S$). A \emph{triangulation} of $S$ is a finite collection of smaller $d$-simplices, referred to as \emph{elementary simplices}, whose union is $S$, where any two distinct elementary simplices intersect in a common face or not at all. A labeling of the triangulation's vertices by $[d+1]$ is a \emph{Sperner labeling} if, for every $j\in[d+1]$, no vertex on facet $j$ is labeled $j$. An elementary simplex is \emph{fully labeled} if all its $d+1$ vertices have distinct labels.

\begin{figure}[t]
    \centering
    \begingroup
\definecolor{bordaOne}{HTML}{2462B4}
\definecolor{bordaTwo}{HTML}{BB6514}
\definecolor{bordaThree}{HTML}{8153A1}
\definecolor{bordaFillOne}{HTML}{DDEAFB}
\definecolor{bordaFillTwo}{HTML}{FFE7CD}
\definecolor{bordaFillThree}{HTML}{EADFF4}
\definecolor{bordaInk}{HTML}{233044}

\begin{tikzpicture}[
  x=1cm,y=1cm,
  font=\small,
  text=bordaInk,
  line cap=round,line join=round,
  winner/.style={font=\large\bfseries\boldmath},
  vertex label/.style={align=center,font=\small},
  edge label/.style={fill=none,inner sep=1.5pt,font=\scriptsize},
  tie edge/.style={line width=1.8pt},
  tie endpoint/.style={draw=white,line width=0.6pt}
]

\coordinate (L1)  at (3,{3*sqrt(3)});
\coordinate (L2)  at (0,0);
\coordinate (L3)  at (6,0);
\coordinate (O)   at (3,{sqrt(3)});
\coordinate (E12) at (1,{sqrt(3)});
\coordinate (E23) at (4,0);
\coordinate (E31) at (4,{2*sqrt(3)});

\fill[bordaFillOne]   (L1)--(E12)--(O)--(E31)--cycle;
\fill[bordaFillTwo]   (L2)--(E23)--(O)--(E12)--cycle;
\fill[bordaFillThree] (L3)--(E31)--(O)--(E23)--cycle;

\draw[tie edge,bordaOne] (E12)--(O)--(E31);
\draw[tie edge,bordaTwo] (O)--(E23);
\draw[line width=0.8pt,bordaInk] (L1)--(L2)--(L3)--cycle;

\filldraw[tie endpoint,fill=bordaOne] (E12) circle[radius=2.2pt];
\filldraw[tie endpoint,fill=bordaOne] (E31) circle[radius=2.2pt];
\filldraw[tie endpoint,fill=bordaTwo] (E23) circle[radius=2.2pt];

\node[winner,text=bordaOne]   at (2.62,3.20) {$c_1$ wins};
\node[winner,text=bordaTwo]   at (2.00,0.78) {$c_2$ wins};
\node[winner,text=bordaThree] at (4.18,1.45) {$c_3$ wins};

\node[vertex label,anchor=south] at (3,5.37)
  {$\ell_1:\ c_1\succ c_3\succ c_2$\\[2pt]
   {\scriptsize $x_1=1$}};
\node[vertex label,anchor=north] at (0,-0.19)
  {{\scriptsize $x_2=1$}\\[2pt]
   $\ell_2:\ c_2\succ c_1\succ c_3$};
\node[vertex label,anchor=north] at (6,-0.19)
  {{\scriptsize $x_3=1$}\\[2pt]
   $\ell_3:\ c_3\succ c_2\succ c_1$};

\node[edge label,text=bordaOne] at (2.00,{sqrt(3)+0.17})
  {$x_1=1/3$};
\node[edge label,text=bordaOne,rotate=60] at (3.33,2.70)
  {$x_3=1/3$};
\node[edge label,text=bordaTwo,rotate=-60] at (3.30,0.79)
  {$x_2=1/3$};

\filldraw[fill=bordaOne,draw=white,line width=1.1pt]
  (O) circle[radius=3.6pt];
\node[anchor=west,font=\normalsize,inner sep=1pt,fill=none]
  at ([xshift=6pt]O) {$x^*$};

\end{tikzpicture}
\endgroup
    \caption{Borda Count on $\Delta_1$ (the simplex over cyclic $\Lset$ profiles, as defined in \Cref{ex:sperner-lset}) for $\ncand=3$, with tie-breaker $\cand{1}\tieorder\cand{2}\tieorder\cand{3}$. Each rational point on the triangle represents a $\Lset$-profile. Shading indicates the winning candidate, and winners on shared boundaries are determined by the tie-breaker. The point $x^*=(1/3,1/3,1/3)$ lies in the closure of all three winner regions, illustrating a cyclic analogue of a common junction. }
    \label{fig:borda-junction}
\end{figure}

\begin{lemma}[Sperner's lemma]\label{lemma:sperner}
    Every Sperner-labeled triangulation of a $d$-simplex contains an odd number of fully labeled elementary $d$-simplices. In particular, there is at least one~\citep[Section~2]{Su99:Rental}.
\end{lemma}

As a warm-up for the rest of the proof, the next example shows how Sperner's lemma can be applied in conjunction with the lemmas proven above to the restricted domain of $\Lset$-profiles.

\begin{example}[Sperner's lemma on $\Lset$-profiles]\label{ex:sperner-lset}
     For any fixed rational $X>0$ consider the $(\ncand-1)$-dimensional simplex $\Delta_X=\{x\in\R_{\geq0}^{\ncand}: \sum_{j=1}^\ncand x_j =X\}$, representing the closure of the set of $\Lset$-profiles with total weight $X$. Triangulate $\Delta_X$ with rational vertices, labeling each vertex $x$ by $j$ whenever $\flam(x\Lset)=\cand{j}$. By the positive cyclic share of winners in $\Lset$-profiles (\Cref{lemma:positive-share}), a vertex labeled $j$ satisfies $x_j\geq\delta X>0$, so no vertex on the facet $x_j=0$ receives label $j$. Applying Sperner's lemma (\Cref{lemma:sperner}) to a sequence of triangulations whose maximum simplex diameter tends to zero gives a sequence of fully labeled and increasingly small elementary simplices. By compactness, along a subsequence their vertices converge to the same point $x^*\in\Delta_X$, which is thus approached by profiles electing each candidate. In other words, $x^* \in \bigcap_{j \in [\ncand]} \cl_{\Delta_X}\{x \in \Delta_X \cap\Q^{m}: \flam(x\Lset)=\cand{j}\}$. Passing \Cref{lemma:positive-share} to the limit gives $x_j^*\geq\delta X>0$ for every $j\in[\ncand]$, placing $x^*$ in the relative interior of $\Delta_X$. Thus, the closure of the set of \(\Lset\)-profiles with fixed total weight contains a cyclic analogue of a common junction; \Cref{fig:borda-junction} demonstrates this for Borda Count (which satisfies the assumptions of \Cref{thm:universal}) and $\ncand=3$.
\end{example}

While \Cref{ex:sperner-lset} illustrates the link between Sperner's lemma and the structure of common junctions, it is insufficient for our complete proof, due to two main reasons. First, $x^*$ in \Cref{ex:sperner-lset} is not exactly a common junction per the definition above, as its external coordinates are zero, and thus $(x^*,0)\notin\Omega =\mathbb{R}^{\ncand!}_{>0}$. Our proof will require common junctions to be full-support (\ie, the coefficient of every ranking is positive), not just on $\Lset$-profiles. Second, as we will see, our argument in fact requires \emph{a Lipschitz family of common junctions}, rather than a single one, the existence of which we will establish in \Cref{junc:lipschitz} below.

Returning to our main argument, we now study the behavior of common junctions in $\Omega=\mathbb{R}_{>0}^{\ncand!}$. Our next lemma shows that if we slightly reduce every ranking's weight at a common junction, followed by a sufficiently large increase in the cyclic coordinates, the winner of any rational profile close to the new point must belong to the support of the cyclic increment.

\begin{lemma}[Support of perturbed common junctions]
\label{junc:lower} Fix $\kappa \coloneqq 2m!\lambda^m$. Let $P$ be a common junction and
$0<\varepsilon<\frac12\min_{r\in  \linprefs}P_r$. If $v \in \mathbb{R}_{\geq 0}^{\ncand}$ is a (possibly irrational) vector of cyclic
coordinates with nonempty proper support $S=\{j \in [\ncand]:v_j>0\}\subsetneq [\ncand]$ and
$\norm{v}_1>\kappa\varepsilon$, then every rational profile $Q$ satisfying
\begin{equation}
 \norm{Q-(P-\varepsilon\one+v\Lset)}_\infty<\frac{\varepsilon}{4}
 \label{junc:lower-neighborhood}
\end{equation}
has $\flam(Q)\in\{\cand{j}:j\in S\}$, where $v\Lset=\sum_{j=1}^{\ncand}v_j\lrank{j}$ and $\one \in \mathbb{R}^{\ncand!}$ is the all-ones vector.
\end{lemma}
\begin{proof}
    First suppose $v$ is rational, and fix a rational profile $Q$ satisfying \eqref{junc:lower-neighborhood}. Let $Z=Q-v\Lset$. Since both $Q$ and $v$ are rational, so is $Z$, and
    \[
        \norm{Z-(P-\varepsilon\one)}_\infty<\frac{\varepsilon}{4}.
    \]
    In particular, every coordinate of $Z$ is positive, since $P_r>2\varepsilon$ for every ranking $r\in\linprefs$.

    For each $j\in[\ncand]$, we have $P\in A_j$ since $P$ is a common junction. By the definition of $A_j$ in \eqref{eq:closed-winners}, we can choose a rational profile $W^{(j)}\in\Omega$ with $\flam(W^{(j)})=\cand{j}$ and $\norm{W^{(j)}-P}_\infty<\varepsilon/4$. For every ranking $r\in\linprefs$, the surplus of $r$ in $W^{(j)}$ relative to $Z$ is $W^{(j)}_r-Z_r=\varepsilon+(W^{(j)}_r-P_r)-\bigl(Z_r-(P_r-\varepsilon)\bigr)$. By the approximation bounds, we have
    \[
        \frac{\varepsilon}{2}<W^{(j)}_r-Z_r<\frac{3\varepsilon}{2}.
    \]
    Starting from $W^{(j)}$, replace each surplus block $(W^{(j)}_r-Z_r)r$ with $\lambda(W^{(j)}_r-Z_r)$ copies of $\lrank{j}$. Each replacement preserves $\cand{j}$ as the winner by bounded manipulation leverage \eqref{eq:bounded-manip-potential}, since $\lrank{j}$ ranks $\cand{j}$ first. After all replacements, the resulting profile is $Z+u_j\lrank{j}$, where
    \[
        u_j=\lambda\sum_{r\in\linprefs}(W^{(j)}_r-Z_r)
        <\frac32\lambda\ncand!\varepsilon,
    \]
    and hence $\flam(Z+u_j\lrank{j})=\cand{j}$.

    Let $u_0=\max_{j\in[\ncand]}u_j$, which is rational and satisfies $u_0<2\lambda\ncand!\varepsilon$. Since $Z_{\lrank{j}}>0$, conditional participation (\Cref{lemma:flam-properties}) and $\flam(Z+u_j\lrank{j})=\cand{j}$ imply
    \[
        \flam(Z+u\lrank{j})=\cand{j}
        \qquad\text{for every $j\in[\ncand]$ and every rational $u>u_0$.}
    \]
    Finally, since $S\subsetneq[\ncand]$, we have
    \[
        \norm{v}_1>\kappa\varepsilon=2\ncand!\lambda^\ncand\varepsilon
        \geq2\ncand!\lambda^{|S|}\varepsilon
        >\lambda^{|S|-1}u_0.
    \]
    Thus, we can apply \Cref{lemma:support} with background profile $Z$ and cyclic coefficient vector $v$, which gives
    \[
        \flam(Q)=\flam(Z+v\Lset)\in\{\cand{j}:j\in S\}.
    \]
    This proves the claim when $v$ is rational.

    For a general real vector $v$, fix a rational profile $Q$ satisfying \eqref{junc:lower-neighborhood}. Since both $\norm{v}_1>\kappa\varepsilon$ and \eqref{junc:lower-neighborhood} are strict inequalities, we can choose a rational vector $v'\in\mathbb{Q}_{\geq0}^{\ncand}$ with the same support $S$ such that
    \[
        \norm{v'}_1>\kappa\varepsilon
        \qquad\text{and}\qquad
        \norm{Q-(P-\varepsilon\one+v'\Lset)}_\infty<\frac{\varepsilon}{4}.
    \]
    Applying the rational case to $v'$ gives $\flam(Q)\in\{\cand{j}:j\in S\}$, completing the proof.
\end{proof}

We now have the necessary toolkit to prove the existence of a Lipschitz continuous function with a range that consists of common junctions, in particular via an application of Sperner's lemma (\Cref{lemma:sperner}).

\begin{proposition}[A unique Lipschitz family of common junctions]
\label{junc:lipschitz}
There is a nonempty open box $U\subset\R^{\ncand!-\ncand+1}$ and a Lipschitz map
\[
 G:U\longrightarrow\Omega,\qquad
 G(y,\tau)=(H(y,\tau),y),\qquad
 \sum_{j=1}^{\ncand}H_j(y,\tau)=\tau,
\]
such that the following conditions are all satisfied: For every $(y,\tau)\in U$, $G(y,\tau)$ is the unique common junction $(x,y)$ satisfying $\sum_{j=1}^{\ncand}x_j=\tau$. Every coordinate of $G(y,\tau)$ exceeds $1$. At each point $z\in U$ where $G$ is differentiable, its Jacobian $J_G(z)$ has rank $\ncand!-\ncand+1$.
\end{proposition}
\begin{proof}
    \textit{Existence:} Choose $T>\ncand$ sufficiently large that $\delta T-2\lambda(\ncand!-\ncand)>1$, and let
    \begin{equation}
        U=(1,1+1/32)^{\ncand!-\ncand}\times(T,T+1).
        \label{junc:parameter-domain}
    \end{equation}
    Fix $(y,\tau)\in U$. For any real-valued profile $(x,y)\in A_j$ with $\sum_{i=1}^{\ncand}x_i=\tau$,  \Cref{cor:extended-share} (which, as noted above, extends to $A_j$ by the continuity of inequality~\eqref{eq:extended-share} and the definition of $A_j$ in \eqref{eq:closed-winners} as the closure of rational profiles electing $\cand{j}$) gives
    \begin{align}
        x_j\geq\delta\tau-\lambda(1-\delta)\norm{y}_1
        >\delta T-2\lambda(\ncand!-\ncand)>1. \label{eq:junction-bound}
    \end{align}
    In particular, every common junction with these parameters has all cyclic coordinates greater than $1$. Its external coordinates also exceed $1$, by the definition of $U$.

    To show that such a common junction exists, consider the $(\ncand-1)$-dimensional simplex
    \[
        \Delta(y,\tau)=\left\{(x,y):x_j\geq1\text{ for every }j\in[\ncand],\ \sum_{j=1}^{\ncand}x_j=\tau\right\}\subset\Omega.
    \]
    This simplex is nonempty since $\tau>\ncand$. Index its facets by $F_j=\{(x,y)\in\Delta(y,\tau):x_j=1\}$ for $j\in[\ncand]$. The preceding bound in \eqref{eq:junction-bound} gives $F_j\cap A_j=\varnothing$ for every $j$.

    Take a sequence of triangulations of $\Delta(y,\tau)$ whose maximum elementary simplex diameter tends to zero. Label each vertex $Q$ by any $j\in[\ncand]$ such that $Q\in A_j$; such a label exists because the sets $\{A_j\}_{j\in[\ncand]}$ cover $\Omega$. No vertex on $F_j$ receives label $j$, so Sperner’s lemma (\Cref{lemma:sperner}) gives a fully labeled elementary simplex in each triangulation.

    Let $Q^{(t,j)}\in A_j$ be the vertex labeled $j$ in the chosen fully labeled simplex of the $t$-th triangulation. By compactness, along a subsequence $Q^{(t,1)}$ converges to a point $P\in\Delta(y,\tau)$. Since the simplex diameters tend to zero, $Q^{(t,j)}$ converges to the same point $P$ for every $j\in[\ncand]$. Each $A_j$ is closed in $\Omega$, and $P\in\Delta(y,\tau)\subset\Omega$, so $P\in\bigcap_{j\in[\ncand]}A_j$. This proves existence of a common junction for every $(y,\tau)\in U$.

    \textit{Uniqueness and Lipschitz continuity:} Take any two common junctions $P=(x,y)$ and $P'=(x',y')$ whose parameters $(y,\tau)$ and $(y',\sigma)$ belong to $U$, where $\sum_{j=1}^{\ncand}x_j=\tau$ and $\sum_{j=1}^{\ncand}x'_j=\sigma$. Write $h=x'-x$, and define $h_j^+=\max\{h_j,0\}$ and $h_j^-=\max\{-h_j,0\}$. Thus, $h=h^+-h^-$ and the supports of $h^+$ and $h^-$ are disjoint. Let $R=\norm{y-y'}_\infty<1/32$.

    For any $4R<\varepsilon<1/2$, define
    \[
        Z=(\max\{x,x'\},\max\{y,y'\})-\varepsilon\one,
    \]
    where the maxima are taken coordinatewise. Since every coordinate of $P$ and $P'$ exceeds $1$, we have $Z>0$. Its cyclic coordinates agree with those of both $P-\varepsilon\one+h^+\Lset$ and $P'-\varepsilon\one+h^-\Lset$, and its external coordinates differ from those of either of these two profiles by at most $R$. Hence,
    \[
        \begin{aligned}
            \norm{Z-(P-\varepsilon\one+h^+\Lset)}_\infty&\leq R<\frac{\varepsilon}{4},\\
            \norm{Z-(P'-\varepsilon\one+h^-\Lset)}_\infty&\leq R<\frac{\varepsilon}{4}.
        \end{aligned}
    \]
    Suppose both $\norm{h^+}_1$ and $\norm{h^-}_1$ exceed $\kappa\varepsilon$ (where $\kappa=2m!\lambda^m$ from \Cref{junc:lower}). Their supports are then nonempty proper subsets of $[\ncand]$. Choose a rational profile $Q$ with $\norm{Q-Z}_\infty<\varepsilon/4-R$. By the triangle inequality, $Q$ satisfies the neighborhood condition of \Cref{junc:lower} for both common junctions $P$ and $P'$, with cyclic increments $h^+$ and $h^-$, respectively. The lemma applies because every coordinate of $P$ and $P'$ exceeds $1$ and $\varepsilon<1/2$. It therefore gives
    \[
        \flam(Q)\in\{\cand{j}:h_j^+>0\}\cap\{\cand{j}:h_j^->0\}=\varnothing,
    \]
    a contradiction. Consequently, $\min(\norm{h^+}_1,\norm{h^-}_1)\leq \kappa\varepsilon$. Letting $\varepsilon\downarrow4R$ gives $\min(\norm{h^+}_1,\norm{h^-}_1)\leq4\kappa R$, also when $R=0$. Since $\sum_{j=1}^{\ncand}h_j=\sigma-\tau$, we obtain
    \begin{equation}
        \norm{x-x'}_1
        =|\tau-\sigma|+2\min(\norm{h^+}_1,\norm{h^-}_1)
        \leq|\tau-\sigma|+8\kappa\norm{y-y'}_\infty.
        \label{junc:lip-estimate}
    \end{equation}
    Equal parameters therefore give the same common junction. Together with existence, this defines $G$ throughout $U$. Moreover, adding $\norm{y-y'}_1$ to \eqref{junc:lip-estimate} and using $\norm{y-y'}_\infty\leq\norm{y-y'}_1$ gives
    \[
        \norm{G(y,\tau)-G(y',\sigma)}_1
        \leq|\tau-\sigma|+(8\kappa+1)\norm{y-y'}_1,
    \]
    which proves that $G$ is Lipschitz.

    \textit{Rank of the Jacobian:} Define the linear map
    \[
        \Pi:\R^{\ncand!}\longrightarrow\R^{\ncand!-\ncand+1},\qquad
        \Pi(x,y)=\left(y,\sum_{j=1}^{\ncand}x_j\right).
    \]
    By construction, $\Pi(G(y,\tau))=(y,\tau)$ throughout $U$. At any point $z\in U$ where $G$ is differentiable, differentiating this identity gives
    \[
        \Pi J_G(z)=I_{\ncand!-\ncand+1}.
    \]
    Thus, $J_G(z)$ has trivial kernel and hence rank $\ncand!-\ncand+1$, as claimed.
\end{proof}

\paragraph{Comparison and transfer vectors} We next use the Lipschitz family of common junctions from \Cref{junc:lipschitz} to construct what we call \emph{comparison vectors} for every pair of candidates, all lying in one $(\ncand-1)$-dimensional space. These vectors will encode constraints derived from the bounded manipulation leverage \eqref{eq:bounded-manip-potential} of $\flam$ as linear inequalities, leading to our final bound on $\lambda$ to prove \Cref{thm:universal}.

For distinct indices $i,j\in[\ncand]$, we first define the finite set of \emph{transfer vectors}
\begin{equation}
    \begin{aligned}
        \mathcal{T}_{ij}={}&\{\lambda r-s:r,s\in\linprefs,\ \cand{i}\vote{r}\cand{j}\} \cup\{s-\lambda r:r,s\in\linprefs,\ \cand{j}\vote{r}\cand{i}\},
    \end{aligned}
    \label{cone:generators}
\end{equation}
where, as before, rankings are identified with their unit vectors. We write $\operatorname{span}_{\geq0}(\mathcal{T}_{ij})$ for the set of all nonnegative linear combinations of vectors in $\mathcal{T}_{ij}$. Each transfer vector describes a ballot replacement that \emph{cannot} change the winner from $\cand{i}$ to $\cand{j}$, provided the replacement background is nonnegative, since such a change would violate bounded manipulation leverage \eqref{eq:bounded-manip-potential}.

Define $\mathcal{C}_{ij}\subseteq\R^{\ncand!}$ to be the cone of vectors $w$ satisfying
\begin{align}
    \lambda w_r&\geq w_s
    &&\text{for all $r,s\in\linprefs$ with $\cand{i}\vote{r}\cand{j}$},
    \label{cone:positive-inequalities}\\
    \lambda w_r&\leq w_s
    &&\text{for all $r,s\in\linprefs$ with $\cand{j}\vote{r}\cand{i}$}.
    \label{cone:negative-inequalities}
\end{align}
We refer to a nonzero vector in $\mathcal{C}_{ij}$ as a \emph{comparison vector} for $(\cand{i},\cand{j})$. Equivalently, $w\in\mathcal{C}_{ij}$ if and only if $w\cdot u\geq0$ for every $u\in\mathcal{T}_{ij}$. Taking $s=r$ in these inequalities and using $\lambda>1$ gives
\begin{equation}
    w_r\geq0\quad\text{if $\cand{i}\vote{r}\cand{j}$},\qquad
    w_r\leq0\quad\text{if $\cand{j}\vote{r}\cand{i}$}.
    \label{cone:weak-signs}
\end{equation}
The next proposition shows that we can choose comparison vectors for all pairs in a common space of dimension $\ncand-1$.

\begin{proposition}[Comparison vectors in a common space]
\label{cone:common-space}
    There is a linear subspace $W\subseteq\R^{\ncand!}$ of dimension $\ncand-1$ such that, for every pair of distinct indices $i,j\in[\ncand]$, there exists a vector $w_{ij}$ satisfying
    \[
        w_{ij}\in W\cap\mathcal{C}_{ij},\qquad
        w_{ij}\cdot(\lrank{i}-\lrank{j})=1.
    \]
\end{proposition}
\begin{proof} Rademacher's theorem states that a Lipschitz map between Euclidean spaces, defined on an open set, is differentiable almost everywhere~\citep[Theorem~3.1]{Heinonen05:Lectures}. Since $U$ from \Cref{junc:lipschitz} is nonempty and open, we can choose $z_0\in U$ where $G$ from the same proposition is differentiable. Put
    \[
        D=J_G(z_0),\qquad T=\operatorname{im}D,\qquad W=T^\perp.
    \]
    By \Cref{junc:lipschitz}, $D$ has rank $\ncand!-\ncand+1$, so $\dim W=\ncand-1$. We fix this space $W$ before choosing the candidate pair. The rest of the proof proceeds in two steps.

    \textit{Step 1: Excluding transfer membership.} Fix distinct $i,j\in[\ncand]$ and write $V=\lrank{i}-\lrank{j}$. We first show that
    \begin{equation}
        -V\notin\cl(\operatorname{span}_{\geq0}(\mathcal{T}_{ij})+T),
        \label{cone:exclusion}
    \end{equation}
    where $\cl$ denotes closure in $\R^{\ncand!}$.
    Suppose otherwise, and set $\eta=1/(64\kappa)$, where $\kappa \coloneqq 2m!\lambda^m$ as defined in \Cref{junc:lower}. Since $T=\operatorname{im}D$ is a linear subspace, the definition of closure gives vectors $v\in\R^{\ncand!-\ncand+1}$ and $c_0\in\operatorname{span}_{\geq0}(\mathcal{T}_{ij})$ with $\norm{c_0-(Dv-V)}_\infty<\eta/2$. By the definition of $\operatorname{span}_{\geq0}(\mathcal{T}_{ij})$, there is a positive integer $q$ such that we can write $c_0=\sum_{h=1}^{q}\alpha_h u_h$, where $\alpha_h\geq0$ and  $u_h\in\mathcal{T}_{ij}$ for each $h \in [q]$. Approximating each $\alpha_h$ by a nonnegative rational number $\beta_h$ then gives
    \[
        c=\sum_{h=1}^{q}\beta_h u_h,\qquad
        \beta_h\in\Q_{\geq0},\quad u_h\in\mathcal{T}_{ij},\qquad
        \norm{c-(Dv-V)}_\infty<\eta.
    \]
    Choose a positive integer $N$ sufficiently large that
    \(
        \frac1N\sum_{h=1}^{q}\beta_h\norm{u_h}_\infty<\eta.
    \)

    For a sufficiently small positive rational $\rho$, set $\varepsilon=\rho/(4\kappa)$ and consider the reference curve
    \[
        B_\rho(t) \coloneq G(z_0+t\rho v)-\varepsilon\one
        +\rho\bigl((1-t)\lrank{i}+t\lrank{j}\bigr),\qquad 0\leq t\leq1.
    \]
    Take $\rho$ small enough that $z_0+t\rho v\in U$ for every $t\in[0,1]$ and $\varepsilon<1/2$. Differentiability of $G$ at $z_0$ gives
    \[
        B_\rho(t)=B_\rho(0)+t\rho(Dv-V)+R_\rho(t),
        \qquad\sup_{0\leq t\leq1}\norm{R_\rho(t)}_\infty<\eta\rho
    \]
    for all sufficiently small $\rho$. The remainder bound is uniform because $\norm{t\rho v}\leq\rho\norm{v}$ for $t\in[0,1]$.

    Choose a rational profile $Q_0$ with $\norm{Q_0-B_\rho(0)}_\infty<\eta\rho$. Starting from $Q_0$, one by one apply the moves $(\rho\beta_h/N)u_h$ in order for $h=1,\ldots,q$, and repeat this sequence for $N$ rounds, skipping zero coefficients. All visited profiles are rational, and the final profile is $Q_0+\rho c$. Since the number of moves is fixed and each move has size $O(\rho)$, every visited profile converges to $G(z_0)>0$ as $\rho\to0$. The blocks removed in the corresponding ballot replacements also have weight $O(\rho)$, so we can take $\rho$ small enough that the background profile for every replacement is positive.\footnote{Here, if $u_h=\lambda r -s $ for some $r,s\in \linprefs$, the background profile when applying the move $(\rho\beta_h/N)u_h$ refers to $Q-(\rho\beta_h/N)s$, where $Q$ is the profile immediately before the move. Similarly, if $u_h=s-\lambda r $, the background profile for the same move is $Q-(\lambda\rho\beta_h/N)r$.}

    Let $Q$ be some intermediary profile reached during round $k+1$, with $0\leq k<N$. Then $Q$ can be written as
    \[
        Q=Q_0+\frac{k\rho}{N}c+\Delta,\qquad
        \norm{\Delta}_\infty\leq\frac\rho N\sum_{h=1}^{q}\beta_h\norm{u_h}_\infty<\eta\rho,
    \]
    where $\Delta$ is the sum of the moves already performed in that round. (If $Q=Q_0+\rho c$, \ie, the final profile, take $k=N$ and $\Delta=0$; the same bound on $\Delta$ holds.) Comparing $Q$ with $B_\rho(k/N)$ gives
    \[
        \begin{aligned}
            Q-B_\rho(k/N)={}&(Q_0-B_\rho(0))
            +\frac{k\rho}{N}\bigl(c-(Dv-V)\bigr)+\Delta-R_\rho(k/N).
        \end{aligned}
    \]
    Each of these four terms has infinity norm less than $\eta\rho$. Hence,
    \[
        \norm{Q-B_\rho(k/N)}_\infty<4\eta\rho=\frac{\varepsilon}{4}.
    \]

    Each $B_\rho(t)$ is obtained from the common junction $G(z_0+t\rho v)$ by subtracting $\varepsilon\one$ and adding a cyclic increment of total weight $\rho>\kappa\varepsilon$, supported on $\{i,j\}$. This support is proper because $\ncand\geq3$. Further, each coordinate of the common junction $G(z_0+t\rho v)$ exceeds $1$, so $\varepsilon$ is less than half of this coordinate. Thus, \Cref{junc:lower} applies and shows that every visited profile elects either $\cand{i}$ or $\cand{j}$. At $t=0$, the increment is supported only on $\{i\}$, so $\flam(Q_0)=\cand{i}$. At $t=1$, it is supported only on $\{j\}$, so $\flam(Q_0+\rho c)=\cand{j}$.

    Let $Q'$ be the first visited profile electing $\cand{j}$, and let $Q$ be its predecessor. Then $\flam(Q)=\cand{i}$, and $Q'=Q+\alpha u_h$ for some $h$ with $\alpha=\rho\beta_h/N\in\Q_{>0}$. If $u_h=\lambda r-s$ with $\cand{i}\vote{r}\cand{j}$, applying bounded manipulation leverage \eqref{eq:bounded-manip-potential} with background profile $Q-\alpha s$ and weights $\lambda\alpha$ and $\alpha$ gives
    \[
        \flam(Q')\voteq{r}\flam(Q)=\cand{i}.
    \]
    If $u_h=s-\lambda r$ with $\cand{j}\vote{r}\cand{i}$, applying the same property with background profile $Q-\lambda\alpha r$ and the same weights gives
    \[
        \cand{i}=\flam(Q)\voteq{r}\flam(Q').
    \]
    The replacement backgrounds are positive by our choice of $\rho$. In either case, the displayed comparison contradicts $\flam(Q')=\cand{j}$. This proves \eqref{cone:exclusion}.

    \textit{Step 2: Constructing the comparison vector.} Let $M=\cl(\operatorname{span}_{\geq0}(\mathcal{T}_{ij})+T)$, a closed convex cone, and choose $p\in M$ minimizing the Euclidean distance to $-V$. Such a point exists because $0\in M$ and the minimization can be restricted to a closed bounded ball. Put $w=p+V$, which is nonzero by \eqref{cone:exclusion}. For any $\xi\in M$, minimality along the segment from $p$ to $\xi$ gives $w\cdot(\xi-p)\geq0$.\footnote{To see why: convexity of $M$ ensures that $p+t(\xi-p)\in M$ for $0\leq t\leq1$. Hence $f(t)=\norm{p+V+t(\xi-p)}_2^2=\norm{w+t(\xi-p)}_2^2$ is minimized at $t=0$ on $[0,1]$. Its right derivative at $0$ is therefore nonnegative, so $0\leq f'(0^+)=2w\cdot(\xi-p)$.} Taking $\xi=0$ and $\xi=2p$ shows that $w\cdot p=0$. Consequently,
    \[
        w\cdot\xi\geq0\quad(\xi\in M),\qquad
        w\cdot V=w\cdot(w-p)=\norm{w}_2^2>0.
    \]
    Since $T$ and $-T$ are contained in $M$, we have $w\in T^\perp=W$. Also, $\mathcal{T}_{ij}\subseteq M$, so $w\in\mathcal{C}_{ij}$ by the relationship between transfer and comparison vectors. Therefore,
    \[
        w_{ij}=\frac{w}{w\cdot V}
    \]
    belongs to $W\cap\mathcal{C}_{ij}$ and satisfies $w_{ij}\cdot(\lrank{i}-\lrank{j})=1$, as required.
\end{proof}

\paragraph{Final steps} We are now ready to finish up the proof of the main theorem. As we will see, the existence of the $(\ncand-1)$-dimensional space containing all comparison vectors from \Cref{cone:common-space} will play a key role in establishing the final bound on $\lambda$.

\begin{proof}[Proof of \Cref{thm:universal}]
    
    Let $\scf$ and $K$ satisfy the assumptions of the theorem, and retain $\lambda=K+1$ and the rule $\flam$ on rational profiles constructed in \Cref{lemma:flam-properties}. Consider the cycle
    \[
        \cand{1}\longrightarrow\cand{2}\longrightarrow\cdots\longrightarrow\cand{\ncand}\longrightarrow\cand{1},
    \]
    with indices interpreted modulo $\ncand$, and let $w_i=w_{i,i+1}$ be the comparison vectors from \Cref{cone:common-space}. Each $w_i$ is nonzero by its normalization. Since these $\ncand$ vectors lie in the same $(\ncand-1)$-dimensional space $W$, there are real coefficients $\alpha_1,\ldots,\alpha_\ncand$, not all zero, such that
    \[
        \sum_{i=1}^{\ncand}\alpha_iw_i=0.
    \]

    We first show that every $\alpha_i$ is nonzero and all coefficients have the same sign. Suppose otherwise. Starting from the cycle above, delete each edge with $\alpha_i=0$ and reverse each edge with $\alpha_i<0$. The resulting directed graph is acyclic: a directed cycle would have to use every edge of the original cycle with a consistent orientation, which requires all coefficients to be nonzero and have the same sign.

\begin{table}[t]
    \centering
    \caption{Key notation used in the proof of \Cref{thm:universal}.}
    \label{tab:universal-notation}
    \small
    \renewcommand{\arraystretch}{1.08}
    \begin{tabular}{@{}p{0.28\textwidth}p{\dimexpr0.72\textwidth-2\tabcolsep\relax}@{}}
        \toprule
        \textbf{Notation} & \textbf{Meaning} \\
        \midrule
        $\scf$ & Original anonymous, onto social choice function on integer profiles. \\
        $K,\lambda$ & Uniform upper bound $K$ on the manipulation leverage of $\scf$, and $\lambda=K+1$. \\
        $\flam$ & Rule on nonzero rational profiles constructed from $\scf$ in \Cref{lemma:flam-properties}. \\
        $P,P_r,P+ar$ & Profile vector, weight of ranking $r$, and addition of weight $a$ of ranking $r$ to $P$. \\
        $\rrank{j},\lrank{j};\ \Rset,\Lset$ & Cyclic rankings starting at $\cand{j}$ in increasing and decreasing index order (modulo $\ncand$); their respective families. \\
        $x\Lset,x\Rset$ & Cyclic profiles $\sum_{j=1}^{\ncand}x_j\lrank{j}$ and $\sum_{j=1}^{\ncand}x_j\rrank{j}$. \\
        $P=(x,y)$ & Cyclic coordinates $x$ (rankings in $\Lset$) and external coordinates $y$ (all other rankings) of the profile $P$. \\
        $\supp(x)$ & Support $\{j:x_j>0\}$ of a nonnegative vector $x$. \\
        $\delta$ & Constant $(\ncand-2)/(\lambda^6+\ncand-2)$ bounding the winner's share in purely $\Lset$-profiles. \\
        $\Omega$ & Space $\R_{>0}^{\ncand!}$ of full-support real profiles. \\
        $\cl_{\Omega},\cl$ & Closure relative to $\Omega$ and closure in the ambient space $\R^{\ncand!}$, respectively. \\
        $A_j$ & Relative closure in $\Omega$ of rational profiles electing $\cand{j}$. \\
        $\kappa$ & Scalar $\kappa=2\ncand!\lambda^{\ncand}$ in the perturbed-junction bound (\Cref{junc:lower}). \\
        $U$ & Open domain for the common junction map, living in $\mathbb{R}^{\ncand!-\ncand+1}$. Each $(y,\tau)\in U$ consists of external coordinates $y$ and total cyclic weight $\tau$. \\
        $G,H$ & Common junction map $G(y,\tau)=(H(y,\tau),y)$ and its cyclic-coordinate map $H$. \\
        $J_G(z),D$ & Jacobian of $G$ at a differentiability point $z\in U$, and its value $D=J_G(z_0)$ at the point $z_0$ chosen in the proof of \Cref{cone:common-space}. \\
        $T,W$ & Subspaces $T=\operatorname{im}D$ and $W=T^\perp$; $W$ has dimension $\ncand-1$. \\
        $\mathcal{T}_{ij}$ & Transfer vectors, encoding replacements that cannot change winner $\cand{i}$ to $\cand{j}$. \\
        $\operatorname{span}_{\geq0}(\mathcal{T}_{ij})$ & Cone of all nonnegative linear combinations of transfer vectors in $\mathcal{T}_{ij}$. \\
        $\mathcal{C}_{ij},w_{ij}$ & Comparison cone and a vector $w_{ij}\in W\cap\mathcal{C}_{ij}$ normalized by $w_{ij}\cdot(\lrank{i}-\lrank{j})=1$. \\
        \bottomrule
    \end{tabular}
\end{table}

    Let $I=\{i\in[\ncand]:\alpha_i\neq0\}$ and put $z_i=\sign(\alpha_i)w_i$ for $i\in I$. The definition of comparison vectors from inequalities \eqref{cone:positive-inequalities} and \eqref{cone:negative-inequalities} gives $-\mathcal{C}_{ij}=\mathcal{C}_{ji}$, so $z_i$ is a comparison vector for its edge's final orientation. Since the graph is acyclic, there is a ranking $r\in\linprefs$ placing the tail of every remaining edge above its head (\ie, $r$ respects the orientation of every remaining edge). The bounds in \eqref{cone:weak-signs} therefore give $(z_i)_r\geq0$ for every $i\in I$. Since $ \sum_{i\in I}|\alpha_i|z_i=0$  and all $|\alpha_i|$ are positive, we must have $(z_i)_r=0$ for every $i\in I$. Applying \eqref{cone:positive-inequalities} to each oriented edge then gives
    \[
        (z_i)_s\leq\lambda(z_i)_r=0
        \qquad\text{for every $s\in\linprefs$ and $i\in I$.}
    \]
    Thus, the nonzero vectors $z_i$ are all coordinatewise nonpositive, so their positively weighted sum cannot be zero. This contradiction proves that the $\alpha_i$ are all nonzero and have the same sign.

    Multiplying the relation by $-1$ if necessary, assume $\alpha_i>0$ for every $i\in[\ncand]$. Set $v_i=\alpha_iw_i$. Since each $\mathcal{C}_{i,i+1}$ is a cone, we have
    \[
        v_i\in\mathcal{C}_{i,i+1},\qquad v_i\neq0,\qquad
        \sum_{i=1}^{\ncand}v_i=0.
    \]
    Define the coordinate range of each $v_i$ by
    \[
        R_i=\max_{s\in\linprefs}(v_i)_s-\min_{s\in\linprefs}(v_i)_s.
    \]
    Each $R_i$ is positive: by \eqref{cone:weak-signs}, the only constant vector in $\mathcal{C}_{i,i+1}$ is the zero vector, while $v_i\neq0$. Choose $j\in[\ncand]$ minimizing $R_i$.

    Let $r$ be the ranking
    \[
        r:\quad\cand{j+1}\succ\cand{j+2}\succ\cdots\succ\cand{j+\ncand},
    \]
    where additions are once again modulo $\ncand$, and let $\bar r$ be the reversal of $r$. For every $i\neq j$, ranking $r$ places $\cand{i}$ above $\cand{i+1}$, while $\bar r$ places them in the opposite order. Applying \eqref{cone:positive-inequalities} and \eqref{cone:negative-inequalities} gives
    \[
        \max_{s\in\linprefs}(v_i)_s\leq\lambda(v_i)_r,
        \qquad
        \min_{s\in\linprefs}(v_i)_s\geq\lambda(v_i)_{\bar r}.
    \]
    Subtracting the second bound from the first yields $R_i\leq\lambda\bigl((v_i)_r-(v_i)_{\bar r}\bigr)$ for each $i \neq j$. Since $R_j$ is minimal and $\sum_{i=1}^{\ncand}v_i=0$, it follows that
    \[
        \begin{aligned}
            (\ncand-1)R_j
            &\leq\sum_{i\neq j}R_i
            \leq\lambda\sum_{i\neq j}\bigl((v_i)_r-(v_i)_{\bar r}\bigr)\\
            &=\lambda\bigl((v_j)_{\bar r}-(v_j)_r\bigr)
            \leq\lambda R_j.
        \end{aligned}
    \]
    Dividing by $R_j>0$ gives $\lambda\geq\ncand-1$. Recalling that $\lambda=K+1$, we obtain $K\geq\ncand-2$, as claimed.
\end{proof}

\section{Discussion and Future Work}\label{sec:future}

In this paper, we proposed to measure the comparative advantage of a strategic voter from a resource augmentation perspective, asking how many additional truthful copies of her vote would suffice to ensure an outcome at least as good as her best manipulation. This novel approach to manipulability gives rise to natural future directions and open problems, some of which we discuss below.

\paragraph{(a) Characterizing Borda Count} As we have seen in \Cref{thm:borda}, the manipulation leverage of Borda Count is $\ncand-2$, which, by our universal lower bound in \Cref{thm:universal}, is optimal among (onto and anonymous) rules with bounded manipulation leverage over all electorates. However, Borda Count is not necessarily the \emph{unique} rule that achieves this bound; indeed, it is not difficult to define slight modifications to the rule that preserve its manipulation leverage, while keeping the final rule onto and anonymous (or even preserving stronger notions of \emph{unanimity} and \emph{homogeneity}, respectively, both of which are satisfied by Borda Count). Still, a natural and important next direction is to investigate whether having manipulation leverage of $\ncand-2$ can be combined with other axiomatic properties to provide a characterization of Borda Count. Investigating this question can shed light on the relationship between our resource-augmentation approach and other well-studied notions from the computational social choice literature.

\oldparagraph{(b) Sublinear Growth in Electorate Size?} Our universal lower bound in \Cref{thm:universal} applies to rules with a finite bound on their manipulation leverage across all electorate sizes. This was motivated by our game-theoretic framework in \Cref{sec:game-theory}, which exposed a key disadvantage of having $\Omega(\nvote)$ manipulation leverage: in the worst case, \emph{any} number of truthful copies may not suffice to disincentivize a voter from manipulating, provided that other voters are making the same choice. There is still value, however, in investigating whether we can identify rules---either from prior literature or by designing new ones---whose manipulation leverage grows \emph{sublinearly} with the electorate size $\nvote$. First, such rules may provide significant advantages in modern applications where $\ncand \gg \nvote$, such as meta-search engines (where the ``voters'' are search algorithms and the ``candidates'' are webpages) and AI alignment/evaluation (where the ``voters'' are evaluators/benchmarks and the ``candidates'' are AI models/outputs). Second, such rules may still avoid the game-theoretic pitfall discussed above: for example, given a rule $\scf$ with $\spdeg{f} \leq \sqrt{\nvote}$, it is easy to check that setting $k$ to the original number of voters in the game from \Cref{sec:game-theory} will ensure that truthful augmentation is a weakly dominant strategy for everyone. As such, we consider this an important direction.

\paragraph{(c) Automating the Search} A promising avenue to investigate (b) and related open problems is to adopt an \emph{automated reasoning} approach, \ie, to use computer-generated proofs to gain intuition on small instances, and then generalize these via theoretical analysis. Indeed, after \citet{Tang08:A,Tang09:Computer} demonstrated the suitability of this approach to social choice theory by reproving the Gibbard–Satterthwaite theorem (as well as the impossibility theorem of \citet{Arrow63:Social}), tools such as SAT solvers have paved the way to many new and important results \citep{Geist11:Automated,Brandt16:Finding,Brandl21:Distribution,Brandt17:Optimal,Brandl19:Strategic,Endriss20:Analysis}. Similarly, one can use such tools to search for rules with low manipulation leverage for a fixed $\nvote$, or for proofs strengthening the Gibbard–Satterthwaite theorem to give stronger lower bounds (similar to what our \Cref{thm:universal} does for bounds over all $\nvote$). A key challenge is that manipulation leverage inherently deals with elections with a varying number of voters, which may be difficult to encode with SAT solvers and, even if encoded, can lead to unnatural rules that behave differently under different $\nvote$. A potential solution to this obstacle is to strengthen our anonymity requirement (\Cref{sec:manipot}) to \emph{homogeneous rules}, for which only the relative frequency of each ranking matters. By encoding these frequencies, one can seek to minimize manipulation leverage via, for example, (mixed integer) linear programs~\citep{Peters25:Core,Berker25:Edge}\footnote{Linear programs have also been adopted in the adjacent field of \emph{automated mechanism design}; see, \eg, \citep{Guo10:Computationally}.} or SMT solvers~\citep{Brandl18:Proving}.

\paragraph{(d) Beyond Single-Winner}
While our work focuses on analyzing the manipulation leverage of rules that map voters' rankings to a single winner (or a tie, later to be broken), our framework can be easily adapted to rules with other input formats (\eg,
ratings/scores, approval ballots) and desired outputs (an aggregate ranking, a committee of winners, or a probability distribution over candidates). As a concrete example, take approval-based multiwinner elections, where the goal is to output a committee of candidates of a fixed size. Here, even weak forms of proportionality are incompatible with strategyproofness~\citep{Peters18:Proportionality}. A natural question is to investigate the manipulation leverage of multiwinner rules. Can we identify or design rules with strong proportionality guarantees (\eg, see properties such as JR, EJR, EJR+, or FJR~\citep{Lackner22:Multi}) \emph{and} low manipulation leverage? Or does assuming such properties impose strong lower bounds on the manipulation leverage? What is the manipulation leverage of known multiwinner rules, such as Proportional Approval Voting, or the Method of Equal Shares~\citep{Peters20:Proportionality}? Likewise, a separate concrete direction is to investigate the manipulation leverage of \emph{randomized} single-winner rules, and whether it is possible to formulate a quantitative analogue of Gibbard's randomized theorem~\citep{Gibbard78:Straightforwardness}, similar to what our \Cref{thm:universal} does for the (deterministic) Gibbard-Satterthwaite theorem. We consider such extensions to be valuable next steps for the study of manipulation leverage.

\paragraph{(e) Budgeted Manipulations} One can view the manipulation leverage as the \emph{exchange rate} between a single manipulation and adding truthful copies. Given a rule with manipulation leverage $k$, a voter is better off finding $k$ copies of herself if the total cost to do this is less than the cost of manipulation (see interpretation (1) in \Cref{sec:how-many}). An interesting follow-up is to consider whether a voter can benefit from \emph{combining} these two strategies (manipulating and finding copies). More concretely, consider a model that specifies the cost a voter must incur to find her (truthful) copies, as well as the cost of manipulations of different ``magnitude'' (\eg, the cost can be increasing with the swap distance between her true vote and the untruthful one). If the voter has a fixed budget to begin with, what is her optimal way of allocating it between finding (some number of) copies and (partial) manipulation? Formally analyzing this model can put our results into a broader context.

\paragraph{(f) Group-Strategyproofness} In this paper, we have focused on manipulations by individual voters. More generally, it is of interest to study the robustness of voting rules to \emph{group} manipulations, \ie, when a \emph{coalition of voters} misreport their preferences to get an outcome that they all strictly prefer to the original winner. One natural extension of our manipulation leverage framework (\Cref{def:manipot}) to this setting is to ask how many copies any subset of voters would need (where each voter in the subset is copied an equal number of times) to ensure that no group manipulation (in the original profile) leads to an outcome that is strictly preferred by all voters in this subset to the outcome with the copies. Our results do not generalize to this setting, and even trivial upper bounds such as $\nvote-1$ for majority-consistent rules (\Cref{remark:majority-consistent}) no longer hold, as adding (an equal number of) copies of a subset of voters, even if they all benefit from the same manipulation, is no longer guaranteed to create a majority or a Condorcet winner. Consequently, some Condorcet extensions can even have infinite group manipulation leverage at a fixed (such as Copeland, see \Cref{thm:copeland}).
Alternative extensions (where the new copies need not be distributed uniformly among the coalition members) may give rise to different results, and studying these is a natural next step.

\paragraph{(g) Relationship to Other Measures} Comparing our results with prior work, there appears to be no clear correlation between the manipulation leverage of a rule and how it fares under other notions of manipulability. For example, while we have shown that Borda Count $(\borda)$ has significantly lower manipulation leverage than other rules (at least with many voters), several papers that study the frequency with which a manipulation can occur (see \Cref{sec:related}) conclude that $\borda$ is ``more manipulable'' than other rules in this regard~\citep{durand2015towards,chamberlin1985investigation,smith1999manipulability}. This is not surprising, since the fraction of profiles in which a voter is able to manipulate does not reveal what she could have gained in those profiles by increasing her voter weight instead; thus, it is incomparable with our resource-augmentation approach. On the other hand, there are also rules that satisfy weaker forms of strategyproofness, but end up having very large manipulation leverage. For example, various majoritarian Condorcet extensions satisfy ``weak-strategyproofness'' when assuming voters do not know what tiebreaker is going to be used and are acting safely~\citep{brandt2015setmono,brandt2023characterizing}, but all such rules have basically the highest possible manipulation leverage for any majority-consistent rule (\Cref{thm:majoritarian}). For other rules that achieve weak-strategyproofness, the manipulation leverage can be unbounded (see \Cref{app:omnipot}), demonstrating that once we remove the aforementioned assumption on voters' ignorance of the tiebreaker, no number of truthful copies can match what a voter can achieve by manipulating. A compelling avenue for future research is to combine our metric with other degrees of manipulability. This direction can, for example, entail computing the \emph{expected} manipulation leverage of rules under different voter distributions / empirical data, or defining an analogously weaker version of $\ksp{k}$ (\Cref{def:kmanip}) that incorporates voter uncertainty about the tiebreaker.

\paragraph{Conclusion} When confronted about the manipulability of his voting method, 18$^{\text{th}}$-century mathematician Jean-Charles de Borda famously claimed ``My scheme is only intended for honest men''~\citep{black1958theory}. 
As we show in this paper, his rule can in fact tolerate possible manipulators who would be better off getting a few honest friends to vote, and meanwhile the situation can be much grimmer for the class of rules named after his intellectual rival Marquis de Condorcet (at least for large electorates). Just how many of these honest copies a voter needs before abandoning manipulation under different models and rules is an exciting research agenda that can significantly enhance our understanding of manipulation under limited resources.

\section*{Acknowledgments}
We would like to thank Felix Brandt, Edith Elkind, Paul Gölz, Lirong Xia, and Brian Hu Zhang for helpful discussions and feedback at various stages of this project. R.E.B., V.C., J.L., and C.O.\ thank the Cooperative AI Foundation, Macroscopic Ventures (formerly Polaris Ventures
/ the Center for Emerging Risk Research), and Jaan Tallinn's donor-advised fund at Founders Pledge for financial support. R.E.B.\ is also supported by the Cooperative AI PhD Fellowship. E.H.\ is supported by the Israel Science Foundation grants 2697/22 and 3007/24. C.O.\ was supported by an FLI PhD Fellowship.

\section*{Statement on AI Use}

The proof of \Cref{thm:universal} in \Cref{sec:universal} was found through repeated interactions with GPT-6.1 Sol. The final proof in \Cref{sec:universal} follows the general construction generated by GPT, with substantial rewrites by the authors to improve clarity. GPT was also used for minor edits throughout the paper (such as for language improvements and editing figures). The
authors take full responsibility for the content of this paper.

\printbibliography

\appendix
\crefalias{section}{appendix}
\crefalias{subsection}{appendix}

\section{Proofs Deferred From the Main Body}

\subsection{Proof of \Cref{thm:pos-upper}}\label{app:pos-upper}
\posup*
\begin{proof}
    Fix some $\profile$, manipulation $\manip{1}$, tiebreaker $\tieorder$, and some $\displaystyle k \geq \max_{\ell \in [\ncand-1]} \upperpos{\ell}$. We will prove that $\pos_\tiebreaker(\kprofile{k}{1}) \voteq{1} \pos_\tiebreaker\maniprofile{1}$. As usual, the truthful preference order of voter 1 is $\cand{1} \vote{1} \cand{2} \vote{1} \cdots \vote{1} \cand{\ncand}$. For any candidate $\cand{} \in \cands$, say $\bs(\cand{}), \bs'(\cand{}),$ and $\bs''(\cand{})$ are the total scores assigned to $\cand{}$ by $\pos$ in $\profile$, $\maniprofile{1}$, and $\kprofile{k}{1}$, respectively. Say $\bs_1(\cand{}), \bs'_1(\cand{})$, and $\bs''_1(\cand{})$ are the contributions of voter 1 to $\bs(\cand{}),\bs'(\cand{})$ and $\bs''(\cand{})$, respectively. Lastly, say  $\bs_{-1}(\cand{}) = \bs(\cand{})-\bs_1(\cand{}), \bs'_{-1}(\cand{}) = \bs'(\cand{})-\bs'_1(\cand{})$, and $\bs''_{-1}(\cand{}) = \bs''(\cand{})-\bs''_1(\cand{})$. First notice that since $k \geq 0$, we have $\pos_\tiebreaker(\kprofile{k}{1}) \voteq{1} \pos_\tiebreaker(\profile)$, since adding more copies of voter 1 can only help the candidates she prefers to $\pos_\tiebreaker(\profile)$. So if $\pos_\tiebreaker(\profile) \voteq{1} \pos_\tiebreaker\maniprofile{1}$, we are done. Otherwise, assume we have $\pos_\tiebreaker\maniprofile{1} = \cand{\ell} \vote{1} \cand{i} = \pos_\tiebreaker(\profile)$ for some $\ell<i$. Since $\cand{\ell} \in \pos\maniprofile{1}$, we have $\bs'(\cand{\ell}) \geq \bs'(\cand{i})$, with the inequality strict if $\cand{i} \tieorder \cand{\ell}$. As $\bs_1(\cand{\ell})=\score_\ell$, $\bs_1(\cand{i})=\score_i$ and $\bs_1'(\cand{\ell})-\bs_1'(\cand{i}) \leq 1$, we have
    \begin{align}
        \bs(\cand{\ell})= \bs'(\cand{\ell})-\bs'_1(\cand{\ell})+\bs_1(\cand{\ell}) \geq \bs'(\cand{i})-\bs_1'(\cand{i}) -1 + \score_\ell= \bs(\cand{i})-\score_i  -1+\score_\ell,\label{eq:posineq}
    \end{align}
    with the inequality strict if $\cand{i} \tieorder \cand{\ell}$. Take any $j \in \{\ell+1 ,\ell+2, \ldots, \ncand\}$. Below, we will show that $\pos_\tiebreaker(\kprofile{k}{1}) \neq \cand{j}$, hence proving $\pos_\tiebreaker(\kprofile{k}{1}) \voteq{1} \cand{\ell} = \pos_\tiebreaker\maniprofile{1}$ as desired. Since $\pos_\tiebreaker(\profile)=\cand{i}$ we have either
    \begin{align}
    \bs(\cand{i}) > \bs(\cand{j})\text{ or }[\bs(\cand{i}) = \bs(\cand{j})\text{ and }\cand{i} \tieqorder \cand{j}]. \label{eq:posineq2}
    \end{align}
    To complete the proof, we will consider two cases.
    
    \textbf{Case 1:} $\upperpos{\ell}=\frac{1}{\score_{\ell}-\score_{\ell+1}}-1$, so $k \geq \frac{1}{\score_{\ell}-\score_{\ell+1}}-1$. In this case, we must have $\score_\ell > \score_{\ell+1}$ (since the second term inside the $\min$ in $\upperpos{\ell}$ is always finite). Combining \eqref{eq:posineq} and \eqref{eq:posineq2}, we get
    \begin{align*}
        \bs''(\cand{\ell}) = \bs(\cand{\ell})+ k\score_\ell &\geq \bs(\cand{i})+\score_\ell-\score_i-1+k\score_\ell \geq  \bs(\cand{j})+\score_\ell-\score_i-1+k(\score_\ell-\score_j)+k\score_j\\& \geq  \bs''(\cand{j})+\score_\ell-\score_{\ell+1} 
        -1+k(\score_\ell-\score_{\ell+1}) \geq \bs''(\cand{j}),
    \end{align*}
    with the inequality strict if $\cand{i} \tieorder \cand{\ell}$ or $\cand{j} \tieorder \cand{i}$. This implies we have either $\bs''(\cand{\ell}) > \bs''(\cand{j})$ or [$\bs''(\cand{\ell}) = \bs''(\cand{j})$ and $\cand{\ell}\tieorder\cand{i} \tieqorder \cand{j}$] (recall that $\ell \neq i$). Either case implies $\pos_\tiebreaker(\kprofile{k}{1}) \neq \cand{j}$, completing the proof.

    \textbf{Case 2:} $\upperpos{\ell}= \frac{(\nvote-1)(\upave{2}-\downave{\posbar{\ell}})+\frac{1}{2}}{\upave{\posbar{\ell}}-\score_{\ell+1}}$, so $k \geq\frac{(\nvote-1)(\upave{2}-\downave{\posbar{\ell}})+\frac{1}{2}}{\upave{\posbar{\ell}}-\score_{\ell+1}}$. We will show that there exists a $t \in [\posbar{\ell}]$ such that $\bs''(\cand{t})>\bs''(\cand{j})$, which completes the proof. Since each voter can rank at most one candidate in each position, we have
    \begin{align}
        \frac{1}{\posbar{\ell}}\sum_{t=1}^{\posbar{\ell}} \bs_{-1}(\cand{t}) \geq \frac{(\nvote-1)}{\posbar{\ell}} \sum_{t'=\ncand-\posbar{\ell}+1}^{\ncand} \score_{t'}=(\nvote-1)\cdot \downave{\posbar{\ell}},\label{eq:posineq3}
    \end{align}
    since any $\posbar{\ell}$ candidates must get (in total) at least the bottom $\posbar{\ell}$ scores from each voter. 
    Similarly, we have $\bs_{-1}(\cand{j})+\bs_{-1}(\cand{\ell}) \leq (\nvote-1)(\score_1+\score_2)= 2(\nvote-1)\cdot \upave{2}$, since any two candidates can get (in total) at most the top two scores from each voter. Combining this with \eqref{eq:posineq} and \eqref{eq:posineq2}, we get
    \begin{align*}
        2\bs_{-1}(\cand{j})-1 =\bs_{-1}(\cand{j})+ \bs(\cand{j})-\score_j-1&\leq \bs_{-1}(\cand{j})+ \bs(\cand{i})-\score_j-1
        \\ &\leq \bs_{-1}(\cand{j})+ (\bs(\cand{\ell})+\score_i+1-\score_\ell)-\score_j-1\\
        & = \bs_{-1}(\cand{j})+ \bs_{-1}(\cand{\ell})+\score_i-\score_j\\
        &\leq  2(\nvote-1)\cdot \upave{2}+\score_i-\score_j\\
        \Rightarrow \bs_{-1}(\cand{j}) &\leq (\nvote-1)\cdot \upave{2}+ \frac{1+\score_i-\score_j}{2}.
    \end{align*}
    Together with \eqref{eq:posineq3}, this implies
    \begin{align*}
        \bs''(\cand{j})- \frac{1}{\posbar{\ell}} \sum_{t=1}^{\posbar{\ell}}\bs''(\cand{t}) &= \left(\bs_{-1}(\cand{j})+ (k+1) \score_j\right)- \frac{1}{\posbar{\ell}} \sum_{t=1}^{\posbar{\ell}}\left(\bs_{-1}(\cand{t}) +(k+1)\score_t \right) \\
        &\leq (\nvote-1)\cdot \upave{2}+ \frac{1+\score_i-\score_j}{2}+ (k+1) \score_j-(\nvote-1)\cdot \downave{\posbar{\ell}}-(k+1) \cdot \upave{\posbar{\ell}} \\
        &= (\nvote-1)\cdot (\upave{2}-\downave{\posbar{\ell}})+ \frac{1+\score_i-\score_j}{2}- (k+1)\cdot(\upave{\posbar{\ell}}-\score_j)\\
        &= (\nvote-1)\cdot (\upave{2}-\downave{\posbar{\ell}})+ \frac{1+\score_i+\score_j}{2}-\upave{\posbar{\ell}} - k \cdot (\upave{\posbar{\ell}}-\score_j)\\
        &\leq  (\nvote-1)\cdot (\upave{2}-\downave{\posbar{\ell}})+ \frac{1}{2}+\score_{\ell+1}-\upave{\posbar{\ell}} - k \cdot (\upave{\posbar{\ell}}-\score_{\ell+1})\\
        &< (\nvote-1)\cdot (\upave{2}-\downave{\posbar{\ell}})+ \frac{1}{2}- k \cdot (\upave{\posbar{\ell}}-\score_{\ell+1})\\
        &\leq (\nvote-1)\cdot (\upave{2}-\downave{\posbar{\ell}})+ \frac{1}{2}- \frac{(\nvote-1)(\upave{2}-\downave{\posbar{\ell}})+\frac{1}{2}}{\upave{\posbar{\ell}}-\score_{\ell+1}}\cdot (\upave{\posbar{\ell}}-\score_{\ell+1})\\
        &=0.
    \end{align*}
    By an averaging argument, this implies that there exists some $t \in [\posbar{\ell}]$ such that $\bs''(\cand{t})> \bs''(\cand{j})$. Hence, $\pos_\tiebreaker(\kprofile{k}{1}) \neq \cand{j}$, completing the proof.
\end{proof}

\subsection{Proof of \Cref{thm:pos-lower}}\label{app:pos-lower}

\posdown*

\begin{proof}
    Fix some $\ell \in [\ncand-1]$ and $0 \leq k < \lowerpos{\ell}$ and assume  $\nvote \geq \frac{2\score_2}{1-\score_2}$. Separate the candidates into three groups $\{\acand{1}, \acand{2}, \aldots, \acand{\ell-1}\}$, $\{\bcand{\ell}, \bcand{\ell+1}\}$, and $\{\ccand{\ell+2}, \ccand{\ell+3}, \cldots, \ccand{\ncand}\}$; note that the first (resp.\ last) group will be empty if $\ell =1$ (resp.\ $\ell=\ncand-1$). For presentation purposes, we will use a different color while referring to candidates from each group. Consider the following profile, consisting of:

        \begin{center}
        \begin{tabular}{c|l}
            \cellcola Voter 1 & \cellcola $\acand{1} \vote{1} \acdots \vote{1} \acand{\ell-1} \vote{1} \bcand{\ell} \vote{1} \bcand{\ell+1} \vote{1} \ccand{\ell+2} \vote{1} \ccdots \vote{1} \ccand{\ncand}$
            \\ \cellcolb Voter 2 & \cellcolb $\bcand{\ell+1} \vote{2} \ccand{\ncand} \vote{2} \ccand{\ncand-1} \vote{2} \ccdots \vote{2} \ccand{\ell+2} \vote{2} \acand{\ell-1} \vote{2} \acand{\ell-2} \vote{2} \acdots \vote{2} \acand{1} \vote{2} \bcand{\ell}$\\
             \cellcola $\frac{\nvote-2}{2}$ voters (Group 1) & \cellcola $\bcand{\ell} \vote{i} \bcand{\ell+1} \vote{i} \ccand{\ncand} \vote{i} \ccand{\ncand-1} \vote{i} \ccdots \vote{i} \ccand{\ell+2} \vote{i} \acand{\ell-1} \vote{i} \acand{\ell-2} \vote{i} \acdots \vote{i} \acand{1}$\\
             \cellcolb
        $\frac{\nvote-2}{2}$ voters (Group 2) &\cellcolb $\bcand{\ell+1} \vote{i} \bcand{\ell} \vote{i} \ccand{\ncand} \vote{i} \ccand{\ncand-1} \vote{i} \ccdots \vote{i} \ccand{\ell+2} \vote{i} \acand{\ell-1} \vote{i} \acand{\ell-2} \vote{i} \acdots \vote{i} \acand{1}$
        \end{tabular} 
    \end{center}
     Say the tiebreaker ranks  $\bcand{\ell} \tieorder \bcand{\ell+1} \tieorder \cand{}$ for all $\cand{} \in \cands\setminus \{\bcand{\ell},\bcand{\ell+1}\}$.
    For any $\cand{} \in \cands$, say $\bs_k(\cand{})$ is the total score assigned to $\cand{}$ by $\pos$ in $\kprofile{k}{1}$. For any $t > \ell+1$, we have $\bs_k(\bcand{\ell+1}) \geq \bs_k(\ccand{t})$, since $\bcand{\ell+1}$ is ranked above $\ccand{t}$ by all voters. Further, since $k< \lowerpos{\ell}$, we have
    \begin{align*}
        \bs_k(\bcand{\ell})-\bs_k(\bcand{\ell+1})&=(k+1)( \score_{\ell}-\score_{\ell+1} ) -1 <0\text{, and}\\
        \forall t < \ell:  \bs_k(\acand{t})-\bs_k(\bcand{\ell+1}) &= (k+1)(\score_t-\score_{\ell+1}) + \score_{\ncand-t} -1+(\nvote-2) (\score_{\ncand-t+1} -\upave{2}) \leq 0.
    \end{align*}
    Thus, for all $\cand{} \in \cands \setminus\{\bcand{\ell+1}\}$, we have either $\bs_k(\bcand{\ell+1})>\bs_k(\cand{})$ or [$\bs_k(\bcand{\ell+1})\geq \bs_k(\cand{})$ and $\bcand{\ell+1} \tieorder \cand{}$], implying that $\pos_\tiebreaker(\kprofile{k}{1})=\bcand{\ell+1}$.

    Going back to the original profile $\profile$ (with no copies), assume voter 1 instead reports the alternative preference order  $\bcand{\ell} \manip{1} \ccand{\ncand} \manip{1} \ccand{\ncand-1} \manip{1} \ccdots \manip{1} \ccand{\ell+2} \manip{1} \acand{\ell-1} \manip{1} \acand{\ell-2} \manip{1} \acdots \manip{1} \acand{1} \manip{1} \bcand{\ell+1}$. For each $\cand{} \in \cands$, say $\bs'(\cand{})$ is the total score assigned to $\cand{}$ by $\pos$ in $\maniprofile{1}$. Then we have $\bs'(\bcand{\ell})=\bs'(\bcand{\ell+1})=1+(\nvote-2)\cdot \upave{2}$, whereas for any other candidate $\cand{} \in \cands \setminus \{\bcand{\ell},\bcand{\ell+1}\}$, we have
    \begin{align*}\bs'(\cand{}) &\leq 2 \score_{2} +(\nvote-2)\score_3 \leq \nvote\score_2 \\\Rightarrow \bs'(\bcand{\ell}) - \bs'(\cand{}) &\geq 1-2\score_2 + (\nvote-2)\cdot\left(\upave{2}-\score_2\right)= 1-2\score_2+\frac{(\nvote-2)(1-\score_2)}{2} \geq 0\end{align*} 
    since $\nvote \geq \frac{2\score_2}{1-\score_2}$. This implies that for all $\cand{} \in \cands\setminus\{\bcand{\ell}\}$, we have
    $\bs'(\bcand{\ell}) \geq \bs'(\cand{})$ and $\bcand{\ell} \tieorder \cand{}$, implying $\pos_\tiebreaker\maniprofile{1} = \bcand{\ell} \vote{1} \bcand{\ell+1} = \pos_\tiebreaker (\kprofile{k}{1})$. Hence, $\pos$ cannot be $\ksp{k}$ for any $k <  \lowerpos{\ell}$, implying $\spdeg{\pos} \geq \lowerpos{\ell}$. Since this is true for all $\ell \in [\ncand-1]$, we get the theorem statement.

\end{proof}

\subsection{Proof of \Cref{thm:black}}\label{app:black}
\blackthm*

\begin{proof}
    The upper bound follows from \Cref{remark:majority-consistent}. For the lower bound, fix any $0 \leq k < \nvote -1$. We will show that $\black$ is not $\ksp{k}$. Consider the following profile:
        \begin{center}
        \begin{tabular}{c|l}
            \cellcola Voter 1 & \cellcola $\cand{1} \vote{1} \cand{2} \vote{1} \cand{3} \vote{1} \cdots$
            \\ \cellcolb $\frac{\nvote-1}{2}$ voters (Group 2) & \cellcolb $\cand{2} \vote{i} \cand{3} \vote{i} \cdots \vote{i} \cand{1}$
            \\\cellcola $\frac{\nvote-1}{2}$ voters (Group 3) &  \cellcola $\cand{3} \vote{i} \cdots \vote{i} \cand{1} \vote{i} \cand{2}$
        \end{tabular} 
    \end{center}
    Say the tiebreaker $\tieorder$ ranks $\cand{3}$ top. Then we have
    \begin{align*}
        \margin{\kprofile{k}{1}}[\cand{1},\cand{2}] &= \margin{\kprofile{k}{1}}[\cand{2},\cand{3}] = k+1 >0\text{, and}\\
        \margin{\kprofile{k}{1}}[\cand{3},\cand{1}] &= (\nvote-1)-(k+1) \geq 0,
    \end{align*}
and any other $\cand{}  \in \cands \setminus \{\cand{1},\cand{2},\cand{3}\}$ has $\margin{\kprofile{k}{1}}[\cand{3},\cand{}]=\nvote+k$. Hence, there is no Condorcet winner in $\kprofile{k}{1}$, as any candidate has at least one other candidate it does not defeat. For each candidate $\cand{} \in \cands$, say $\bs_k(\cand{})$ is the total Borda score of $\cand{}$ in $\kprofile{k}{1}$. Then we have 
\begin{align*}
    \bs_k(\cand{3}) - \bs_k(\cand{1}) &= (\nvote-1)(\ncand-2)-2(k+1) \geq 0,\\
    \bs_k(\cand{3}) - \bs_k(\cand{2}) &= \frac{\nvote-1}{2}(\ncand-2)-(k+1) \geq 0,
\end{align*}
and $\bs_k(\cand{3})- \bs_k(\cand{})>0$ for any other $\cand{}$. Since $\tieorder$ ranks $\cand{3}$ top, this implies that $\black_\tiebreaker(\kprofile{k}{1})=\borda_\tiebreaker(\kprofile{k}{1})=\cand{3}$. On the other hand, if Voter 1 reports an alternative preference order $\manip{1}$ ranking $\cand{2}$ top, she will make it the Condorcet winner in $\maniprofile{1}$. Hence, we have $\black_\tiebreaker\maniprofile{1}= \cand{2} \vote{1} \cand{3} = \black_\tiebreaker(\kprofile{k}{1})$, completing the proof.
\end{proof}

\subsection{Proof of \Cref{thm:maximin}}\label{app:maximin}

\maximinthm*

For presentation purposes, we will separate the theorem into two propositions (one for the upper bound and one for the lower bound) and prove them separately.

\begin{proposition}
    The manipulation leverage of Maximin is at most $\frac{(\ncand-2)\nvote}{\ncand-1}+2$.
\end{proposition}

\begin{proof}
    Fix any profile $\profile$, tiebreaker $\tieorder$, manipulation $\manip{1}$, and integer $k\geq \frac{(\ncand-2)\nvote}{\ncand-1}+1$. For each  distinct $\cand{}, \cand{}'\in \cands$, write $\nbetter{\cand{}}{\cand{}'} \coloneq |\{i \in \voters: \cand{} \vote{i} \cand{}' \}|$ and  $\mscore{\cand{}}  \coloneq \min_{\cand{}' \in \cands \setminus \{\cand{}\} } \nbetter{\cand{}}{\cand{}'}$. Thus the maximin winners are $\maximin(\profile)=\arg \max_{\cand{} \in \cands} \mscore{\cand{}}$, since $\margin{\profile}[\cand{},\cand{}']=2\nbetter{\cand{}}{\cand{}'}-\nvote$ for any distinct $\cand{},\cand{}'\in \cands$. Define $\nbetterk{\cand{}}{\cand{}'}$ and $\mscorek{\cand{}}$ analogously for the profile $\kprofile{k}{1}$, and define   $\nbetter{\cand{}}{\cand{}'}'$ and $\mscore{\cand{}}'$ analogously for the profile $\maniprofile{1}$. Say $\maximin_\tiebreaker(\profile)=\cand{x}$ and $\maximin_\tiebreaker\maniprofile{1}=\cand{y}$. For any $\cand{} \in \cands$ such that $\cand{y} \vote{1} \cand{}$, we have
    \begin{equation}  \nbetter{\cand{y}}{\cand{}} \geq \nbetter{\cand{y}}{\cand{}}' \geq \mscore{\cand{y}}' \geq \mscore{\cand{x}}' \geq \mscore{\cand{x}} -1, \label{eq:xvsy}\end{equation}
    where the first inequality holds because voter 1 already ranks $\cand{y}$ above $\cand{}$, so her manipulation cannot increase the number of votes supporting $\cand{y}$ over $\cand{}$. The third inequality holds because $\cand{y}$ wins after manipulation, and the last because changing one vote decreases $\mscore{\cand{x}}$ for any candidate by at most one. Construct a simple cycle $(\cand{(1)},\cand{(2)},\ldots, \cand{(\ell)}, \cand{(1)})$ such that $\cand{(i+1)} \in \argmin_{\cand{} \in \cands \setminus\{\cand{(i)}\}} \nbetter{\cand{(i)}}{\cand{}}$ for each $i \in [\ell]$ (where we interpret $\ell+1=1$) and so $\mscore{\cand{(i)}} = \nbetter{\cand{(i)}}{\cand{(i+1)}}$ for each $i\in [\ell]$. Note that such a cycle must exist, since each $\cand{} \in \cands$ has a nonempty $\argmin_{\cand{}' \in \cands \setminus\{\cand{}\}} \nbetter{\cand{}}{\cand{}'}$, and there are finitely many candidates. Each voter must rank $\cand{(i)} \succ \cand{(i+1)}$ for at least one $i \in [\ell]$, since voter preferences are acyclic (so disagreeing with every edge in the cycle leads to a contradiction). This implies
    \begin{equation}
        \sum_{i=1}^{\ell} \mscore{\cand{(i)}}= \sum_{i=1}^{\ell} \nbetter{\cand{(i)}}{\cand{(i+1)}}  \geq \nvote. \label{eq:sumncc}
    \end{equation}
    Lastly, label the top candidate of the truthful vote of voter 1 as  $\cand{z} \coloneq \cand{1}$ (to avoid confusion with $\cand{(1)}$). We now prove that $\maximin$ is $\ksp{k}$ under two different cases.
    
    \noindent\textbf{Case 1:} $\ell \leq \ncand-1$. By \eqref{eq:sumncc}, there exists some $i^* \in [\ell]$ such that $\mscore{\cand{(i^*)}}\geq \frac{\nvote}{\ell} \geq \frac{\nvote}{\ncand-1}$. Since $\cand{x} \in \maximin(\profile)$,  we also have $\mscore{\cand{x}} \geq \mscore{\cand{(i^*)}} \geq \frac{\nvote}{\ncand-1}$. This implies that for any $\cand{} \in \cands$ such that $\cand{y} \vote{1} \cand{}$ we have
    \begin{equation}
      \mscorek{\cand{}} \leq  \nbetterk{\cand{}}{\cand{y}} = \nbetter{\cand{}}{\cand{y}}  =  \nvote-\nbetter{\cand{y}}{\cand{}}  \leq \nvote-\mscore{\cand{x}} +1  \leq \frac{(\ncand-2)\nvote}{\ncand-1} +1  ,\label{eq:scstar}
    \end{equation}
    where the penultimate inequality follows from \eqref{eq:xvsy}.
    As $\nbetterk{\cand{z}}{\cand{}'} = \nbetter{\cand{z}}{\cand{}'}+ k$ for any $\cand{}' \in \cands \setminus \{\cand{z}\}$, we have
    \begin{equation}
        \mscorek{\cand{z}} = \mscore{\cand{z}} +k\geq 1+ k  \geq  \frac{(\ncand-2)\nvote}{\ncand-1}+2 \label{eq:sz},
    \end{equation}
    where $\mscore{\cand{z}} \geq 1$ since at least voter 1 ranks $\cand{z}$ above every other candidate in $\profile$. Comparing \eqref{eq:scstar} with \eqref{eq:sz}, we see that for any $\cand{} \in \cands$ such that $\cand{y} \vote{1} \cand{}$, we have $\mscorek{\cand{z}} > \mscorek{\cand{}}$ and therefore $\cand{} \notin \maximin(\kprofile{k}{1})$. This ensures that $\maximin_\tiebreaker(\kprofile{k}{1}) \voteq{1} \cand{y} = \maximin_\tiebreaker\maniprofile{1}$, as desired.
    
    \noindent\textbf{Case 2:} $\ell = \ncand$. This implies that all candidates are in the cycle constructed above, implying by \eqref{eq:sumncc} that $\sum_{\cand{} \in \cands \setminus \{\cand{z}\}} \mscore{\cand{}} \geq \nvote -\mscore{\cand{z}}$. Since $\mscore{\cand{x}} \geq \mscore{\cand{}}$ for all $\cand{} \in \cands$, this gives us $\mscore{\cand{x}} \geq \frac{\nvote-\mscore{\cand{z}}}{\ncand-1}$. This implies that for any $\cand{} \in \cands$ such that $\cand{y} \vote{1} \cand{}$, we have
    \begin{equation}
      \mscorek{\cand{}} \leq  \nbetterk{\cand{}}{\cand{y}} = \nbetter{\cand{}}{\cand{y}} =  \nvote- \nbetter{\cand{y}}{\cand{}}  \leq \nvote-\mscore{\cand{x}} +1  \leq \frac{(\ncand-2)\nvote}{\ncand-1} +\frac{\mscore{\cand{z}}}{\ncand-1}+1 \label{eq:sc2},
    \end{equation}
    where the penultimate inequality follows from \eqref{eq:xvsy}. Since $\nbetterk{\cand{z}}{\cand{}'} = \nbetter{\cand{z}}{\cand{}'}+ k$ for any $\cand{}' \in \cands \setminus \{\cand{z}\}$, we have
    \[\mscorek{\cand{z}} = \mscore{\cand{z}}  +k \geq \mscore{\cand{z}} +\frac{(\ncand-2)\nvote}{\ncand-1}+1> \frac{\mscore{\cand{z}}}{\ncand-1}+\frac{(\ncand-2)\nvote}{\ncand-1}+1, \]
    where the last inequality follows from the fact that $\mscore{\cand{z}} \geq 1$. Comparing this to \eqref{eq:sc2}, we see that for any $\cand{} \in \cands$ such that $\cand{y} \vote{1} \cand{}$, we have $\mscorek{\cand{z}} > \mscorek{\cand{}} $ and therefore $\cand{} \notin \maximin(\kprofile{k}{1})$. Once again, this ensures that $\maximin_\tiebreaker(\kprofile{k}{1}) \voteq{1} \cand{y} = \maximin_\tiebreaker\maniprofile{1}$. Thus, $\spdeg{\maximin}\leq\left\lceil\frac{(\ncand-2)\nvote}{\ncand-1}+1\right\rceil\leq\frac{(\ncand-2)\nvote}{\ncand-1}+2$, completing the proof.

\end{proof}

\begin{proposition}
    The manipulation leverage of Maximin ($\maximin$) is at least $\frac{(\ncand-2)(\nvote-2)}{\ncand-1}+1$, assuming that $\nvote-2$ is divisible by $\ncand-1$ and that $\ncand \geq 4$. Further, $\maximin$ is not $\ksp{k}$ for any $k<\frac{(\ncand-2)(\nvote-2)}{\ncand-1}+1$.
\end{proposition}
\begin{proof}
    Fix some $0 \leq k  < \frac{(\ncand-2)(\nvote-2)}{\ncand-1}+1$. We will show that $\maximin$ is not $\ksp{k}$. To do so, separate the candidates into three groups $\{\acand{1},\acand{2},\aldots,\acand{\ncand-2}\}$,  $\{\bcand{\ncand-1}\}$, and $\{\ccand{\ncand}\}$ and consider the following profile (note that voter 1 has a different ranking than usual, for presentation purposes):
    \begin{center}
        \begin{tabular}{c|l}
            \cellcola Voter 1 & \cellcola $\ccand{\ncand} \vote{1} \acand{1} \vote{1} \bcand{\ncand-1} \vote{1}  \acand{2} \vote{1} \acand{3} \vote{1} \acdots \vote{1} \acand{\ncand-2}$
            \\ \cellcolb Voter 2 & \cellcolb $\acand{\ncand-2} \vote{2} \acand{\ncand-3} \vote{2} \acdots \vote{2} \acand{2} \vote{2} \bcand{\ncand-1} \vote{2} \acand{1} \vote{2} \ccand{\ncand} $
            \\\cellcola $\frac{\nvote-2}{\ncand-1}$ voters  &  \cellcola $\acand{1}\vote{i} \acand{2} \vote{i} \acdots \vote{i} \acand{\ncand-2} \vote{i} \bcand{\ncand-1} \vote{i} \ccand{\ncand}$
            \\\cellcolb $\frac{\nvote-2}{\ncand-1}$ voters  &  \cellcolb $\bcand{\ncand-1} \vote{i} \acand{1}\vote{i} \acand{2} \vote{i} \acdots \vote{i} \acand{\ncand-2} \vote{i}  \ccand{\ncand}$
            \\\cellcola $\frac{\nvote-2}{\ncand-1}$ voters  &  \cellcola $\acand{\ncand-2}\vote{i} \bcand{\ncand-1} \vote{i} \acand{1}\vote{i} \acand{2} \vote{i} \acdots \vote{i} \acand{\ncand-3} \vote{i}  \ccand{\ncand}$
            \\\cellcolb $\frac{\nvote-2}{\ncand-1}$ voters  &  \cellcolb $\acand{\ncand-3} \vote{i} \acand{\ncand-2}\vote{i} \bcand{\ncand-1} \vote{i} \acand{1}\vote{i} \acand{2} \vote{i} \acdots \vote{i} \acand{\ncand-4} \vote{i}  \ccand{\ncand}$
            \\\cellcola \vdots  &  \multicolumn{1}{c}{\cellcola \vdots}
            \\\cellcolb $\frac{\nvote-2}{\ncand-1}$ voters  &  \cellcolb $\acand{3} \vote{i} \acand{4}\vote{i} \acdots \vote{i} \acand{\ncand-2}\vote{i} \bcand{\ncand-1} \vote{i} \acand{1} \vote{i} \acand{2} \vote{i} \ccand{\ncand}$
            \\\cellcola $\frac{\nvote-2}{\ncand-1}$ voters  &  \cellcola$\acand{2} \vote{i} \acand{3}\vote{i} \acand{4}\vote{i} \acdots \vote{i} \acand{\ncand-2}\vote{i} \bcand{\ncand-1} \vote{i} \ccand{\ncand} \vote{i} \acand{1}$
        \end{tabular} 
    \end{center}
    Say that the tiebreaker $\tieorder$ ranks $\cand{\ncand-1}$ top and $\cand{1}$ second. Write $\margink=\margin{\kprofile{k}{1}}$ for convenience. It can then be checked that
    \begin{itemize}
        \item $\forall i \in \{2,3,\ldots,\ncand-1\}: \margink[\cand{i}, \cand{\ncand}]=\nvote-2-k$
        \item $\margink[\cand{1}, \cand{\ncand}]=\frac{\ncand-3}{\ncand-1}(\nvote-2)-k$,
        \item $\forall 1 \leq i < j \leq \ncand-2: \margink[\cand{i}, \cand{j}]=\frac{(\ncand-1)-2(j-i)}{\ncand-1}(\nvote-2)+k$,
        \item $\forall 2 \leq i \leq \ncand-2: \margink[\cand{i}, \cand{\ncand-1}]=\frac{(\ncand-1)-2(\ncand-1-i)}{\ncand-1}(\nvote-2)-k$, and
        \item $\margink[\cand{1}, \cand{\ncand-1}]=\frac{(\ncand-1)-2(\ncand-2)}{\ncand-1}(\nvote-2)+k$.
    \end{itemize}
    Accordingly, we have
    \begin{itemize}
        \item $\candmarg{\cand{1}} \coloneq \min_{\cand{} \in \cands \setminus\{\cand{1}\}} \margink[\cand{1}, \cand{}]=\min\left(\frac{\ncand-3}{\ncand-1}(\nvote-2)-k,\, k- \frac{\ncand-3}{\ncand-1}(\nvote-2)\right)$,
        \item $\forall i \in \{2,3,\ldots,\ncand-2\}: \candmarg{\cand{i}} \coloneq \min_{\cand{} \in \cands \setminus\{\cand{i}\}} \margink[\cand{i}, \cand{}]=-\frac{\ncand-3}{\ncand-1}(\nvote-2)-k$, 
        \item $\candmarg{\cand{\ncand-1}} \coloneq \min_{\cand{} \in \cands \setminus\{\cand{\ncand-1}\}} \margink[\cand{\ncand-1}, \cand{}]=\min\left(\frac{\ncand-3}{\ncand-1}(\nvote-2)-k,\, k- \frac{\ncand-3}{\ncand-1}(\nvote-2)\right)$, and
        \item $\candmarg{\cand{\ncand}} \coloneq \min_{\cand{} \in \cands \setminus\{\cand{\ncand}\}} \margink[\cand{\ncand}, \cand{}]=k - (\nvote-2)$.
    \end{itemize}
    It is easy to check that $\candmarg{\cand{\ncand-1}} \geq \candmarg{\cand{i}}$ for any $1 \leq i \leq \ncand-2$. Further,
    \begin{align*}
        \candmarg{\cand{\ncand-1}}-\candmarg{\cand{\ncand}}= \min\left(\frac{2\ncand-4}{\ncand-1}(\nvote-2)-2k,\, \frac{2(\nvote-2)}{\ncand-1}\right)\geq \min\left(0,\, \frac{2(\nvote-2)}{\ncand-1}\right)\geq 0.
    \end{align*}
    Since $\tieorder$ ranks $\cand{\ncand-1}$ first, this implies that $\maximin_\tiebreaker(\kprofile{k}{1})=\cand{\ncand-1}$. Now, go back to the original profile and say voter 1 reports an alternative preference order $\manip{1}$ that ranks $\ccand{\ncand} \manip{1} \acand{1}  \manip{1}  \acand{2} \manip{1} \acand{3} \manip{1} \acdots \manip{1} \acand{\ncand-2} \manip{1} \bcand{\ncand-1}$. Now, we have:
    \begin{align*}
        \min_{\cand{} \in \cands\setminus\{\cand{\ncand-1}\}} \margin{\maniprofilenp{1}}[\cand{\ncand-1}, \cand{}] \leq  &\margin{\maniprofilenp{1}}[\cand{\ncand-1}, \cand{\ncand-2}]=-\frac{\ncand-3}{\ncand-1}(\nvote-2)-2,\\
         \min_{\cand{} \in \cands\setminus\{\cand{\ncand}\}} &\margin{\maniprofilenp{1}}[\cand{\ncand}, \cand{}] =-(\nvote-2)\text{, and}\\
         \forall i \in [\ncand-2]: \min_{\cand{} \in \cands\setminus\{\cand{i}\}} &\margin{\maniprofilenp{1}}[\cand{i}, \cand{}] =-\frac{\ncand-3}{\ncand-1}(\nvote-2),
    \end{align*}
    implying $\cand{1} \in \maximin\maniprofile{1}$ and $\cand{\ncand-1} \notin \maximin\maniprofile{1}$. Since $\cand{1}$ is ranked second by $\tieorder$ after $\cand{\ncand-1}$, this implies $\maximin_\tiebreaker\maniprofile{1}=\cand{1} \vote{1} \cand{\ncand-1} = \maximin_\tiebreaker(\kprofile{k}{1})$, completing the proof.
\end{proof}

\subsection{Adapting \Cref{thm:biranking} to even $\nvote$}\label{app:biranking-even}

In \Cref{sec:biranking-majoritarian} we proved the following general lower bound on the manipulation leverage of biranking-majority-consistent rules for odd $\nvote$.

\birankingodd*

We now discuss two ways of adapting the proof of \Cref{thm:biranking} to even $\nvote$, as discussed in \Cref{remark:odd-n}. First, if strengthen biranking majority consistency significantly to rules  that pick the unique weak Condorcet winner whenever one exists, we can get an identical lower bound, up to an additive constant.

\begin{proposition}\label{prop:weakcondo-even}
        Say $\scf$ is an anonymous SCF that picks the unique weak Condorcet winner whenever one exists. Then, $\spdeg{\scf} \geq \frac{\nvote-4}{2}$.
\end{proposition}
\begin{proof}
    Odd $\nvote$ is covered by \Cref{thm:biranking}, since the top choice of the majority ranking in a profile consisting of only two rankings is a Condorcet (and thus a weak Condorcet) winner, implying $\scf$ is biranking-majority consistent. Similarly, the claim for $\nvote=2$ and $\nvote=4$ follows from the Gibbard-Satterthwaite theorem, since $\scf$ is nondictatorial (by anonymity) and onto (since any candidate can be a unique weak Condorcet winner), implying $\spdeg{\scf} > 0$. For even $\nvote \geq 6$, consider a profile with $\frac{3(\nvote-2)}{2}$ voters divided into three groups:
    \begin{center}
        \begin{tabular}{c|l}
            \multicolumn{2}{c}{\cellcolor{green!10} Profile 1}\\
            \cellcola $\frac{\nvote-2}{2}$ voters (Group 1) & \cellcola $\bcand{1} \vote{i} \bcand{2} \vote{i} \bcand{3} \vote{i} \acand{4} \vote{i} \acdots \vote{i} \acand{\ncand}$
            \\ \cellcolb $\frac{\nvote-2}{2}$ voters (Group 2) & \cellcolb $\bcand{2} \vote{i} \bcand{3} \vote{i} \bcand{1} \vote{i} \acand{4} \vote{i} \acdots \vote{i} \acand{\ncand}$
            \\\cellcola $\frac{\nvote-2}{2}$ voters (Group 3) &  \cellcola $\bcand{3} \vote{i} \bcand{1} \vote{i} \bcand{2} \vote{i} \acand{4} \vote{i} \acdots \vote{i} \acand{\ncand}$
        \end{tabular} 
    \end{center}
    Since $\scf$ will output only one winner, at least one out of the three groups of voters will get neither their first nor their second choice as the winner. WLOG, say this is Group 1, \ie, the winner that $\scf$ outputs on Profile 1 is not in $\{\cand{1},\cand{2}\}$. Then remove $k=\frac{\nvote-6}{2}$ copies of Group 1, in order to get a second profile:
    \begin{center}
        \begin{tabular}{c|l}
            \multicolumn{2}{c}{\cellcolor{green!10} Profile 2}\\
            \cellcola Voter 1 & \cellcola $\bcand{1} \vote{1} \bcand{2} \vote{1} \bcand{3} \vote{1} \acand{4} \vote{1} \acdots \vote{1} \acand{\ncand}$
            \\
            \cellcolb Voter 2 & \cellcolb $\bcand{1} \vote{2} \bcand{2} \vote{2} \bcand{3} \vote{2} \acand{4} \vote{2} \acdots \vote{2} \acand{\ncand}$
            \\ \cellcola $\frac{\nvote-2}{2}$ voters (Group 2) & \cellcola $\bcand{2} \vote{i} \bcand{3} \vote{i} \bcand{1} \vote{i} \acand{4} \vote{i} \acdots \vote{i} \acand{\ncand}$
            \\\cellcolb $\frac{\nvote-2}{2}$ voters (Group 3) &  \cellcolb $\bcand{3} \vote{i} \bcand{1} \vote{i} \bcand{2} \vote{i} \acand{4} \vote{i} \acdots \vote{i} \acand{\ncand}$
        \end{tabular} 
    \end{center}
    If we label Profile 2 as $\profile$, then Profile 1 is just $\kprofile{k}{1}$. Consider a manipulation in $\profile$ (Profile 2) where voter 1 also reports the ranking of Group 2. Then, $\margin{\maniprofilenp{1}}[\bcand{2},\bcand{1}]=0$, $\margin{\maniprofilenp{1}}[\bcand{2},\bcand{3}]=2$ and $\margin{\maniprofilenp{1}}[\bcand{2},\acand{\ell}]=\nvote$ for all $\ell \in \{4,\ldots,\ncand\}$, making $\cand{2}$ a weak Condorcet winner. Further, it is the unique weak Condorcet winner: $\cand{1}$ is pairwise defeated by $\cand{3}$, and every other candidate is pairwise defeated by $\cand{2}$. Since $\scf$ elects the unique weak Condorcet winner, we have $\scf\maniprofile{1}= \cand{2} \vote{1} \scf(\kprofile{k}{1})$, showing that $\scf$ is not $\ksp{\left(\frac{\nvote-6}{2}\right)}$. Thus, $ \spdeg{\scf} \geq \frac{\nvote-6}{2}+1 = \frac{\nvote-4}{2}$. 
\end{proof}

Second, we show that \Cref{prop:weakcondo-even} can be strengthened to require only Condorcet consistency, at the expense of weakening the lower bound on the manipulation leverage from $\approx \frac{\nvote}{2}$ to $\approx \frac{\nvote}{3}$.

\begin{proposition}
    For any anonymous Condorcet extension $\scf$, we have $\spdeg{\scf} \geq 2 \left\lfloor\frac{\nvote}{6}\right\rfloor > \frac{\nvote-6}{3}$.
\end{proposition}

\begin{proof}
    Odd $\nvote$ is covered by \Cref{thm:biranking}, since Condorcet consistency implies biranking majority consistency. Similarly, the claim for $\nvote=2$ and $\nvote=4$ follows from the Gibbard-Satterthwaite theorem, since $\scf$ is nondictatorial (by anonymity) and onto (since any candidate can be a Condorcet winner), implying $\spdeg{\scf} > 0$. For even $\nvote \geq 6$, let $q= \left\lfloor\frac{\nvote}{6}\right\rfloor$ and consider the following profile with $\nvote$ voters:

    \begin{center}
        \begin{tabular}{c|l}
            \multicolumn{2}{c}{\cellcolor{green!10} Profile 1}\\
            \cellcola $2 q$ voters (Group 1) & \cellcola $\bcand{1} \vote{i} \bcand{2} \vote{i} \bcand{3} \vote{i} \acand{4} \vote{i} \acdots \vote{i} \acand{\ncand}$
            \\ \cellcolb $2 q$ voters (Group 2) & \cellcolb $\bcand{2} \vote{i} \bcand{3} \vote{i} \bcand{1} \vote{i} \acand{4} \vote{i} \acdots \vote{i} \acand{\ncand}$
            \\\cellcola $2 q$ voters (Group 3) &  \cellcola $\bcand{3} \vote{i} \bcand{1} \vote{i} \bcand{2} \vote{i} \acand{4} \vote{i} \acdots \vote{i} \acand{\ncand}$
            \\ \cellcolb $\frac{\nvote -6q}{2}$ voters (Group 4A) & \cellcolb $\cand{1} \vote{i} \cand{2} \vote{i} \cdots  \vote{i} \cand{\ncand-1}\vote{i} \cand{\ncand}$
            \\\cellcola $\frac{\nvote -6q}{2}$ voters (Group 4B) &  \cellcola $\cand{\ncand} \vote{i} \cand{\ncand-1} \vote{i} \cdots \vote{i} \cand{2}\vote{i} \cand{1}$
        \end{tabular} 
    \end{center}
    Here, the voters in Group 4A and 4B are simply to complete the first three groups to $\nvote$ rankings without changing the margins (as they cancel each other out); thus, they can be replaced with $\frac{\nvote-6q}{2}$ pairs of any ranking and its reverse. 
    
    Since $\scf$ will output only one winner, at least one out of the three groups of voters will get neither their first nor their second choice as the winner. WLOG, say this is Group 2, so the output of $\scf$ on Profile 1 is some $\cand{}^* \in \cands \setminus \{\bcand{2},\bcand{3}\}$. Let voter 1 be any voter in Group 1 and let her manipulated vote $\manip{1}$ be her vote in Profile 1, \ie, $\bcand{1} \manip{1} \bcand{2} \manip{1} \bcand{3} \manip{1} \acand{4} \manip{1} \acdots \manip{1} \acand{\ncand}$. In other words, we label Profile 1 as $\maniprofile{1}$. Let the truthful vote of voter 1 have $\cand{}^* \vote{1} \bcand{3} \vote{1} \cand{}$ for all $\cand{} \in \cands \setminus \{\cand{}^* , \bcand{3}\}$ and fix $k= 2q-1$. In the augmented profile $\kprofile{k}{1}$, we have
    \begin{itemize}
        \item $\margin{\kprofile{k}{1}}[\bcand{3}, \bcand{2}] = -(2q-1)+(k+1)=1$, since $\cand{}^* \neq \bcand{2}$,
        \item $\margin{\kprofile{k}{1}}[\bcand{3}, \bcand{1}] \geq (2q+1)-(k+1)=1$, and
        \item for each $\ell \in \{4,\ldots, \ncand\}$, $\margin{\kprofile{k}{1}}[\bcand{3}, \acand{\ell}] \geq (6q-1) -(k+1) = 4q-1>0$ since $q \geq 1$.
    \end{itemize}
    Hence, $\bcand{3}$ is a Condorcet winner, implying $\scf(\kprofile{k}{1})=\bcand{3}$. Since $\scf\maniprofile{1} = \cand{}^* \vote{1} \bcand{3}= \scf(\kprofile{k}{1})$, this shows that $\scf$ is not $\ksp{k}$, proving $\spdeg{\scf} \geq k+1 = 2q = 2 \left\lfloor \frac{\nvote}{6}\right\rfloor$.

\end{proof}

\section{Group Manipulability}\label{app:group}

In this section, we introduce a straightforward adaptation of $k$-Augmentation Strategyproofness (\Cref{def:kmanip}) and the manipulation leverage (\Cref{def:manipot}) to the group manipulation setting, as discussed in \Cref{sec:future}(f).

\begin{definition}[$k$-Augmentation Group-Strategyproofness]\label{def:kgmanip}
   Fix an electorate size $\nvote$ and a non-negative integer $k$. A social choice function $\scf$ is \emph{$k$-augmentation group-strategyproof ($\kgsp{k})$ on $\nvote$-voter profiles} if for any profile $\profile \in \linprefs^\nvote$, any nonempty subset of voters $S \subseteq \voters$, and any set of alternative preference orders $\manip{S}=(\manip{i})_{i \in S}$, there exists an $i \in S$ such that $\scf( \kprofile{k}{S}) \voteq{i} \scf\maniprofile{S}$.
\end{definition}

Here, $\vote{S}$ and $\vote{-S}$ for a set of voters $S$ are defined analogously to $\vote{i}$ and $\vote{-i}$ for a given voter $i$. In particular, $\kprofile{k}{S}$ consists of adding $k$ copies of \emph{each} voter in $S$ to the profile $\profile$. As such, the two profiles being compared have $\nvote+k|S|$ and $\nvote$ voters, respectively. When $\nvote$ is clear from context, we simply say that $\scf$ is $\kgsp{k}$. Using this definition, we now extend the manipulation leverage to group manipulability.

\begin{definition}[Group Manipulation Leverage]\label{def:gmanipot}
    For any SCF $\scf$ and electorate size $\nvote$, we define its \emph{group manipulation leverage on $\nvote$-voter profiles} as
    \[\spgdeg{\scf}\coloneq \min \{k \in \mathbb{Z}_{\geq 0}: \scf\text{ is }\kgsp{k'}\text{ on $\nvote$-voter profiles for all }k'\geq k\}\]
and $\spgdeg{\scf}\coloneq \infty$ if no such $k$ exists.
\end{definition}

Just like with manipulation leverage, for each electorate size $\nvote$ we define the group manipulation leverage of a multi-valued SCC $\scc$ as $\spgdeg{\scc} \coloneq \max_{\tieorder \in \linprefs} \spgdeg{\scc_\tiebreaker}$.

We now show that our results from the main body do not necessarily generalize to this setting. In particular, recall from \Cref{remark:majority-consistent} that any majority-consistent rule (and therefore any Condorcet extension) has a trivial upper bound of $\nvote-1$ on its manipulation leverage. However, as we show next, a Condorcet extension can have infinite group manipulation leverage even at a fixed electorate size! This is because adding an equal number of copies of a coalition of voters, even if that coalition agrees on a manipulation, may not necessarily lead to a Condorcet winner, even if we add arbitrarily many copies.

To show this, we focus on the Condorcet extension known as Copeland's method ($\copeland$). Under Copeland, each candidate gets one point for every other candidate it defeats in a pairwise match, and half a point for each candidate it ties with. $\copeland$ then outputs the candidate(s) with the highest score. More formally,
\[
\copeland(\profile)= \argmax_{\cand{} \in \cands}|\{ \cand{}' \in \cands \setminus \{\cand{}\}: \margin{\profile}[\cand{},\cand{}']>0\}| + \frac{1}{2}|\{ \cand{}' \in \cands \setminus \{\cand{}\}: \margin{\profile}[\cand{},\cand{}']=0\}|.
\]

\begin{theorem}\label{thm:copeland}
    For every $\ncand \geq 5$ and every odd electorate size $\nvote \geq 13$, Copeland has $\spgdeg{\copeland}=\infty$. In particular, for every such $\nvote$, Copeland is not $\kgsp{k}$ on $\nvote$-voter profiles for any $k \geq 0$.
\end{theorem}

\begin{proof}
    Fix $\ncand\geq5$ and an odd $\nvote\geq13$. Say $S=[2]$ is our coalition of voters, with preferences $\cand{3} \vote{1} \cand{2} \vote{1} \cand{1} \vote{1}\cand{4} \vote{1} \cand{5} \vote{1} \cdots $ for voter 1 and  $\cand{5} \vote{2} \cand{4} \vote{2} \cand{2} \vote{2}\cand{1} \vote{2} \cand{3} \vote{2} \cdots $ for voter 2, where both voters rank all remaining candidates after the first five. We construct the background profile $\vote{-S}$ with exactly $\nvote-2$ voters. First, we specify its margins among the first five candidates:
    \begin{align*}
 1&=\margin{\vote{-S}}[\cand{1},\cand{3}]
= \margin{\vote{-S}}[\cand{1},\cand{4}]= \margin{\vote{-S}}[\cand{1},\cand{5}]
= \margin{\vote{-S}}[\cand{2},\cand{1}]
= \margin{\vote{-S}}[\cand{2},\cand{5}]\\&
= \margin{\vote{-S}}[\cand{3},\cand{2}]= \margin{\vote{-S}}[\cand{3},\cand{4}]= \margin{\vote{-S}}[\cand{4},\cand{2}]= \margin{\vote{-S}}[\cand{5},\cand{3}]=
\margin{\vote{-S}}[\cand{5},\cand{4}].
\end{align*}
Additionally, set $\margin{\vote{-S}}[\cand{i},\cand{j}]=1$ for every $i\leq5<j$, and for every $5<i<j\leq\ncand$. These margins are realized by an eleven-voter profile: start with the ranking $\cand{1}\succ\cdots\succ\cand{\ncand}$ and add a two-voter McGarvey gadget~\citep{mcgarvey53theorem} for each ordered pair
\[
    (\cand{2},\cand{1}),\quad(\cand{3},\cand{2}),\quad(\cand{4},\cand{2}),\quad(\cand{5},\cand{3}),\quad(\cand{5},\cand{4}).
\]
For an ordered pair $(a,b)$, the gadget consists of the rankings $a\succ b\succ\rho$ and $\bar\rho\succ a\succ b$, where $\rho$ orders the remaining candidates and $\bar\rho$ is its reversal. It contributes a margin of $2$ in favor of $a$ over $b$ and zero to every other pair. Thus, each gadget changes its designated margin from $-1$ to $1$. Finally, add $(\nvote-13)/2$ pairs of a ranking and its reverse. These pairs contribute zero to every margin, giving the required background profile with exactly $\nvote-2$ voters. 

Fix any $k \geq 0$ and notice that $\kprofile{k}{S}$ has the same margins among candidates $\{\cand{1},\ldots,\cand{5}\}$ as $\vote{-S}$, except $\margin{\kprofile{k}{S}}[\cand{2},\cand{1}]=2k+3$ instead (This is because each additional copy of voter 1 and voter 2 cancel each other out in every pairwise match among these five candidates, except this pair, on which the two voters agree). We cannot have $\cand{i} \in \copeland(\kprofile{k}{S})$ for any $i> 5$, since it can defeat at most $\ncand-6$ candidates, whereas any $\cand{j}$ for $j \leq 5$ defeats at least $\ncand-5$ candidates. Out of the five relevant candidates, the Copeland scores of $\cand{1}$ to $\cand{5}$ (minus the $\ncand-5$ candidates they all defeat)  are 3, 2, 2, 1, and 2, respectively. Hence, $\copeland_\tiebreaker(\kprofile{k}{S})=\cand{1}$. 

Going back to $\profile=(\vote{S},\vote{-S})$, suppose voters 1 and 2 instead report alternative preference orders where they both rank $\cand{2}$ top. This group manipulation results in $\margin{\maniprofile{S}}[\cand{2},\cand{3}]=\margin{\maniprofile{S}}[\cand{2},\cand{4}]=1$ as well, making $\cand{2}$ a Condorcet winner. Thus, $\copeland_\tiebreaker\maniprofile{S}=\cand{2}$. Since $\cand{2} \vote{i} \cand{1}$ for all $i \in S=[2]$, this shows that Copeland is not $\kgsp{k}$ on $\nvote$-voter profiles for any $k\geq 0$. Since $\nvote$ was an arbitrary odd integer at least $13$, this proves the theorem. Notice that we did not need to use the tiebreaker $\tieorder$, since both $\kprofile{k}{S}$ and $\maniprofile{S}$ had clear winners.
\end{proof}

It is worth noting that the proof of \Cref{thm:copeland} does rely on \Cref{def:kgmanip,def:gmanipot} requiring that each member of the coalition be copied an equal number of times. Alternative extensions (where the new copies need not be distributed uniformly among the coalition members) may give rise to different results, and studying these is a natural and important future direction.

\section{Pareto and Omninomination}\label{app:omnipot}

In this section, we analyze two social choice correspondences that achieve weaker forms of strategyproofness.\footnote{See \citet{brandt2011necessary} for more details on the strategic properties of these rules.} Given a profile $\profile$, we say a candidate $\cand{}$ \emph{Pareto dominates} another candidate $\cand{}'$ if $\cand{} \succ_i \cand{}'$ for all voters $i \in \voters$. The Pareto rule ($\pareto$) returns all candidates that are not Pareto dominated. The Omninomination rule ($\omni$), on the other hand, returns all candidates that are the top choice of \emph{some} voter.
We prove that both of these rules have unbounded manipulation leverage.

\begin{proposition}\label{prop:paromni}
    The manipulation leverage of Pareto and the manipulation leverage of Omninomination are both $\infty$. In particular, neither rule is $\ksp{k}$ for any $k \geq 0$.
\end{proposition}

\begin{proof}
     We prove the claim by constructing an instance in which a voter admits a manipulation that strictly improves the outcome over truthful reporting, while adding any number of (truthful) copies of the voter leaves the outcome unchanged. Consider the following profile $\profile$: 
    \begin{center}
    \begin{tabular}{c|l}
            \cellcola Voter $1$ & \cellcola  $\acand{1} \vote{1} \acand{2} \vote{1} \acand{3}  \vote{1} \cand{4} \vote{1}\cdots \vote{1}\cand{\ncand}$\\
            
            \cellcolb $\nvote-1$ voters & \cellcolb $ \acand{3} \vote{i} \acand{1} \vote{i} \acand{2} \vote{i} \cand{4} \vote{i}\cdots \vote{i}\cand{\ncand}$
    \end{tabular}
    \end{center}
    Say the tiebreaker ranks $\cand{2} \tieorder \cand{3} \tieorder \cand{1} \tieorder \cdots$ (the rest of the order does not matter).
    
    Consider adding any $k\geq 0$ copies of voter 1 to $\profile$, and observe that the election outcomes in $\kprofile{k}{1}$ are as follows.
    Under Pareto, the only candidates that are not Pareto dominated (and hence belong to the winner set before tie-breaking) are $\cand{1}$ and $\cand{3}$ (since candidates $\{\cand{2},\cand{4},\ldots\cand{\ncand}\}$ are Pareto dominated by candidate $\cand{1}$).
    After tie-breaking, the winner is therefore $\pareto_\tiebreaker(\kprofile{k}{1})=\cand{3}$. Under Omninomination, it is easy to see that both---and only---$\cand{1}$ and $\cand{3}$ are the top choice of some voter. The tiebreaker $\tieorder$ breaks the tie in favor of $\cand{3}$, so once again the winner is $\omni_\tiebreaker(\kprofile{k}{1})=\cand{3}$. Note that for both rules, adding more copies of voter 1 has not changed the outcome.

    Going back to the original profile $\profile$ (with no copies), assume voter 1 instead reports the alternative preference order
    $\bcand{2} \manip{1} \acand{1} \manip{1} \acand{3} \manip{1} \cand{4} \manip{1}\cdots \manip{1}\cand{\ncand}$.
    Now, under Pareto, we get a tie among $\cand{1}, \cand{2},$ and $\cand{3}$, which is broken in favor of $\cand{2}$ by $\tieorder$.
    Under Omninomination, the new outcome is a tie between $\cand{2}$ and $\cand{3}$, which is again broken in favor of $\cand{2}$. Thus, we have 
    \[\pareto_\tiebreaker\maniprofile{1}=\omni_\tiebreaker\maniprofile{1}=\cand{2}\vote{1} \cand{3} =\pareto_\tiebreaker(\kprofile{k}{1})=\omni_\tiebreaker(\kprofile{k}{1}),\]
    showing that neither rule is $\ksp{k}$ for any $k \geq 0$.
\end{proof}

\end{document}